\documentclass{article}

\usepackage[main, final]{neurips_2026}

\usepackage[utf8]{inputenc}
\DeclareUnicodeCharacter{FF0C}{,}
\usepackage[T1]{fontenc}

\usepackage{graphicx}
\usepackage{amsmath}
\usepackage{amsthm}
\usepackage{hyperref}
\usepackage{url}
\usepackage{booktabs}
\usepackage{array}
\usepackage{amsfonts}
\usepackage{bm}
\usepackage{nicefrac}
\usepackage{microtype}
\usepackage{xcolor}
\usepackage{enumitem}

\graphicspath{{./}}
\hypersetup{
  hidelinks,
  pdftitle={Planning as Dynamics Relaxation: Hippocampal Recurrent Network Realizes Optimal Goal-Directed Navigation},
  pdfauthor={Yuhang He, Junfeng Zuo, Tianhao Chu, Si Wu}
}

\newtheorem{theorem}{Theorem}
\newtheorem{corollary}{Corollary}
\newtheorem{proposition}{Proposition}

\makeatletter
\let\papersection\section
\let\papersubsection\subsection
\newcommand{\AppendixContents}{%
  \papersection*{Appendix Contents}
  \begingroup
  \setlength{\parskip}{0pt}%
  \@starttoc{aptoc}%
  \endgroup
}
\newcommand{\EnableAppendixContents}{%
  \renewcommand{\section}[1]{%
    \papersection{##1}%
    \addcontentsline{aptoc}{section}{\protect\numberline{\thesection}##1}%
  }%
  \renewcommand{\subsection}[1]{%
    \papersubsection{##1}%
    \addcontentsline{aptoc}{subsection}{\protect\numberline{\thesubsection}##1}%
  }%
}
\newcommand{\DisableAppendixContents}{%
  \let\section\papersection
  \let\subsection\papersubsection
}
\makeatother

\title{Planning as Dynamics Relaxation: Hippocampal
Recurrent Network Realizes Optimal Goal-Directed Navigation}

\author{%
  Yuhang He\textsuperscript{1,2,3,4,5}\quad
  Junfeng Zuo\textsuperscript{1,2,3,4,5}\quad
  Tianhao Chu\textsuperscript{1,3}\quad
  Si Wu\textsuperscript{1,2,3,4,5}\thanks{Corresponding author: \texttt{siwu@pku.edu.cn}.}\\
  \textsuperscript{1}School of Psychological and Cognitive Sciences, Peking University\\
  \textsuperscript{2}IDG/McGovern Institute for Brain Research, Peking University\\
  \textsuperscript{3}Peking-Tsinghua Center for Life Sciences,\\
  Academy for Advanced Interdisciplinary Studies, Peking University\\
  \textsuperscript{4}Center of Quantitative Biology, Peking University\\
  \textsuperscript{5}Key Laboratory of Machine Perception (Ministry of Education), Peking University
}

\begin{document}

\maketitle

\begin{abstract}
Neural correlates of the cognitive map are well documented, yet how a neural circuit turns a stored map into an obstacle-avoiding route remains unclear. Here we show that a hippocampal network with appropriate recurrent connections achieves optimal goal-directed navigation through its own relaxation dynamics. The recurrent weights represent the transition probabilities between the locations the neurons encode, walls are reflected by vanishing weights, and behavioral-timescale synaptic plasticity (BTSP) writes this connectivity during goal-free exploration. When a goal is recalled, a top-down current clamps the goal neuron and the population relaxes into a goal-locked activity field. We prove that this field is the desirability field of a linearly-solvable Markov decision process: the relaxed state is the exact optimum of a stated control problem, the KL-regularized objective of Todorov, and the optimal next-step distribution is carried by the branch synaptic currents. The identity extends to directed transitions, and a weak-cost limit returns the Green's function of the recurrent Laplacian. A place-tuned readout converts the field's log-gradient into a direction of travel, and the same synapses answer a predictive query when the source is placed at the current location. Simulations on mazes with complex obstacles show near-shortest routes, an identity that survives a thousand-fold range of recall amplitude, the loss of 80\% of place cells, and tens of thousands of states, low-rank updates of every plan after a local change, and a transient that is anytime: the local action is ready long before the field converges. We hope this offers a general circuit principle for planning in abstract cognitive maps beyond spatial navigation.
\end{abstract}

\section{Introduction}
\label{sec:intro}

An animal in a familiar environment can recall a destination and plan a route that avoids walls and takes available detours. The neural correlates of this cognitive map are well established: hippocampal place cells signaling location \citep{OKeefeDostrovsky1971}, entorhinal grid cells tiling the explored space \citep{Hafting2005}, and experience-dependent plasticity that binds them into a stored representation \citep{McHugh1996,Mehta2000,Lever2002,Bittner2017}. How these circuits produce an obstacle-avoiding route from the stored map is much less understood.

\textbf{Related work} (App.~\ref{app:related-work}). Predictive-map and successor-representation models propose that the hippocampus stores long-horizon transition statistics from which a value function for any goal can be derived \citep{Dayan1993,Stachenfeld2017,PirayDaw2021,Bono2023,Fang2023,George2023}, a view that extends to abstract task structures \citep{Behrens2018,Whittington2020,Whittington2022Review,George2021,Garvert2023,Courellis2024,ElGaby2024,Tacikowski2024,Bakermans2025}. Linearly-solvable Markov decision processes (LMDPs) reduce planning to a linear equation whose solution, the desirability field, gives the navigation direction through its log-gradient \citep{Kappen2005,Todorov2006,Todorov2009,Levine2018}, with harmonic and Laplacian methods as a related geometric picture \citep{ConnollyBurnsWeiss1990,MahadevanMaggioni2007,Machado2017}; other biological models simulate prospective trajectories during theta sequences and replay \citep{JohnsonRedish2007,PfeifferFoster2013,MattarDaw2018,Liu2019,Schwartenbeck2023,Jensen2024}. A lineage of recurrent models lets a network settle toward a plan: the cognitive graph of \citet{Muller1996}, wavefront propagation \citep{PonulakHopfield2013}, attractor preplay \citep{CorneilGerstner2015}, recurrent successor representations \citep{Fang2023}, Endotaxis \citep{Zhang2024Endotaxis}, and the spacetime attractor of \citet{Jensen2026}; their relaxed states support planning as an arrival time, a trajectory, or a predictive field. Whether the network that stores the map can itself relax into the optimal plan, with the optimum proven rather than approximated, has remained open.

Here we show that a hippocampal recurrent network performs goal-directed navigation through its own relaxation dynamics (Fig.~\ref{fig:concept}). The recurrent weights represent the transition probabilities between the locations the neurons encode, written during goal-free exploration by BTSP \citep{Bittner2017,Li2024}, with walls reflected by vanishing weights. When the animal recalls a destination, a top-down input clamps the corresponding neuron and the population relaxes into a goal-locked field, which we prove to be the desirability field of an LMDP: the relaxed state is the exact optimum of a stated control problem, and the optimal next-step distribution is carried by the circuit's own branch currents. A downstream place-tuned population combines the field with its local place kernel by dendritic multiplication \citep{MitchellSilver2003}, and its population vector points along the field's log-gradient toward the goal.

\textbf{Our specific contributions are as follows:}
\begin{itemize}[leftmargin=0.8em,topsep=2pt plus 1pt minus 1pt,itemsep=1pt plus 1pt minus 0.5pt,parsep=0pt plus 0.5pt,partopsep=0pt]
\item \textbf{LMDP equivalence.} The goal-clamped steady state equals the LMDP desirability field, with the optimal policy in the branch synaptic currents (Theorem~\ref{thm:equiv}); the identity extends to directed transitions (Corollary~\ref{cor:directed}) and, in a weak-cost limit, to the Green's function of the recurrent connectivity, a circuit basis for the spectral account of \citet{Zuo2026Laplacian} (Sec.~\ref{sec:exp-green}).
\item \textbf{Place-tuned readout.} A place-tuned population multiplies the field by its local place kernel; its population vector points toward the goal, is invariant to the recall amplitude, and reads the raw firing rate (Sec.~\ref{sec:readout}).
\item \textbf{Learning by BTSP.} BTSP writes the connectivity during goal-free exploration, and a local change updates every plan through a low-rank correction (Sec.~\ref{sec:learning}).
\item \textbf{Scope and comparisons.} Grid mazes, continuous arenas, directed maps, and maps with up to $38{,}556$ states confirm the identity to machine precision, show that the transient is anytime, and compare the circuit with learned goal-conditioned models at matched experience (Sec.~\ref{sec:experiments}).
\end{itemize}

\begin{figure}[t]
\centering
\includegraphics[width=\linewidth]{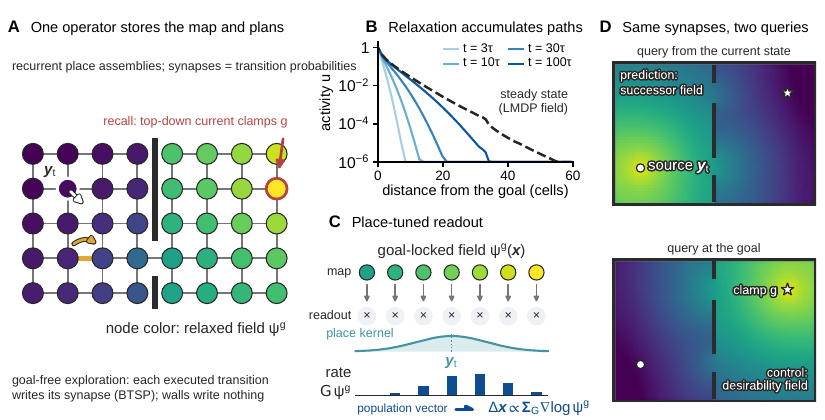}
\caption{\textbf{One recurrent operator stores the map and computes the plan.}
(\textbf{A}) Recurrent place assemblies; a synapse exists only where BTSP wrote an executed transition (gold), so walls write nothing. Recalling a goal clamps one assembly (red); the population relaxes into the goal-locked field $\boldsymbol{\psi}^g$ (node color, log scale), and a readout at the current position $\boldsymbol{y}_t$ points up its log-gradient (white arrow).
(\textbf{B}) The activity profile along the route grows out from the clamped goal and converges to the LMDP desirability field (two-room maze, $\alpha=0.985$).
(\textbf{C}) Place-tuned readout cells multiply the field by a local place kernel $\mathcal{G}$ centered at $\boldsymbol{y}_t$; their population vector is $\boldsymbol{\Sigma}_G\nabla\log\psi^g$ (Eq.~\ref{eq:popvec}).
(\textbf{D}) Two queries of one operator: a source at the current state yields the successor field (prediction), a clamp at the goal the desirability field (control).}
\label{fig:concept}
\end{figure}

\section{The model}
\label{sec:model}

\subsection{Recurrent network model}
\label{sec:network}

The hippocampus (e.g.\ CA3) is densely interconnected by recurrent collaterals \citep{AmaralWitter1989,Ishizuka1990,Sammons2024,Mishra2016,Li2024}, and after exploration its pyramidal cells form a place code \citep{OKeefeBurgess1996}. We model this population as $N$ recurrently connected units in continuous time; a unit is an assembly of cells that codes one location (Sec.~\ref{sec:scope}). Unit $i$ codes location $i$ with firing rate $r_i(t)$, and the recurrent synapse from unit $j$ to unit $i$ has weight $S_{ij}$, collected in $\boldsymbol{S}\in\mathbb{R}^{N\times N}$ (Sec.~\ref{sec:learning} describes how plasticity writes $\boldsymbol{S}$). Each unit receives a background $I^{\mathrm{bg}}_i$, an entorhinal input $I^{\mathrm{EC}}_i(\boldsymbol{y}_t)$ tuned to the current position $\boldsymbol{y}_t$ \citep{Hafting2005}, and a top-down goal current $I^{\mathrm{goal}}_i(g)\propto\delta_{ig}$ delivered to the goal unit during recall \citep{Ito2015}:
\begin{equation}
\tau_u\,\dot r_i
= -r_i
+ g_s\sum_j S_{ij}\,\phi(r_j)
- g_{\mathrm{inh}}\,\phi(r_i)
+ I^{\mathrm{bg}}_i
+ I^{\mathrm{EC}}_i(\boldsymbol{y}_t)
+ I^{\mathrm{goal}}_i(g),
\label{eq:rate}
\end{equation}
where $\phi$ is a smooth saturating activation function (we use $\tanh$), $\tau_u$ is the membrane time constant, and $g_s,g_{\mathrm{inh}}>0$ scale the recurrent excitation and the local inhibition. Without recall the population rests at $\boldsymbol{r}^0(\boldsymbol{y}_t)$, set by the background and entorhinal drives. Recalling goal $g$ clamps the goal unit at $r_g=r_g^0+\bar u_g$, with $\bar u_g>0$ the recall amplitude, and the remaining units evolve toward a goal-driven steady state. Writing $u_i:=r_i-r_i^0$ and $\delta I_i:=I^{\mathrm{goal}}_i(g)$ for the only input that changes, the expansion of $\phi$ around $r_i^0$ gives
\begin{equation}
\tau\,\dot{\boldsymbol{u}}
= -(\boldsymbol{I}_N-\alpha\,\boldsymbol{S})\,\boldsymbol{u} + \delta\boldsymbol{I} + \mathcal{O}(\|\boldsymbol{u}\|^2),
\qquad
\alpha := \frac{g_s\,\rho}{1+g_{\mathrm{inh}}\,\rho}\in(0,1),
\qquad
\tau := \frac{\tau_u}{1+g_{\mathrm{inh}}\,\rho},
\label{eq:linear}
\end{equation}
where $\rho:=\phi'(r_i^0)$ is the slope of the activation function at rest (App.~\ref{app:linearization-homogeneous-gain}). Local inhibition reduces both the effective gain $\alpha$ and the effective time constant $\tau$, and $\alpha<1$ keeps the network stable. A length-$L$ path through $\boldsymbol{S}$ contributes $\alpha^L=e^{-\kappa L}$ with $\kappa:=-\log\alpha$, so $1/\kappa$ is the number of synaptic steps activity travels before it decays.

\subsection{Goal-clamped relaxation produces the desirability field}
\label{sec:equivalence}

Relaxation carries its ordinary meaning in physics and numerical analysis: after the goal clamp perturbs the resting state, activity spreads outward through the recurrent connections, about one synaptic hop per time constant $\tau$, until the population settles into a new steady field (Fig.~\ref{fig:concept}B). With symmetric $\boldsymbol{S}$ the dynamics are a gradient flow on a convex energy, and one Euler step of Eq.~\eqref{eq:linear} is one Jacobi sweep of the equation the field solves. Planning is this settling, and Theorem~\ref{thm:equiv} characterizes its result.

With the goal unit clamped at $u_g=\bar u_g$, let $\mathcal{I}_g:=\{1,\dots,N\}\setminus\{g\}$ denote the non-goal units. Setting $\dot{\boldsymbol{u}}=0$ in Eq.~\eqref{eq:linear} on $\mathcal{I}_g$ and dropping the quadratic remainder gives
\begin{equation}
(\boldsymbol{I}_{\mathcal{I}_g} - \alpha\, \boldsymbol{S}_{\mathcal{I}_g\mathcal{I}_g})\, \boldsymbol{u}_{\mathcal{I}_g}
= \boldsymbol{b}_g,
\qquad
(\boldsymbol{b}_g)_i := \alpha\, S_{ig}\, \bar u_g
\;\;\text{for } i \in \mathcal{I}_g,
\label{eq:dirichlet}
\end{equation}
where $\boldsymbol{b}_g$ is the current the clamped goal sends into unit $i$ through the synapse $S_{ig}$. Since $\alpha<1$, the inverse expands as $\sum_{L\ge0}\alpha^L\boldsymbol{S}^L_{\mathcal{I}_g\mathcal{I}_g}$: $u_i$ sums every legal path from $i$ to the goal, weighted by the product of its synaptic weights and damped by $e^{-\kappa L}$, so shorter routes weigh more and walls, whose synapses never form, contribute nothing.

To interpret this field, we state a planning problem and match its solution to the relaxation. An animal at location $i$ chooses a distribution $\mu(\cdot\mid i)$ over its next location and pays a per-step cost with two parts: a fixed amount $q_0>0$ for the step, plus a deviation penalty $\lambda\,D_{\mathrm{KL}}\bigl(\mu(\cdot\mid i)\,\|\,P_0(\cdot\mid i)\bigr)$, where $P_0(\cdot\mid i)$ describes how the animal moves with no goal in mind; the plasticity that writes $\boldsymbol{S}$ also writes $\boldsymbol{P}_0$ as the spontaneous next-step distribution among the accessible neighbors of $i$ (Sec.~\ref{sec:btsp-rule}). ``Optimal'' throughout the paper means the exact optimum of this KL-regularized objective, the linearly-solvable class of \citet{Todorov2006}. With $V(i)$ the expected total future cost from $i$ under the best plan and $V(g)=0$, the substitution $z^g(i):=\exp(-V(i)/\lambda)$ turns the Bellman equation into a linear one \citep{Kappen2005,Todorov2006,Todorov2009}:
\begin{equation}
(\boldsymbol{I}_{\mathcal{I}_g} - \gamma\, \boldsymbol{P}_{0,\,\mathcal{I}_g\mathcal{I}_g})\, \boldsymbol{z}^g_{\mathcal{I}_g}
= \gamma\, \boldsymbol{p}_g,
\qquad
(\boldsymbol{p}_g)_i := P_0(g \mid i),
\qquad
\gamma := e^{-q_0/\lambda} \in (0,1),
\label{eq:lmdp}
\end{equation}
with $z^g(g)=1$, and the optimal next-step distribution is $\mu^*(j\mid i)\propto P_0(j\mid i)\,z^g(j)$. Eqs.~\eqref{eq:dirichlet} and~\eqref{eq:lmdp} have the same form: each holds the goal at a fixed value and spreads it through the population, through $\alpha\boldsymbol{S}$ or through $\gamma\boldsymbol{P}_0$, and $\alpha=\gamma$ matches the recurrent decay $\kappa$ per synaptic step to the planning cost $q_0/\lambda$ per step. The plasticity rule samples each accessible transition forward and backward equally often (Sec.~\ref{sec:btsp-rule}), so $\boldsymbol{P}_0$ satisfies detailed balance with respect to a positive measure $\boldsymbol{\pi}$, $\pi_iP_0(j\mid i)=\pi_jP_0(i\mid j)$, and the same rule writes
\begin{equation}
\boldsymbol{S} = \boldsymbol{\Pi}^{1/2}\, \boldsymbol{P}_0\, \boldsymbol{\Pi}^{-1/2},
\qquad
\boldsymbol{\Pi} := \operatorname{diag}(\pi_1,\dots,\pi_N),
\label{eq:S-from-P0}
\end{equation}
which detailed balance makes symmetric (Apps.~\ref{app:reversibility-symmetric-coordinate} and~\ref{app:lmdp-exact-linear-network}).

\begin{theorem}[Equivalence between goal-clamped relaxation and the LMDP desirability]
\label{thm:equiv}
Suppose $\boldsymbol{P}_0$ is reversible with respect to a positive measure $\boldsymbol{\pi}$, the recurrent weights are $\boldsymbol{S}=\boldsymbol{\Pi}^{1/2}\boldsymbol{P}_0\boldsymbol{\Pi}^{-1/2}$, $\alpha=\gamma$, and the network operates in the linear regime of Eq.~\eqref{eq:linear}, so that the goal-clamped steady state solves Eq.~\eqref{eq:dirichlet}. Then at every unit $i$,
\begin{equation}
\frac{u_i}{\bar u_g}
\;=\;
\sqrt{\frac{\pi_i}{\pi_g}}\; z^g(i),
\label{eq:state-map}
\end{equation}
where $z^g$ is the LMDP desirability field with discount $\gamma$ and $z^g(g)=1$, and the optimal policy is carried by the branch synaptic currents, $\mu^*(j\mid i)\propto S_{ij}\,u_j$.
\end{theorem}

The remainder dropped in Eq.~\eqref{eq:dirichlet} has an explicit bound (App.~\ref{app:linearization-homogeneous-gain}). The step cost $q_0$ sets the recurrent decay and penalizes long trajectories; the deviation cost measures how far each planned step bends the spontaneous flow away from rest; the relaxation stores the value of this trade-off at every unit. When units cover the environment evenly, bidirectional sampling makes $\boldsymbol{\pi}$ uniform (App.~\ref{app:reversibility-lazy-normalization}), $\boldsymbol{S}=\boldsymbol{P}_0$, and the normalized field $\psi^g(i):=u_i/\bar u_g$ coincides with $z^g(i)$ (general case: App.~\ref{app:asymmetry-readout-bias}).

\label{sec:perturbation}Strong recall leaves the linear regime at the clamped unit, where $\phi$ saturates. The leading effect is a uniform rescaling of the whole field: the drive the clamped unit sends into the network shrinks to an effective amplitude $A_{\mathrm{eff}}(\bar u_g)<\bar u_g$ (App.~\ref{app:perturbation-nonlinear-steady}), every non-goal unit, still in its linear range, is rescaled by the same factor, and the log-gradient readout of Sec.~\ref{sec:readout} cancels it. A residual shape distortion remains in a shell around the goal where neighboring units also saturate; it decays exponentially with distance, so the shell radius grows only as $\log\bar u_g$ and stays bounded for any saturating $\phi$ (App.~\ref{app:perturbation-far-field-amplitude}; Sec.~\ref{sec:exp-linearization}).

Reversibility is the condition for a symmetric realization, not for the computation; the identity extends to genuinely one-way transitions:

\begin{corollary}[Directed transitions]
\label{cor:directed}
Let $\boldsymbol{P}_0$ be any transition kernel for which Eq.~\eqref{eq:lmdp} has a unique solution, and let $\boldsymbol{S}=\boldsymbol{D}\boldsymbol{P}_0\boldsymbol{D}^{-1}$ for any positive diagonal $\boldsymbol{D}=\operatorname{diag}(d_1,\dots,d_N)$. Under the remaining assumptions of Theorem~\ref{thm:equiv}, $u_i/\bar u_g=(d_i/d_g)\,z^g(i)$ and $\mu^*(j\mid i)\propto S_{ij}u_j$ (App.~\ref{app:lmdp-directed-corollary}).
\end{corollary}

In the directed case the legal move $i\to j$ is stored in the synapse from unit $j$ to unit $i$, $S_{ij}\propto P_0(j\mid i)$, so activity flows from the goal backward along legal routes. This is the synapse that reverse replay writes under a causal plasticity kernel: a sequence that plays $j$ before $i$ potentiates $j\to i$, so reverse replay of the experienced trajectory writes $\boldsymbol{S}$, while forward replay writes its transpose, the predictive operator of Sec.~\ref{sec:readout}. Reverse replay is prominent at reward \citep{FosterWilson2006,Ambrose2016}, and the BTSP rule inferred from CA3 place-field dynamics is strongly asymmetric in a new environment and close to symmetric once it is familiar \citep{Madar2025}; in a familiar reversible environment the distinction disappears, because $\boldsymbol{S}$ is symmetric.

\subsection{A place-tuned readout of the goal-locked field}
\label{sec:readout}

To choose a step, the animal needs a single direction at its current position $\boldsymbol{y}_t$. It is supplied by a downstream population of place-tuned readout cells driven by the entorhinal and self-motion inputs that set place tuning at rest \citep{OKeefeBurgess1996,McNaughton2006,Aoki2019,Qian2025}. Each readout cell has a preferred location $\boldsymbol{x}$ and receives a place-tuned input $\mathcal{G}_{\boldsymbol{y}_t}(\boldsymbol{x})$, largest at $\boldsymbol{x}=\boldsymbol{y}_t$ and falling off over the place-field width, together with the value $\psi^g(\boldsymbol{x})$ from the upstream unit coding the same location; the two multiply on the dendrite \citep{MitchellSilver2003,Silver2010}, so the cell fires at a rate proportional to $\mathcal{G}_{\boldsymbol{y}_t}(\boldsymbol{x})\,\psi^g(\boldsymbol{x})$ (Fig.~\ref{fig:concept}C). The next-step direction is the population vector, which weights each cell's preferred displacement $\boldsymbol{x}-\boldsymbol{y}_t$ by its rate:
\begin{equation}
\Delta\boldsymbol{x}_{\mathrm{pop}}(\boldsymbol{y}_t,g)
= \frac{\int (\boldsymbol{x}-\boldsymbol{y}_t)\,\mathcal{G}_{\boldsymbol{y}_t}(\boldsymbol{x})\,\psi^g(\boldsymbol{x})\,d\boldsymbol{x}}
       {\int \mathcal{G}_{\boldsymbol{y}_t}(\boldsymbol{x})\,\psi^g(\boldsymbol{x})\,d\boldsymbol{x}}
\;\approx\; \boldsymbol{\Sigma}_G(\boldsymbol{y}_t)\,\nabla\!\log\psi^g(\boldsymbol{y}_t),
\label{eq:popvec}
\end{equation}
where $\boldsymbol{\Sigma}_G(\boldsymbol{y}_t)$ is the covariance of $\mathcal{G}_{\boldsymbol{y}_t}$ around $\boldsymbol{y}_t$. The approximation is a local Taylor expansion of $\log\psi^g$, valid when the place field is narrower than the scale on which $\psi^g$ varies (App.~\ref{app:popvec-centered-kernel}), and it matches the mean displacement of the optimal next-step distribution $\mu^*$ at $\boldsymbol{y}_t$ (App.~\ref{app:popvec-lmdp}). The circuit never computes $\boldsymbol{x}-\boldsymbol{y}_t$ online: each readout cell's offset from the current position is a fixed property of that cell, so the population vector is a projection through offset-weighted synapses, with divisive normalization supplying the denominator \citep{Silver2010}.

Two invariances make the readout robust. The direction is unchanged by any overall rescaling of $\psi^g$, so uniform gain changes in the recurrent population, including changes of the recall amplitude $\bar u_g$, drop out (Sec.~\ref{sec:perturbation}). And the readout can take the raw firing rate $r=r^0+u$ rather than the deviation $u$:

\begin{proposition}[Reading the raw rate]
\label{prop:raw-rate}
Let the kernel be centered, $\int(\boldsymbol{x}-\boldsymbol{y}_t)\mathcal{G}_{\boldsymbol{y}_t}(\boldsymbol{x})\,d\boldsymbol{x}=\boldsymbol{0}$, and the resting rate homogeneous, $r^0(\boldsymbol{x})=c$. Then $\Delta\boldsymbol{x}_{\mathrm{pop}}$ computed from $c+u$ equals $\Delta\boldsymbol{x}_{\mathrm{pop}}$ computed from $u$ times the positive factor $\int\mathcal{G}u/(cM_0+\int\mathcal{G}u)$, $M_0=\int\mathcal{G}$; the direction is exactly invariant (App.~\ref{app:popvec-raw-rate}).
\end{proposition}

A position-dependent resting profile adds only a goal-independent term, which the same centering removes; at rest the readout returns the place code itself, because $\Delta\boldsymbol{x}_{\mathrm{pop}}$ vanishes for a symmetric kernel.

Prediction and control are two queries of the same synapses (Fig.~\ref{fig:concept}D): a source at the current state, $(\boldsymbol{I}-\alpha\boldsymbol{S})^{-1}\boldsymbol{e}_{\boldsymbol{y}_t}$, is the resolvent column of a successor representation \citep{Dayan1993,Stachenfeld2017} and predicts where the animal will be; a clamp at the goal is the Dirichlet problem of Eq.~\eqref{eq:dirichlet} and tells it which way to go.

\subsection{The weak-cost limit recovers Green's-function navigation}
\label{sec:exp-green}

A recent account of hippocampal--entorhinal coding \citep{Zuo2026Laplacian} assigned two roles to a discrete Laplacian on the explored environment: its eigenfunctions reproduce the firing maps of grid, band, and boundary cells, and its Green's function is a potential whose gradient guides goal-directed paths. The recurrent relaxation supplies the circuit. Let $\boldsymbol{L}:=\boldsymbol{I}-\boldsymbol{S}$ be the Laplacian of the recurrent connectivity; it equals the Laplacian of \citet{Zuo2026Laplacian} up to a positive scalar (App.~\ref{app:green-laplacian-identity}), and $\boldsymbol{L}$, $\boldsymbol{S}$, and that Laplacian share one set of eigenvectors, so those spatial codes emerge as firing patterns of a network whose connectivity BTSP writes. When $q_0\to0$, $\alpha\to1$ and Eq.~\eqref{eq:dirichlet} reduces to $\boldsymbol{L}_{\mathcal{I}_g\mathcal{I}_g}\boldsymbol{u}_{\mathcal{I}_g}=\boldsymbol{S}_{\mathcal{I}_g g}\bar u_g$, whose solution is the Green's function of $\boldsymbol{L}_{\mathcal{I}_g\mathcal{I}_g}$, the navigation potential of \citet{Zuo2026Laplacian} (App.~\ref{app:green-boundary-driven}). In the planning picture, $\boldsymbol{z}^g$ flattens as $q_0\to0$ and its structure moves into the leading correction $\log z^g(i)\approx-(q_0/\lambda)\,\mathbb{E}_i[\tau_g]$, the expected number of spontaneous steps to the goal; the two cost regimes are two ends of one axis, both arising from the same weights (App.~\ref{app:exp-green-limit}).

\subsection{Assumptions and scope}
\label{sec:scope}

Table~\ref{tab:scope} lists the assumptions behind the results and where each is relaxed or tested. The units of the theory are assemblies: the computation is unchanged under any full-rank recoding of the population, so overlapping non-negative codes, place-cell-like in form, preserve the field and the policy (App.~\ref{app:exp-extended-codes}); one unit per location is a coordinate choice.

\begin{table}[t]
\centering
\footnotesize
\setlength{\tabcolsep}{4pt}
\renewcommand{\arraystretch}{1.05}
\begin{tabular}{@{}>{\raggedright\arraybackslash}p{0.42\linewidth}>{\raggedright\arraybackslash}p{0.56\linewidth}@{}}
\toprule
Assumption (what it gives) & Where it is relaxed or tested \\
\midrule
Reversible $\boldsymbol{P}_0$ (symmetric $\boldsymbol{S}$; Theorem~\ref{thm:equiv}; gradient flow) & directed transitions: Corollary~\ref{cor:directed}, Fig.~\ref{fig:scope}A; sampling asymmetry: App.~\ref{app:exp-asymmetric-btsp} \\
Linear regime of Eq.~\eqref{eq:linear} (steady state solves the desirability equation) & saturation rescales the field, shell near the goal: Sec.~\ref{sec:perturbation}, Fig.~\ref{fig:equivalence}C--F \\
$\alpha=\gamma<1$ (recurrent decay equals the step cost; stability) & horizon and latency across $\alpha$: Fig.~\ref{fig:scope}D, App.~\ref{app:exp-extended-latency} \\
Transition-gated BTSP, symmetric kernel ($\boldsymbol{S}$ vanishes on walls; $\boldsymbol{S}=\boldsymbol{P}_0$) & asymmetric kernel: Sec.~\ref{sec:equivalence}, App.~\ref{app:exp-asymmetric-btsp}; partial coverage: Fig.~\ref{fig:scope}C \\
Centered local place kernel (readout follows $\nabla\log\psi^g$; raw rate suffices) & anisotropic kernels near walls: App.~\ref{app:exp-extended-readout} \\
Fixed lazy rate $\epsilon$, Eq.~\eqref{eq:lazy} (one operator for all goals and edits) & $\epsilon$ collapse; $38{,}556$-state maps: Fig.~\ref{fig:scope}F \\
\bottomrule
\end{tabular}
\caption{Assumptions, what each provides, and where each is relaxed or tested.}
\label{tab:scope}
\end{table}

\section{Learning the recurrent connectivity}
\label{sec:learning}

\subsection{Transition-gated symmetric plasticity}
\label{sec:btsp-rule}

Theorem~\ref{thm:equiv} requires that $\boldsymbol{S}$ vanishes on inaccessible transitions and that $\boldsymbol{P}_0$ satisfies detailed balance; both arise from BTSP during goal-free exploration. BTSP is triggered by dendritic plateau potentials and potentiates synapses whose presynaptic input fired within a few seconds of the plateau \citep{Bittner2017,MageeGrienberger2020}; CA3 place fields emerge through a symmetric BTSP-like process at recurrent CA3--CA3 synapses \citep{Li2024}, consistent with broad symmetric spike-timing-dependent plasticity there \citep{Mishra2016}, and symmetric rules support stable predictive spatial codes in models \citep{Keck2025}. We model the induction kernel as symmetric, $k(\tau)=k(-\tau)\ge0$, with $\tau$ the lag between pre- and postsynaptic firing within the BTSP window.

The seconds-long window lets BTSP bind firing at successively visited locations, in sequences that arise from physical exploration \citep{McNaughton2006} and from theta sequences and replay \citep{JohnsonRedish2007,Davidson2009,WikenheiserRedish2015,PfeifferFoster2013,WidloskiFoster2022}. Let $a_i(t)\ge0$ be the firing rate of unit $i$ and $\Xi_{ij}(t,\tau)\in\{0,1\}$ a transition gate that equals one when units $i$ and $j$ fire in sequence at lag $\tau$ within the BTSP window centered at $t$, marking an accessible transition. The accumulated weight is
\begin{equation}
W^{\mathrm{tr}}_{ij}
=\sum_\tau k(\tau)\,
\mathbb{E}\!\bigl[a_i(t)\,a_j(t+\tau)\,\Xi_{ij}(t,\tau)\bigr].
\label{eq:btsp}
\end{equation}
Walls permit no traversal, and replay respects the same barriers \citep{Gupta2010,WuFoster2014,WidloskiFoster2022}, so $W^{\mathrm{tr}}_{ij}=0$ on these pairs. Traversals and replay sample accessible transitions in either direction \citep{FosterWilson2006,DibaBuzsaki2007,Gupta2010}, so with the symmetric kernel $\boldsymbol{W}^{\mathrm{tr}}$ is symmetric in expectation; the asymmetric residue left by finite sampling and goal-directed behavior is controlled by a single scale, $\chi=\alpha\|\boldsymbol{E}_{\rm asym}\|/(1-\alpha)$ (Apps.~\ref{app:asymmetry-irreversible-sampling} and~\ref{app:exp-asymmetric-btsp}). We turn $\boldsymbol{W}^{\mathrm{tr}}$ into the transition kernel by distributing each row's outgoing weight among accessible neighbors and keeping the remainder as the probability of staying put:
\begin{equation}
P_0(j\mid i)=\epsilon\,W^{\mathrm{tr}}_{ij}\;(j\ne i),
\qquad
P_0(i\mid i)=1-\epsilon\sum_{j\ne i}W^{\mathrm{tr}}_{ij},
\qquad
0<\epsilon\le 1/\max_i\textstyle\sum_{j\ne i}W^{\mathrm{tr}}_{ij},
\label{eq:lazy}
\end{equation}
with the lazy rate $\epsilon$ fixed once per environment and kept across later edits. When $\boldsymbol{W}^{\mathrm{tr}}$ is symmetric, so is $\boldsymbol{P}_0$; detailed balance holds with respect to the uniform measure, $\boldsymbol{S}=\boldsymbol{P}_0$, and Theorem~\ref{thm:equiv} applies (nonuniform measures: App.~\ref{app:asymmetry-readout-bias}). One set of synapses thus carries the map and generates the network's spontaneous flow through it.

\subsection{Local environmental change as a low-rank update}
\label{sec:lowrank}

The plasticity that wrote the map remains active after the environment changes. When a corridor opens or a doorway closes, transition events at the affected synapses resume or cease and the corresponding entries of $\boldsymbol{W}^{\mathrm{tr}}$ update, while entries away from the change keep their values; a few traversals through the changed region update the connectivity for all future plans. Editing $k$ synapses changes $k$ entries of $\boldsymbol{S}$, and the change in goal-locked fields decomposes (App.~\ref{app:woodbury-unit-clamped-all-goal-fields}) into a rank-$\mathcal{O}(k)$ propagation pattern plus a per-goal scaling that keeps each clamp at its amplitude. The readout is invariant to this scaling, so a few directions of synaptic change explain the rerouting toward every goal (Sec.~\ref{sec:exp-btsp-navigation}).

\section{Experiments}
\label{sec:experiments}

We tested the theory on grid mazes, each open cell coded by one unit with eight-neighbor connectivity, and on irregular graphs of place centers in continuous arenas (App.~\ref{app:exp-common}; code in App.~\ref{app:exp-reproduction}).

\begin{figure}[t]
\centering
\includegraphics[width=\linewidth]{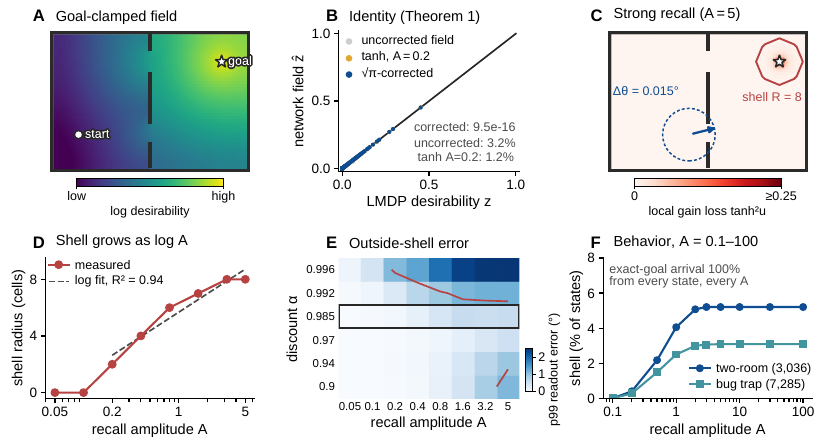}
\caption{\textbf{The relaxation reproduces the LMDP field, and the readout is amplitude-robust.}
(\textbf{A}) Two-room maze; color: log desirability at $\alpha=0.985$.
(\textbf{B}) Network steady state in the LMDP coordinate against the LMDP solution $\boldsymbol{z}$: the $\sqrt{\boldsymbol{\pi}}$-corrected field lies on the identity ($9.5\times10^{-16}$); uncorrected, it deviates by $3.2\%$; weak nonlinear recall ($A=0.2$) by $1.2\%$.
(\textbf{C}) Strong recall ($A=5$) confines the loss of local gain to a shell of radius $8$ cells around the goal; at a probe outside the shell the readout error is $0.015^\circ$.
(\textbf{D}) Shell radius grows as $\log A$ ($R^2=0.94$).
(\textbf{E}) Outside-shell 99th-percentile readout error across $\alpha$ and $A$; the row $\alpha=0.985$ (box) stays below $1^\circ$ (contour).
(\textbf{F}) For $A=0.1$--$100$ the shell covers at most $5.2\%$ of the two-room maze and $3.1\%$ of the bug trap, and every rollout from every start reaches the goal.}
\label{fig:equivalence}
\end{figure}

\subsection{The relaxation reproduces the LMDP field, and the readout is amplitude-robust}
\label{sec:exp-linearization}

Theorem~\ref{thm:equiv} predicts that $\boldsymbol{\psi}^g$ equals the desirability $\boldsymbol{z}$ up to the factor $\sqrt{\pi_i/\pi_g}$. To make this factor visible we used a two-room maze whose narrow doorways make the spontaneous occupancy mildly non-uniform (Fig.~\ref{fig:equivalence}A). At $\alpha=0.985$ the corrected field $\widehat z_i:=\sqrt{\pi_g/\pi_i}\,\psi^g_i$ matched $\boldsymbol{z}$ to machine precision, and the uncorrected field differed by $3.2\%$, the size of the correction (Fig.~\ref{fig:equivalence}B). The agreement held across all $\alpha$ tested and on a $7{,}285$-state bug-trap maze, where the branch-current policy matched $\mu^*$ to $2.8\times10^{-16}$ in total variation (App.~\ref{app:exp-equivalence}).

With $\phi=\tanh$ and recall amplitude $A:=\bar u_g$ from $0.05$ to $5$, the strongest recall deviated from the linear prediction inside a shell of about eight units around the goal (Fig.~\ref{fig:equivalence}C); outside it the readout direction matched the linear prediction to $0.15^\circ$ on average and $0.60^\circ$ at the 99th percentile. The shell radius grew as $\log A$ (Fig.~\ref{fig:equivalence}D), and a joint sweep over $\alpha$ and $A$ kept the outside-shell error below $1^\circ$ along the main row (Fig.~\ref{fig:equivalence}E). From $A=0.1$ to $100$ on both mazes the shell never exceeded $5.2\%$ of the states, greedy rollouts from every state reached the goal in all $103{,}190$ cases (Fig.~\ref{fig:equivalence}F), and stochastic rollouts started inside the shell, where the goal is adjacent, were captured by it in all $33{,}792$ cases (App.~\ref{app:exp-extended-recall}).

\begin{figure}[t]
\centering
\includegraphics[width=\linewidth]{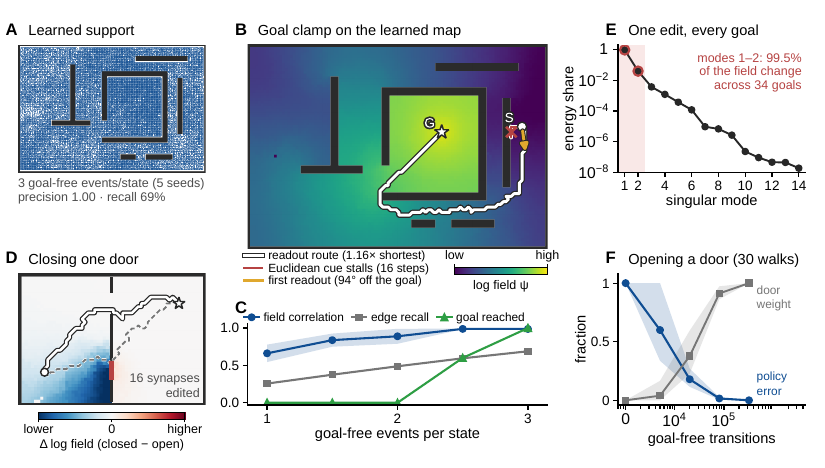}
\caption{\textbf{Goal-free BTSP solves the bug trap, and a local edit reroutes plans globally.}
(\textbf{A}) Connectivity learned by goal-free BTSP after $3$ transitions per unit (precision $1.00$, recall $69\%$; $5$ runs).
(\textbf{B}) Goal-clamped field (color) and population-vector route (white, $1.16\times$ the shortest legal route); the first readout points $94^\circ$ away from the Euclidean direction, and an agent following the Euclidean cue stalls at the wall after $16$ steps.
(\textbf{C}) Field correlation with the true-layout field, edge recall, and goal arrival against transitions per unit.
(\textbf{D}) Two-door maze: closing the near door edits the $16$ crossing synapses, and the next relaxation reroutes through the far door (white); color: field difference, closed minus open.
(\textbf{E}) Singular spectrum of the field differences across $34$ recalled goals: the first two modes carry $99.5\%$ of the energy.
(\textbf{F}) A newly opened doorway is learned by goal-free exploration alone: the door weight is written and the policy error relative to the true open map falls to zero in all $30$ trajectories.}
\label{fig:btsp}
\end{figure}

\subsection{Goal-free BTSP solves the bug trap, and a local edit reroutes plans globally}
\label{sec:exp-btsp-navigation}

In the bug-trap maze the goal sits a short Euclidean distance from the start, yet the only legal route detours through a distant entrance (Fig.~\ref{fig:btsp}B). Each sampled transition strengthened the forward and reverse synapses equally (App.~\ref{app:exp-btsp}). After $3$ transitions per unit the learned connectivity recovered $69\%$ of the accessible connections with no false positives (Fig.~\ref{fig:btsp}A), and its goal-clamped field correlated at $0.988\pm0.003$ (mean $\pm$ SEM, $5$ runs) with the field of the true layout. Every run reached the goal by following the population vector, with median route length $1.16\times$ the shortest legal route; the first readout already pointed $94^\circ$ away from the Euclidean direction, toward the distant entrance (Fig.~\ref{fig:btsp}B,C). On the same map a Dijkstra oracle gives route ratio $1.00$ and the raw-rate readout of Proposition~\ref{prop:raw-rate} gives $1.03$ (App.~\ref{app:exp-extended-baselines}).

To test the low-rank rerouting of Sec.~\ref{sec:lowrank}, we trained BTSP on a two-door maze ($2$ transitions per unit). Closing the near door updated only the $16$ synapses that cross it, and the next relaxation rerouted plans through the far door without further learning (Fig.~\ref{fig:btsp}D). Stacking the field changes for $34$ recalled goals into a matrix, the first two singular modes carried $99.5\%$ of the energy (Fig.~\ref{fig:btsp}E): one local edit reroutes plans toward every goal through a rank-$2$ correction (App.~\ref{app:exp-extended-edits}). The converse edit is learned online: when a doorway opens, goal-free exploration writes the door synapses and the policy error relative to the true open map falls to zero in all $30$ trajectories (Fig.~\ref{fig:btsp}F); a closed doorway is detected by blocked transitions, and a gate on blocked attempts rerouted $8/8$ streaming episodes (App.~\ref{app:exp-extended-edits}).

\begin{figure}[t]
\centering
\includegraphics[width=\linewidth]{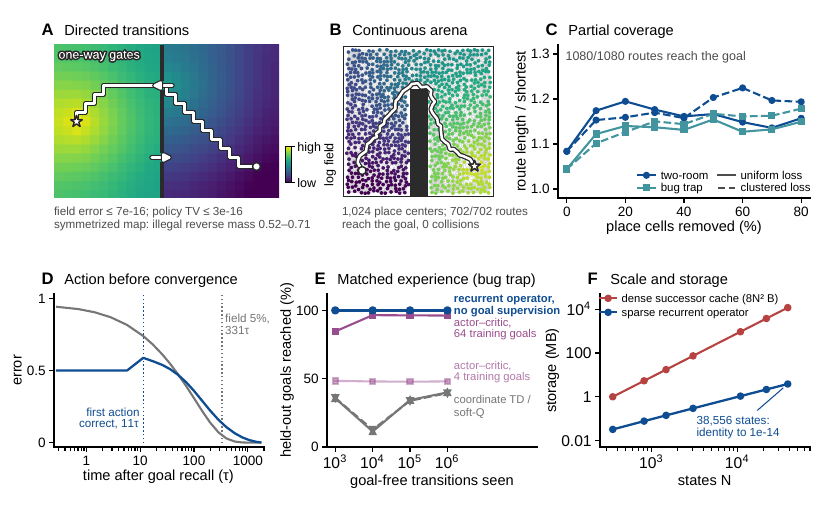}
\caption{\textbf{Scope: directed maps, continuous arenas, partial coverage, anytime readiness, matched experience, and scale.}
(\textbf{A}) A map with two one-way gates: the directed realization of Corollary~\ref{cor:directed} reproduces the directed LMDP field to machine precision; a symmetrized map assigns $0.52$--$0.71$ of its mass to illegal reverse crossings.
(\textbf{B}) Irregular graph of $1{,}024$ place centers in a continuous arena with barriers: $702/702$ routes reach the goal without collision.
(\textbf{C}) Random or clustered loss of place cells: all $1{,}080$ routes reach the goal through $80\%$ dropout on both mazes.
(\textbf{D}) Anytime readiness (bug trap, $\alpha=0.993$): the current action is correct after $11\tau$; the whole field converges to $5\%$ after $331\tau$.
(\textbf{E}) Matched experience on the bug trap, held-out goals: the recurrent operator reaches $100\%$ of stochastic arrivals from $10^3$ transitions on; soft-Q and goal-coordinate TD learners about $40\%$ at $10^6$; a goal-conditioned actor--critic $96\%$ with $64$ training goals, $48\%$ with $4$.
(\textbf{F}) Storage grows linearly with the number of states through $38{,}556$ states, where the field identity holds to $10^{-14}$; a dense all-goal cache grows quadratically.}
\label{fig:scope}
\end{figure}

\subsection{Scope: directed maps, continuous arenas, partial coverage, latency, and baselines}
\label{sec:exp-scope}

Figure~\ref{fig:scope} collects tests of the assumptions in Table~\ref{tab:scope}. On a $425$-state map with two one-way gates, the directed realization of Corollary~\ref{cor:directed} matched the directed LMDP field to $6.8\times10^{-16}$ and its policy to $2.8\times10^{-16}$ over $18$ conditions, whereas a symmetrized map assigned up to $0.71$ of its probability mass to an illegal reverse crossing (Fig.~\ref{fig:scope}A; acquisition from random walks in App.~\ref{app:exp-extended-directed}). On irregular graphs of place centers in continuous arenas, $702/702$ routes reached the goal without touching a barrier (Fig.~\ref{fig:scope}B), and refining the discretization converged monotonically to the screened-Poisson continuum problem (App.~\ref{app:exp-extended-continuum}). Removing up to $80\%$ of place cells, uniformly or in clusters, left $1{,}080/1{,}080$ routes successful (Fig.~\ref{fig:scope}C), and an operator written from scratch after dropout reached the goal in $689/720$ routes at $50\%$ loss (App.~\ref{app:exp-extended-coverage}). Heterogeneous unit gains, activity noise, and synaptic jitter are covered in App.~\ref{app:exp-extended-robustness}.

The transient is anytime: the local action depends on the local gradient of a partially relaxed field, so it is correct long before the whole field converges. On the bug trap at $\alpha=0.993$ the current action is ready after $11\tau$ while the field converges to $5\%$ after $331\tau$ (Fig.~\ref{fig:scope}D), and across $24$ conditions with $\alpha\ge0.985$ the median speedup was $82\times$ (App.~\ref{app:exp-extended-latency}). Latency is set by the spectral gap of $\boldsymbol{I}-\alpha\boldsymbol{S}$ (measured slow rate within $10^{-12}$ of the spectral prediction), and the field of a nearby goal or the pre-edit field serves as a warm start.

At matched experience, the operator was compared with goal-conditioned learners trained on the same transition streams (Fig.~\ref{fig:scope}E; App.~\ref{app:exp-extended-baselines}). On the bug trap with held-out goals it reached $5{,}120/5{,}120$ stochastic arrivals at every budget from $10^3$ to $10^6$ transitions; a tabular soft-Q learner and goal-coordinate TD value learners reached $39$--$40\%$ at $10^6$, and a goal-conditioned soft actor--critic $96\%$ with $64$ training goals and $48\%$ with $4$. Successor representations built from the same resolvent coordinates agree with our fields and behave the same; the difference is storage: the operator keeps $\mathcal{O}(N)$ synapses and answers any goal by one relaxation, whereas a dense all-goal cache grows as $N^2$ ($252\times$ at $3{,}036$ states; Fig.~\ref{fig:scope}F), and the identities held through $38{,}556$ states (App.~\ref{app:exp-extended-scale}).

\section{Discussion}
\label{sec:discussion}

We presented a hippocampal recurrent network that plans obstacle-avoiding routes through its own relaxation dynamics: BTSP writes the recurrent weights during goal-free exploration, goal recall relaxes the same circuit into the LMDP desirability field with the optimal policy in its branch currents, and a place-tuned readout converts the field into a direction of travel.

\textbf{Position in the literature.} Models that let a recurrent network settle toward a plan divide by where the goal enters. When the source sits at the current state, as in recurrent successor representations \citep{Fang2023} and Endotaxis \citep{Zhang2024Endotaxis}, the relaxed state is a predictive field and the goal enters through a readout. When the goal enters as a boundary condition, as in wavefront planning \citep{PonulakHopfield2013} and the spacetime attractor of \citet{Jensen2026}, the relaxed state supports control as an arrival time or a trajectory over states and future times. Our claim is the identity: one recurrent operator stores the map and relaxes into the provably optimal plan, with the policy in its own currents and the operator written by an observed plasticity rule. The two boundary-driven models are complementary: the desirability field is an exponentially weighted compression of the spacetime attractor's future time slices, and prefrontal circuits could supply the goal boundary condition that the hippocampal network evaluates on its fine-grained graph \citep{Ito2015}.

\textbf{Alternative goal mechanisms.} Goal-modulated connectivity would need goal-specific synapses; the model has none, and the first recall of a new goal is already correct. A learned goal-conditioned readout needs per-goal supervision; ours is zero-shot, and the learned models of Fig.~\ref{fig:scope}E need tens of training goals to approach it. Decision-time replay predicts latency that grows with path length; here it is set by the spectral gap.

\textbf{Predictions.} Four consequences are falsifiable in recordings during goal recall. (i) Place-field centers stay stable during recall and acquire a multiplicative goal-dependent gain proportional to $\psi^g(\boldsymbol{x})$, larger for cells with shorter legal routes to the goal, consistent with the stable CA1 fields of \citet{Duvelle2021}; successor accounts predict fields reshaped toward reward. (ii) The first recall of a novel goal is already correct. (iii) Decision latency is set by the spectral gap of the recurrent operator and does not grow with path length. (iv) A local change of the environment produces a low-rank change of the goal-locked fields, shared across goals and concentrated near the edit, detectable within a few traversals in barrier-change paradigms \citep{WidloskiFoster2022}.

\textbf{Limitations.} The model describes a hippocampal population without committing to a subfield; BTSP-like signatures have been reported in several \citep{Bittner2017,MageeGrienberger2020,Qian2025,Li2024}, and intracellular recordings during planning pauses could test which one hosts it. The transition gate $\Xi_{ij}$ is idealized: self-motion signals and internally generated sequences likely cooperate to restrict BTSP to accessible transitions, and identifying this mechanism, with the BTSP kernel and the inhibitory gain in vivo, is a concrete next step. The exact branch-current policy needs action-labelled access to individual synaptic currents; the population vector is the implemented approximation (App.~\ref{app:exp-extended-readout}).

\textbf{Outlook.} The same idea extends beyond spatial navigation. Hippocampal--entorhinal maps have been reported for sensory continua, social hierarchies, and abstract task structures \citep{Behrens2018,Garvert2023,Courellis2024,Tacikowski2024}, where a ``step'' is a movement in space, a shift of attention, or an imagined step along a chain of associations (state-dependent costs: App.~\ref{app:exp-extended-costs}). Behind the technical equivalence lies a simple picture of cognition. Each synaptic weight $S_{ij}$ records one transition, and the agent's actions decide which transitions exist. States are then defined by the transitions they participate in: the transitions come first, and the states emerge from them. Goal-directed cognition becomes the relaxation of this learned web of transitions, and the planning is the relaxation itself.

\begin{ack}
This work received no third-party funding. The authors declare no competing interests.
\end{ack}


\newpage
\appendix
\AppendixContents
\EnableAppendixContents
\newpage

\section{Related work}
\label{app:related-work}

\paragraph{Cognitive maps as learned transition structure.}
The cognitive map hypothesis began as a behavioral account of flexible 
navigation \citep{Tolman1948,OKeefeNadel1978}. Neural correlates were 
later identified in hippocampal place cells, entorhinal grid cells, and 
experience-dependent spatial coding \citep{OKeefeDostrovsky1971,
OKeefeBurgess1996,Hafting2005,McNaughton2006,McHugh1996,Mehta2000,
Lever2002}. The same principle has since been extended beyond physical 
space, to relational, temporal, and abstract task structure 
\citep{Behrens2018,Whittington2020,Whittington2022Review,George2021,
ElGaby2024,Courellis2024,Tacikowski2024,Garvert2023,Bakermans2025}. Our 
model adopts this transition-centric view and goes one step further: the 
transition structure is encoded directly in the recurrent connectivity. 
Each recurrent synapse carries the probability of a locally accessible 
move; walls and blocked corridors leave the corresponding synapses at 
zero; and goal recall clamps a single neuron in the same population. 
Planning then emerges as the relaxation of the very circuit that stores 
the map.

\paragraph{Successor representations and predictive maps.}
Successor representations (SRs) store long-horizon expected occupancies 
under a reference policy and convert any reward or goal vector into a 
value function \citep{Dayan1993,Russek2017,Momennejad2017,Gershman2018}. 
This idea has been developed into hippocampal predictive-map accounts, 
in which place- and grid-like codes reflect the eigenstructure of a 
learned transition matrix \citep{Stachenfeld2017,deCothiBarry2020,
deCothi2022,PirayDaw2021,Bono2023,Fang2023,George2023}. These models 
explain why hippocampal representations should carry predictive 
structure, while leaving open how the predictive computation is 
realized in a recurrent circuit. Our framework supplies this circuit: 
the predictive field is generated at recall time by a goal-clamped 
relaxation, and we prove it equals the desirability field of an LMDP 
(Theorem~\ref{thm:equiv}). The same recurrent connectivity serves all 
goals, with the choice of goal entering only as a boundary clamp at 
recall.

\paragraph{Linearly-solvable planning, harmonic navigation, and Laplacian 
methods.}
Linearly-solvable Markov decision processes (LMDPs) reformulate 
KL-regularized planning as a linear desirability equation 
\citep{Kappen2005,Todorov2006,Todorov2009,Levine2018}. In robotics, 
potential fields, harmonic functions, and Laplace equations have long 
served as path-planning tools that generate obstacle-avoiding routes 
\citep{Khatib1986,ConnollyBurnsWeiss1990,ConnollyGrupen1993,Connolly1997,
RimonKoditschek1992}. In reinforcement learning, Laplacian eigenfunctions 
and proto-value functions provide geometry-aware bases for representation 
and option discovery \citep{MahadevanMaggioni2007,Machado2017}. In 
neuroscience and machine learning, Laplacian and spectral accounts have 
been linked to grid-cell coding and spatial abstraction 
\citep{Dordek2016,Stachenfeld2017,Banino2018,CuevaWei2018,Zuo2026Laplacian}. 
These four lines treat the desirability equation, the harmonic potential, 
the Laplacian basis, and the cognitive map as four separate objects. 
Our framework joins them in one recurrent circuit: the same connectivity 
that stores transitions, under a goal clamp, produces a solution to the 
LMDP desirability equation; in the weak-cost limit, this solution 
converges to a Green's function of the Laplacian formed from the same 
connectivity.

\paragraph{Planning as network dynamics.}
Recurrent neural networks have long been proposed as computational 
substrates for pattern completion, attractor dynamics, and optimization 
\citep{Amari1977,HopfieldTank1985,Tsodyks1999}. In hippocampal-entorhinal 
models, continuous attractors support path integration and spatial codes 
\citep{SamsonovichMcNaughton1997,BurakFiete2009}, and spiking and 
recurrent network models can generate wavefronts, preplay events, and 
goal-directed sequences \citep{ErdemHasselmo2012,PonulakHopfield2013,
CorneilGerstner2015,HagaFukai2018,Jensen2024}. In machine learning, 
differentiable planning modules embed dynamic programming inside a 
neural network \citep{Tamar2016,Lee2018}. Our framework specifies a 
different role for the recurrent connectivity: the recurrent weights 
store the passive transition probabilities of the environment, the 
goal is delivered as a clamping current at a single neuron, and the 
steady state of the clamped circuit equals the LMDP desirability 
field. Planning therefore takes the form of a physical relaxation of 
the stored map, performed by the very circuit that holds the map; at 
decision time, the computation reduces to a single equilibration of 
the recurrent dynamics.

\paragraph{Replay, theta sequences, and internally generated 
trajectories.}
Hippocampal theta sequences and sharp-wave ripple replay express 
forward, reverse, and nonlocal trajectories at compressed timescales 
\citep{JohnsonRedish2007,FosterWilson2006,DibaBuzsaki2007,Davidson2009,
Gupta2010,PfeifferFoster2013,WikenheiserRedish2015,Buzsaki2015,
JooFrank2018}. These internally generated sequences have been linked 
to memory updating, planning, and reinforcement learning 
\citep{Jadhav2012,MattarDaw2018,Liu2019,Liu2021,SchuckNiv2019,
Schwartenbeck2023,Sagiv2025,He2026Ripples,Zhou2026Ripples}. Recent work 
also emphasizes that replay can serve offline learning, map 
construction, and experience restructuring, with its behavioral role 
varying across species and tasks \citep{vanDerMeerBendor2025,Forli2025}. 
Our framework assigns replay a mechanistic role compatible with these 
observations. Replay and theta sequences activate pairs of assemblies 
within the BTSP gating window, thereby writing or updating the 
corresponding recurrent transition probabilities. Once the connectivity 
is in place, the online action signal at each step is the goal-clamped 
relaxation field and its local log-gradient. The offline sampling that 
builds the map is therefore separated from the online relaxation that 
reads the map for action.

\paragraph{CA3 recurrence, BTSP, and local transition learning.}
CA3 has long been modeled as a recurrent associative network 
supporting pattern completion and memory retrieval 
\citep{Marr1971,TrevesRolls1994,Nakazawa2002}. Anatomically, CA3 
pyramidal cells are densely connected by recurrent collaterals 
\citep{AmaralWitter1989,Ishizuka1990,Sammons2024}. At the plasticity 
level, behavioral-timescale synaptic plasticity (BTSP) creates new 
place fields within seconds of behavioral input 
\citep{Bittner2017,MageeGrienberger2020,Qian2025}; recent CA3 
recordings reveal BTSP-like dynamics in memory-supporting circuits 
\citep{Li2024}; and CA3-CA3 synapses can express broad, symmetric 
timing-dependent plasticity \citep{Mishra2016}. Our learning rule 
abstracts these observations into a transition-gated symmetric 
accumulation: when an accessible move or a replayed transition 
co-activates two assemblies within the BTSP window, the corresponding 
recurrent connection is strengthened symmetrically. Walls and blocked 
corridors leave the corresponding connections at zero, since they 
yield no transition events. Exact reversibility is an idealization 
that simplifies the equivalence theorem; the appendix treats finite 
sampling and direction-biased replay as perturbations, and 
bidirectional sampling projects the empirical transition statistics 
onto their reversible component.

\paragraph{Goal coding and a local population-vector readout.}
The hippocampal formation contains goal-, reward-, and vector-related 
signals that coexist with stable spatial coding 
\citep{Hok2007,GauthierTank2018,Aoki2019,Sarel2017,OrmondOKeefe2022,
Nyberg2022,Sosa2025}, and stable place fields can persist across 
changes in replay structure or task-dependent sequence content after 
environmental edits \citep{Duvelle2021,WuFoster2014,WidloskiFoster2022}. 
Our readout mechanism explains this coexistence by keeping location 
tuning and goal-locked modulation in separate components of the same 
firing rate. A downstream place-tuned cell receives two inputs: a 
local place-kernel input centered on the animal's current position, 
and a recurrent input from the upstream desirability field evaluated 
at the cell's preferred location. The two inputs combine 
multiplicatively on the dendrite, so the cell's place-field center is 
preserved while a goal-dependent gain shapes its rate. This dendritic 
gain interaction is supported by classical work on dendritic nonlinear 
integration and shunting gain control 
\citep{MitchellSilver2003,Poirazi2003,LondonHausser2005,Silver2010}. 
The model therefore predicts that, during recall, hippocampal cells 
retain their resting place-field centers while acquiring a subthreshold 
or rate-level modulation proportional to the desirability of those 
centers.

\paragraph{Local remapping and global rerouting.}
A computational consequence of representing the environment as a 
recurrent transition matrix is that local environmental changes 
correspond to local synaptic edits. When a doorway closes or a corridor 
opens, only the affected transition weights are revised; at the next 
goal recall, the relaxation propagates this local edit and updates the 
field for whichever goal is recalled. This provides a mechanistic 
interpretation for the flexible replay rerouting around barriers 
reported in \citet{WidloskiFoster2022}, and for the adaptive planning 
that coexists with stable place fields in \citet{Duvelle2021}. 
Algebraically, a small synaptic edit is a low-rank perturbation of the 
recurrent connectivity, so the induced changes in goal-locked fields 
across all goals are captured by a small number of correction patterns. 
The map update is local, while its planning consequence is immediately 
global; this locality-globality coupling is a hallmark of the 
recurrent-relaxation account.
\paragraph{Planning as relaxation.}
A lineage of recurrent models lets a network settle toward a plan, and the models divide by where the goal enters. \citet{Muller1996} proposed that the hippocampus stores a cognitive graph whose synaptic resistances encode the distances between place fields; the graph holds sufficient information for navigation, and the optimal path is found by standard graph-search algorithms. \citet{PonulakHopfield2013} propagate a wavefront of spiking activity from the goal and read the plan from the synaptic changes it leaves behind; \citet{CorneilGerstner2015} use attractor dynamics for preplay and rapid path planning; \citet{HagaFukai2018} learn goal-directed sequences through reverse replay. Successor representations can be learned by spike-timing-dependent plasticity \citep{George2023,Bono2023} and computed by recurrent dynamics \citep{Fang2023}; recurrent successor models and Endotaxis \citep{Zhang2024Endotaxis} place the source at the current state, relax into a predictive field, and let the goal enter through a learned readout. The spacetime attractor of \citet{Jensen2026}, proposed for prefrontal cortex, has one subpopulation per future time step, receives the current location in its present subspace and reward inputs in its future subspaces, and embeds copies of the environment's adjacency matrix between consecutive subspaces; its fixed point is a trajectory over states and future times that resembles an approximate posterior of a planning-as-inference problem. These relaxed states support planning as an arrival time, a trajectory, or a predictive field. The present model differs in the status of its relaxed state: a stationary field over one population, reused across all goals and horizons, that is identical to the optimal control field of a stated problem (Theorem~\ref{thm:equiv}), with the optimal policy in the branch synaptic currents, with the operator written by BTSP during goal-free exploration, and with local edits rerouting all plans at low rank. The two kinds of query, source at the current state and clamp at the goal, are two uses of the same operator (Sec.~\ref{sec:readout}); the desirability field is an exponentially weighted compression of the future time slices of the spacetime attractor, and prefrontal circuits could supply the goal boundary condition that the hippocampal network evaluates on its fine-grained graph \citep{Ito2015,Jensen2026}.

\section{Notation and standing assumptions}
\label{app:notation}

The explored environment is represented by a finite set of assemblies
\[
\mathcal X=\{1,\ldots,N\}.
\]
Assembly \(i\) codes location \(i\), with physical coordinate \(\boldsymbol{x}_i\) when a coordinate is needed. The recalled goal is denoted by \(g\). The non-goal set is
\[
\mathcal I_g=\mathcal X\setminus\{g\}.
\]
For a matrix \(\boldsymbol{M}\), \(\boldsymbol{M}_{AB}\) denotes the submatrix with rows indexed by \(A\) and columns indexed by \(B\). Dirichlet equations are written on \(\mathcal I_g\), with the goal assembly held fixed.

If the reversible support graph has several connected components, the theory below is applied to the closed irreducible component containing the recalled goal. States in components not containing \(g\) have no passive path to \(g\), and therefore have zero Dirichlet response and zero LMDP desirability for that goal. To keep notation light, \(\mathcal X\) denotes the goal-connected component throughout the appendix.

The passive transition kernel is denoted by \(\boldsymbol{P}_0\). The row convention is
\[
(\boldsymbol{P}_0\boldsymbol{f})_i=\sum_j P_0(j\mid i)f_j ,
\]
so \(P_0(j\mid i)\) is the probability of moving from \(i\) to \(j\) under the spontaneous goal-free dynamics. The recurrent matrix \(\boldsymbol{S}\) acts with the same convention:
\[
(\boldsymbol{S}\boldsymbol{f})_i=\sum_jS_{ij}f_j .
\]
The exact symmetric theory assumes that \(\boldsymbol{P}_0\) is finite, row-stochastic, nonnegative, irreducible on the goal-connected component, and reversible with respect to a positive stationary measure \(\boldsymbol{\pi}\), normalized by \(\sum_i\pi_i=1\):
\begin{equation}
\pi_iP_0(j\mid i)=\pi_jP_0(i\mid j).
\label{eq:app-detailed-balance-notation}
\end{equation}
Let
\[
\boldsymbol{\Pi}=\operatorname{diag}(\pi_1,\ldots,\pi_N),
\qquad
\boldsymbol{S}=\boldsymbol{\Pi}^{1/2}\boldsymbol{P}_0\boldsymbol{\Pi}^{-1/2}.
\]
Then
\begin{equation}
S_{ij}=\sqrt{\frac{\pi_i}{\pi_j}}\,P_0(j\mid i),
\label{eq:app-S-coordinate}
\end{equation}
and \(\boldsymbol{S}\) is symmetric. Since \(\boldsymbol{S}\) is similar to \(\boldsymbol{P}_0\), the two matrices have the same eigenvalues. On the goal-connected irreducible component, the Perron vector of \(\boldsymbol{S}\) is
\[
\boldsymbol{\varphi}_0=\sqrt{\boldsymbol{\pi}},
\qquad
(\sqrt{\boldsymbol{\pi}})_i=\sqrt{\pi_i},
\]
with
\[
\boldsymbol{S}\boldsymbol{\varphi}_0=\boldsymbol{\varphi}_0,
\qquad
\|\boldsymbol{\varphi}_0\|_2=1.
\]
The graph Laplacian in the symmetric network coordinate is
\[
\boldsymbol{L}=\boldsymbol{I}-\boldsymbol{S}.
\]
On a connected component,
\[
\ker \boldsymbol{L}=\operatorname{span}\{\sqrt{\boldsymbol{\pi}}\}.
\]
If the full graph has several closed components, the nullspace is the span of the corresponding component-wise \(\sqrt{\boldsymbol{\pi}}\) vectors.

The Dirichlet Laplacian for a clamped goal is
\[
\boldsymbol{L}_g^{\mathrm{Dir}}
=
\boldsymbol{I}_{\mathcal I_g}-\boldsymbol{S}_{\mathcal I_g\mathcal I_g}.
\]
Because the killed passive chain on \(\mathcal I_g\) is transient when every state in the component can hit \(g\),
\[
r_{\mathrm{sp}}(\boldsymbol{P}_{0,\mathcal I_g\mathcal I_g})<1.
\]
Since \(\boldsymbol{S}_{\mathcal I_g\mathcal I_g}\) is similar to \(\boldsymbol{P}_{0,\mathcal I_g\mathcal I_g}\) and is symmetric, all eigenvalues of \(\boldsymbol{S}_{\mathcal I_g\mathcal I_g}\) are strictly smaller than one in absolute value, and
\[
\boldsymbol{L}_g^{\mathrm{Dir}}\succ0.
\]
The recurrent discount is
\[
\alpha\in(0,1),
\qquad
\kappa=-\log\alpha.
\]
The LMDP discount is
\[
\gamma=\exp(-q_0/\lambda).
\]
The exact network--LMDP equivalence uses
\[
\alpha=\gamma.
\]
The corresponding screening mass is
\[
m=\alpha^{-1}-1=\exp(q_0/\lambda)-1.
\]
All fields denoted by \(\boldsymbol{u}\) in the linear, nonlinear, and readout analyses are background-subtracted goal-induced perturbations above the goal-free fixed point. Raw firing rates have the form
\[
\boldsymbol{r}=\boldsymbol{r}^0+\boldsymbol{u},
\qquad
r_i=r_i^0+u_i.
\]
The LMDP-equivalent field is \(\boldsymbol{u}\), not the raw rate \(\boldsymbol{r}\). The normalized goal-locked activity field is
\[
\boldsymbol{\psi}^g=(\psi_i^g)_{i\in\mathcal X},
\qquad
\psi_i^g=\frac{u_i}{\bar u_g},
\]
where \(\bar u_g\) is the clamped goal perturbation. The LMDP desirability field is
\[
\boldsymbol{z}^g=(z_i^g)_{i\in\mathcal X},
\qquad
z_i^g=\exp[-V_i^g/\lambda],
\qquad
z_g^g=1.
\]
Under uniform \(\boldsymbol{\pi}\),
\[
\psi_i^g=z_i^g.
\]
For general reversible \(\boldsymbol{\pi}\),
\[
\psi_i^g=\sqrt{\frac{\pi_i}{\pi_g}}\,z_i^g.
\]
The graph distance \(d(i,j)\) is the minimum number of accessible transitions in a path from \(i\) to \(j\). For an accessible oriented edge \(e=(i,j)\), the edge gradient is
\[
\nabla_e f=f_j-f_i.
\]
When a continuous spatial gradient \(\nabla f(\boldsymbol{x}_i)\) is written on an embedded graph, it denotes the local least-squares or interpolation gradient obtained from neighboring assemblies and their coordinates.

\section{Linearization of the recurrent population}
\label{app:linearization}

Let \(r_i\) denote the raw firing rate in the recurrent rate equation
\begin{equation}
\tau_u\dot r_i
=
-r_i
+
g_s\sum_jS_{ij}\phi(r_j)
-
g_{\mathrm{inh}}\phi(r_i)
+
I_i(t),
\label{eq:app-rate-raw}
\end{equation}
where
\[
I_i(t)=I_i^{\mathrm{bg}}+I_i^{\mathrm{EC}}(\boldsymbol{y}_t)+I_i^{\mathrm{goal}}(g,t).
\]
The local inhibitory term reduces the local gain and the effective time constant in the first-order dynamics.

\subsection{Expansion around the goal-free fixed point}
\label{app:linearization-goal-free-fixed-point}

During goal-free behavior,
\[
I_i^{\mathrm{goal}}(g,t)=0.
\]
For fixed current position \(\boldsymbol{y}_t\), the background state \(\boldsymbol{r}^0\) satisfies
\begin{equation}
0
=
-r_i^0
+
g_s\sum_jS_{ij}\phi(r_j^0)
-
g_{\mathrm{inh}}\phi(r_i^0)
+
I_i^0,
\label{eq:app-background-fixed-point}
\end{equation}
where
\[
I_i^0=I_i^{\mathrm{bg}}+I_i^{\mathrm{EC}}(\boldsymbol{y}_t).
\]
The recall equations below are stated for fixed \(\boldsymbol{y}_t\), or for a quasi-static planning pause in which the change in \(\boldsymbol{r}^0(\boldsymbol{y}_t)\) is negligible over the relaxation time. If \(\boldsymbol{y}_t\) varies during the relaxation, the perturbation equation includes the additional forcing term \(-\tau_u\dot r_i^0(\boldsymbol{y}_t)\).
Define the goal-induced perturbation and raw input perturbation by
\[
u_i=r_i-r_i^0,
\qquad
\delta I_i^{\mathrm{raw}}=I_i(t)-I_i^0.
\]
Let
\[
\rho_i=\phi'(r_i^0).
\]
The Taylor estimates assume \(\phi\in C^2\) on the relevant bounded activity interval. For a piecewise-smooth nonlinearity, the same estimates apply on intervals that do not cross a kink or hard saturation boundary.
Taylor expansion gives
\[
\phi(r_i^0+u_i)
=
\phi(r_i^0)
+
\rho_i u_i
+
\frac12\phi''(r_i^0)u_i^2
+
\mathcal{O}(|u_i|^3).
\]
Substituting into Eq.~\eqref{eq:app-rate-raw} and subtracting Eq.~\eqref{eq:app-background-fixed-point} yields
\begin{equation}
\tau_u\dot u_i
=
-u_i
+
g_s\sum_jS_{ij}\rho_ju_j
-
g_{\mathrm{inh}}\rho_i u_i
+
\delta I_i^{\mathrm{raw}}
+
\mathcal{R}_i(\boldsymbol{u}),
\label{eq:app-exact-linearization}
\end{equation}
where
\begin{align}
\mathcal{R}_i(\boldsymbol{u})
&=
g_s\sum_jS_{ij}
\Big[
\phi(r_j^0+u_j)-\phi(r_j^0)-\rho_j u_j
\Big]
\nonumber\\
&\quad
-
g_{\mathrm{inh}}
\Big[
\phi(r_i^0+u_i)-\phi(r_i^0)-\rho_i u_i
\Big].
\label{eq:app-remainder-exact}
\end{align}
On any bounded activity interval there is a finite constant \(C_\phi\), depending on \(g_s\), \(g_{\mathrm{inh}}\), \(\boldsymbol{S}\), and the supremum of \(|\phi''|\) on that interval, such that
\begin{equation}
|\mathcal{R}_i(\boldsymbol{u})|
\le
C_\phi
\left(
|u_i|^2+\sum_j |S_{ij}|\,|u_j|^2
\right).
\label{eq:app-remainder-bound-raw}
\end{equation}

\subsection{Homogeneous-gain scalar-discount reduction}
\label{app:linearization-homogeneous-gain}

The exact scalar-discount theory is obtained when the background gain is homogeneous:
\[
\rho_i=\rho>0
\qquad
\text{for all }i.
\]
Then Eq.~\eqref{eq:app-exact-linearization} becomes
\[
\tau_u\dot{\boldsymbol{u}}
=
-(1+g_{\mathrm{inh}}\rho)\boldsymbol{u}
+
g_s\rho\,\boldsymbol{S}\boldsymbol{u}
+
\delta\boldsymbol{I}^{\mathrm{raw}}
+
\bm{\mathcal{R}}(\boldsymbol{u}).
\]
Define
\begin{equation}
c_\ell=1+g_{\mathrm{inh}}\rho,
\qquad
\tau=\frac{\tau_u}{c_\ell},
\qquad
\alpha=\frac{g_s\rho}{c_\ell},
\label{eq:app-alpha-definition}
\end{equation}
and
\begin{equation}
\delta\widetilde{\boldsymbol{I}}=\frac{\delta\boldsymbol{I}^{\mathrm{raw}}}{c_\ell},
\qquad
\widetilde{\bm{\mathcal{R}}}=\frac{\bm{\mathcal{R}}}{c_\ell}.
\label{eq:app-scaled-input}
\end{equation}
The homogeneous-gain dynamics are
\begin{equation}
\tau\dot{\boldsymbol{u}}
=
-(\boldsymbol{I}-\alpha \boldsymbol{S})\boldsymbol{u}+\delta\widetilde{\boldsymbol{I}}+\widetilde{\bm{\mathcal{R}}}(\boldsymbol{u}).
\label{eq:app-linear-plus-remainder}
\end{equation}
Dropping the nonlinear remainder gives
\begin{equation}
\tau\dot{\boldsymbol{u}}
=
-(\boldsymbol{I}-\alpha \boldsymbol{S})\boldsymbol{u}+\delta\widetilde{\boldsymbol{I}}.
\label{eq:app-linear-dynamics}
\end{equation}
The condition
\[
0<\alpha<1
\]
is the gain constraint
\begin{equation}
0<g_s\rho<1+g_{\mathrm{inh}}\rho.
\label{eq:app-gain-condition}
\end{equation}

Since \(\boldsymbol{S}\) is symmetric and similar to a Markov kernel, its spectrum lies in \([-1,1]\). Along an eigenvector of \(\boldsymbol{S}\) with eigenvalue \(\lambda_n(\boldsymbol{S})\), the homogeneous linearized dynamics have decay rate
\[
\frac{1-\alpha\lambda_n(\boldsymbol{S})}{\tau}.
\]
For \(\alpha\in(0,1)\),
\[
1-\alpha\lambda_n(\boldsymbol{S})\ge 1-\alpha>0.
\]
Thus the linearized fixed point is stable and \(\boldsymbol{I}-\alpha \boldsymbol{S}\) is invertible. Similarly,
\[
\|\boldsymbol{S}_{\mathcal I_g\mathcal I_g}\|_2\le1,
\]
so \(\boldsymbol{I}-\alpha \boldsymbol{S}_{\mathcal I_g\mathcal I_g}\) is invertible for every \(\alpha\in(0,1)\).

In the scaled variables, the nonlinear remainder satisfies
\begin{equation}
|\widetilde{\mathcal{R}}_i(\boldsymbol{u})|
\le
\widetilde C_\phi
\left(
|u_i|^2+\sum_j |S_{ij}|\,|u_j|^2
\right)
\label{eq:app-remainder-bound}
\end{equation}
for a finite constant \(\widetilde C_\phi\).

\subsection{Diagonal-gain operator}
\label{app:linearization-diagonal-gain}

A localized entorhinal drive can make the background \(r^0_i\) spatially nonuniform. For a saturating nonlinearity, this gives nonuniform gains \(\rho_i=\phi'(r_i^0)\). Define
\[
\boldsymbol{D}_\rho=\operatorname{diag}(\rho_1,\ldots,\rho_N),
\qquad
\boldsymbol{D}_\ell=\boldsymbol{I}+g_{\mathrm{inh}}\boldsymbol{D}_\rho.
\]
The first-order part of Eq.~\eqref{eq:app-exact-linearization} is
\begin{equation}
\tau_u\dot{\boldsymbol{u}}
=
-\bigl(\boldsymbol{D}_\ell-g_s\boldsymbol{S}\boldsymbol{D}_\rho\bigr)\boldsymbol{u}+\delta\boldsymbol{I}^{\mathrm{raw}}.
\label{eq:app-diagonal-gain}
\end{equation}
At steady state, or after using local diagonal time constants, the effective recurrent operator is
\begin{equation}
\boldsymbol{K}=g_s\boldsymbol{D}_\ell^{-1}\boldsymbol{S}\boldsymbol{D}_\rho.
\label{eq:app-effective-K}
\end{equation}
The scalar-discount theorem uses the homogeneous case
\[
\boldsymbol{K}=\alpha \boldsymbol{S}.
\]
For nonuniform gain, the scalar-discount equivalence is treated as a perturbation of this case, or as a state-dependent gain operator.

Let
\[
\boldsymbol{E}_K=\boldsymbol{K}-\alpha \boldsymbol{S}.
\]
If
\[
\|\boldsymbol{E}_K\|_2<1-\alpha,
\]
then \(\boldsymbol{I}-\boldsymbol{K}\) is invertible and
\begin{equation}
\|(\boldsymbol{I}-\boldsymbol{K})^{-1}-(\boldsymbol{I}-\alpha \boldsymbol{S})^{-1}\|_2
\le
\frac{\|\boldsymbol{E}_K\|_2}
{(1-\alpha)(1-\alpha-\|\boldsymbol{E}_K\|_2)}.
\label{eq:app-diagonal-gain-perturbation-bound}
\end{equation}

When \(\boldsymbol{S}\) is symmetric and \(0<\rho_{\min}\le\rho_i\le\rho_{\max}<\infty\), \(\boldsymbol{K}\) is diagonally symmetrizable. Indeed,
\[
K_{ij}
=
g_s\frac{\rho_j}{1+g_{\mathrm{inh}}\rho_i}S_{ij}.
\]
With
\[
\omega_i=\rho_i(1+g_{\mathrm{inh}}\rho_i),
\]
one has
\[
\omega_iK_{ij}
=
g_s\rho_i\rho_jS_{ij}
=
\omega_jK_{ji}.
\]
Thus \(\boldsymbol{K}\) is self-adjoint in the weighted inner product induced by \(\boldsymbol{\omega}\). It is not generally a stochastic reversible transition kernel, because its row sums and effective discount can vary across states.

\section{LMDP desirability and exact linear network equivalence}
\label{app:lmdp-derivation}

\subsection{Bellman equation with KL control cost}
\label{app:lmdp-kl-bellman}

At each non-goal state \(i\in\mathcal I_g\), the agent chooses a controlled next-state distribution
\[
\mu(\cdot\mid i)
\]
supported on the accessible successors of \(i\). The per-step cost is
\[
q_0
+
\lambda
D_{\mathrm{KL}}
\!\left(
\mu(\cdot\mid i)\,\Vert\,P_0(\cdot\mid i)
\right),
\]
where \(q_0>0\) and \(\lambda>0\). The terminal value is
\[
V(g)=0.
\]
For \(i\in\mathcal I_g\), the Bellman equation is
\begin{equation}
V(i)
=
q_0
+
\min_{\mu(\cdot\mid i)}
\left\{
\lambda
\sum_j\mu(j\mid i)\log\frac{\mu(j\mid i)}{P_0(j\mid i)}
+
\sum_j\mu(j\mid i)V(j)
\right\}.
\label{eq:app-lmdp-bellman}
\end{equation}
The Lagrange multiplier calculation for the constraint \(\sum_j\mu(j\mid i)=1\) gives
\begin{equation}
\mu^*(j\mid i)
=
\frac{
P_0(j\mid i)\exp[-V(j)/\lambda]
}{
\sum_kP_0(k\mid i)\exp[-V(k)/\lambda]
}.
\label{eq:app-lmdp-policy-v}
\end{equation}
The minimized bracket in Eq.~\eqref{eq:app-lmdp-bellman} is
\[
-\lambda\log\sum_jP_0(j\mid i)\exp[-V(j)/\lambda].
\]
Thus
\begin{equation}
V(i)
=
q_0 -
\lambda\log\sum_jP_0(j\mid i)\exp[-V(j)/\lambda].
\label{eq:app-soft-bellman}
\end{equation}

\subsection{Linear desirability equation}
\label{app:lmdp-linear-desirability}

Define
\[
z(i)=\exp[-V(i)/\lambda],
\qquad
\gamma=\exp(-q_0/\lambda).
\]
Exponentiating Eq.~\eqref{eq:app-soft-bellman} gives
\begin{equation}
z(i)
=
\gamma\sum_jP_0(j\mid i)z(j),
\qquad
i\in\mathcal I_g,
\label{eq:app-z-linear}
\end{equation}
with boundary condition
\[
z(g)=1.
\]
Splitting the sum into non-goal and goal terms yields
\begin{equation}
(\boldsymbol{I}-\gamma \boldsymbol{P}_{0,\mathcal I_g\mathcal I_g})\boldsymbol{z}_{\mathcal I_g}
=
\gamma \boldsymbol{p}_g,
\qquad
(\boldsymbol{p}_g)_i=P_0(g\mid i).
\label{eq:app-lmdp-dirichlet}
\end{equation}

The optimal policy in desirability coordinates is
\begin{equation}
\mu^*(j\mid i)
=
\frac{P_0(j\mid i)z(j)}
{\sum_kP_0(k\mid i)z(k)}
=
\frac{\gamma P_0(j\mid i)z(j)}{z(i)}.
\label{eq:app-lmdp-policy-z}
\end{equation}
On the finite goal-connected component, \(z(i)>0\) for all \(i\), and \(\mu^*(\cdot\mid i)\) has the same support as \(P_0(\cdot\mid i)\). Outside the goal-reachable component, \(z=0\) and the finite-value policy formula is not used.

Since \(\gamma\in(0,1)\),
\[
(\boldsymbol{I}-\gamma \boldsymbol{P}_{0,\mathcal I_g\mathcal I_g})^{-1}
=
\sum_{t\ge0}\gamma^t\boldsymbol{P}_{0,\mathcal I_g\mathcal I_g}^{t}.
\]
Therefore
\begin{equation}
\boldsymbol{z}_{\mathcal I_g}
=
\sum_{L\ge1}
\gamma^L
\boldsymbol{P}_{0,\mathcal I_g\mathcal I_g}^{L-1}\boldsymbol{p}_g.
\label{eq:app-lmdp-path-expansion}
\end{equation}
If \(\tau_g\) is the first hitting time of \(g\) under the passive chain, then
\begin{equation}
z(i)
=
\mathbb E_{\boldsymbol{P}_0,i}
\!\left[
\gamma^{\tau_g}\mathbb{1}_{\{\tau_g<\infty\}}
\right].
\label{eq:app-z-hitting}
\end{equation}
On the goal-connected finite component, \(\mathbb P_i(\tau_g<\infty)=1\).

\subsection{Exact linear network--LMDP theorem}
\label{app:lmdp-exact-linear-network}

\begin{theorem}[Linear network--LMDP equivalence]
\label{thm:app-linear-equivalence}
Assume that \(\boldsymbol{P}_0\) is finite, row-stochastic, nonnegative, irreducible on the goal-connected component, and reversible with respect to positive \(\boldsymbol{\pi}\). Let
\[
\boldsymbol{S}=\boldsymbol{\Pi}^{1/2}\boldsymbol{P}_0\boldsymbol{\Pi}^{-1/2}.
\]
Assume homogeneous linear gain, so the background-subtracted goal-induced network dynamics are governed by \(\boldsymbol{I}-\alpha \boldsymbol{S}\) with
\[
0<\alpha<1.
\]
During goal recall, impose a true Dirichlet clamp
\[
u_g=\bar u_g
\]
and let the non-goal assemblies solve
\begin{equation}
(\boldsymbol{I}-\alpha \boldsymbol{S}_{\mathcal I_g\mathcal I_g})\boldsymbol{u}_{\mathcal I_g}
=
\alpha \boldsymbol{S}_{\mathcal I_g g}\bar u_g.
\label{eq:app-network-dirichlet}
\end{equation}
Set
\[
\alpha=\gamma=\exp(-q_0/\lambda).
\]
Then the normalized network field satisfies
\begin{equation}
\psi_i^g
=
\frac{u_i}{\bar u_g}
=
\sqrt{\frac{\pi_i}{\pi_g}}\,z_i^g
\label{eq:app-equivalence-result}
\end{equation}
for every state in the goal-connected component, where \(z^g\) is the LMDP desirability field with boundary condition \(z_g^g=1\).
\end{theorem}

\begin{proof}
Define the symmetric-coordinate desirability
\[
\tilde z_i=\sqrt{\pi_i}\,z_i.
\]
Multiplying Eq.~\eqref{eq:app-lmdp-dirichlet} by \(\boldsymbol{\Pi}^{1/2}\) gives
\[
(\boldsymbol{I}-\gamma \boldsymbol{S}_{\mathcal I_g\mathcal I_g})\tilde{\boldsymbol{z}}_{\mathcal I_g}
=
\gamma\tilde{\boldsymbol{p}}_g,
\]
where
\[
(\tilde{\boldsymbol{p}}_g)_i=\sqrt{\pi_i}P_0(g\mid i).
\]
Using detailed balance,
\[
\sqrt{\pi_i}P_0(g\mid i)
=
\sqrt{\pi_g}
\sqrt{\frac{\pi_i}{\pi_g}}P_0(g\mid i)
=
\sqrt{\pi_g}S_{ig}.
\]
Hence
\begin{equation}
(\boldsymbol{I}-\gamma \boldsymbol{S}_{\mathcal I_g\mathcal I_g})\tilde{\boldsymbol{z}}_{\mathcal I_g}
=
\gamma\sqrt{\pi_g}\boldsymbol{S}_{\mathcal I_g g}.
\label{eq:app-symmetric-z-equation}
\end{equation}
Set
\[
u_i=\frac{\bar u_g}{\sqrt{\pi_g}}\tilde z_i
=
\bar u_g\sqrt{\frac{\pi_i}{\pi_g}}z_i.
\]
Then \(u_g=\bar u_g z_g=\bar u_g\), and Eq.~\eqref{eq:app-symmetric-z-equation} becomes Eq.~\eqref{eq:app-network-dirichlet}. Since \(\boldsymbol{I}-\alpha \boldsymbol{S}_{\mathcal I_g\mathcal I_g}\) is invertible, the solution is unique.
\end{proof}

A source-normalized full-resolvent representation is equivalent to the Dirichlet clamp. Let
\[
\boldsymbol{R}=(\boldsymbol{I}-\alpha \boldsymbol{S})^{-1}.
\]
Solving
\[
(\boldsymbol{I}-\alpha \boldsymbol{S})\tilde{\boldsymbol{u}}^{(g)}=c_g \boldsymbol{e}_g
\]
gives
\[
\tilde{\boldsymbol{u}}^{(g)}=c_g \boldsymbol{R}\boldsymbol{e}_g.
\]
Choosing
\[
c_g=\frac{\bar u_g}{R_{gg}}
\]
enforces \(\tilde{\boldsymbol{u}}^{(g)}_g=\bar u_g\). For \(i\ne g\), the equation has no source term, so
\[
(\boldsymbol{I}-\alpha \boldsymbol{S}_{\mathcal I_g\mathcal I_g})\tilde{\boldsymbol{u}}^{(g)}_{\mathcal I_g}
=
\alpha \boldsymbol{S}_{\mathcal I_g g}\bar u_g,
\]
which is the same Dirichlet field.

\subsection{Directed transitions: diagonal similarity}
\label{app:lmdp-directed-corollary}

Theorem~\ref{thm:equiv} uses reversibility to obtain a symmetric recurrent matrix. The identity between the clamped steady state and the desirability field does not need symmetry; it needs the recurrent matrix to be similar to the passive kernel through a positive diagonal matrix. This subsection proves Corollary~\ref{cor:directed} of the main text.

Let \(\boldsymbol{P}_0\) be any row-stochastic kernel on the \(N\) states, reversible or not, and let \(\gamma\in(0,1)\). The desirability equation
\[
z^g(i)=\gamma\sum_j P_0(j\mid i)\,z^g(j)\quad(i\ne g),\qquad z^g(g)=1,
\]
has a unique solution because the spectral radius of \(\gamma\boldsymbol{P}_{0,\mathcal I_g\mathcal I_g}\) is at most \(\gamma<1\). Let \(\boldsymbol{D}=\operatorname{diag}(d_1,\dots,d_N)\) with \(d_i>0\), and define
\[
\boldsymbol{S}:=\boldsymbol{D}\boldsymbol{P}_0\boldsymbol{D}^{-1},\qquad S_{ij}=\frac{d_i}{d_j}\,P_0(j\mid i).
\]
Set \(\alpha=\gamma\) and define the candidate field \(u_i:=\bar u_g\,(d_i/d_g)\,z^g(i)\), so that \(u_g=\bar u_g\). For \(i\ne g\),
\[
\alpha\sum_j S_{ij}u_j
=\gamma\sum_j \frac{d_i}{d_j}P_0(j\mid i)\,\bar u_g\frac{d_j}{d_g}z^g(j)
=\bar u_g\frac{d_i}{d_g}\,\gamma\sum_j P_0(j\mid i)z^g(j)
=\bar u_g\frac{d_i}{d_g}\,z^g(i)=u_i .
\]
Hence \(\boldsymbol{u}\) solves the clamped steady-state equation \((\boldsymbol{I}-\alpha\boldsymbol{S}_{\mathcal I_g\mathcal I_g})\boldsymbol{u}_{\mathcal I_g}=\alpha\boldsymbol{S}_{\mathcal I_g g}\bar u_g\) of Eq.~\eqref{eq:dirichlet}. The solution is unique because \(\boldsymbol{S}_{\mathcal I_g\mathcal I_g}\) is similar to \(\boldsymbol{P}_{0,\mathcal I_g\mathcal I_g}\) and therefore shares its spectrum, so the spectral radius of \(\alpha\boldsymbol{S}_{\mathcal I_g\mathcal I_g}\) is below one. This proves \(u_i/\bar u_g=(d_i/d_g)z^g(i)\).

For the policy, \(P_0(j\mid i)=(d_j/d_i)S_{ij}\) and \(z^g(j)=(d_g/\bar u_g)\,u_j/d_j\), so
\[
\mu^*(j\mid i)\;\propto\;P_0(j\mid i)\,z^g(j)=\frac{d_g}{d_i\,\bar u_g}\,S_{ij}\,u_j\;\propto\;S_{ij}\,u_j,
\]
because the prefactor does not depend on \(j\). The optimal next-step distribution is the normalized vector of branch currents \(S_{ij}u_j\) arriving at unit \(i\), exactly as in the reversible case. The choice \(\boldsymbol{D}=\boldsymbol{\Pi}^{1/2}\) for a reversible kernel gives the symmetric matrix of Eq.~\eqref{eq:S-from-P0} and recovers Theorem~\ref{thm:equiv}; the choice \(\boldsymbol{D}=\boldsymbol{I}\) gives \(S_{ij}=P_0(j\mid i)\) and \(u_i/\bar u_g=z^g(i)\) with no occupancy factor. App.~\ref{app:exp-extended-directed} verifies the corollary numerically for three positive diagonals, including \(\boldsymbol{D}=\boldsymbol{I}\). Because the readout depends on the field only through \(\nabla\log\), a non-uniform \(\boldsymbol{D}\) adds the goal-independent bias \(\nabla\log d\) to the readout direction, the analogue of the \(\tfrac12\nabla\log\pi\) bias of App.~\ref{app:asymmetry-readout-bias}.

\paragraph{Synaptic orientation.} In the paper's convention \(S_{ij}\) is the synapse from unit \(j\) onto unit \(i\). The corollary stores the legal move \(i\to j\) in the synapse \(j\to i\), so the clamped activity propagates from the goal backward along legal routes. A causal plasticity kernel, which potentiates a synapse when its presynaptic unit fires before the postsynaptic unit, writes \(j\to i\) whenever a sequence plays \(j\) and then \(i\). Reverse replay of the executed trajectory \(i\to j\) plays exactly this order and therefore writes \(\boldsymbol{S}\); forward replay of the same trajectory writes the transpose \(\boldsymbol{S}^{\!\top}\), the predictive operator whose resolvent column from the current state is the successor representation of Sec.~\ref{sec:readout}. With a symmetric kernel and bidirectional sampling the two coincide, which is the reversible case of the main text.

\section{Reversibility, learned transition kernels, and normalization}
\label{app:reversibility}

\subsection{Symmetric network coordinate}
\label{app:reversibility-symmetric-coordinate}

Assume detailed balance:
\[
\pi_iP_0(j\mid i)=\pi_jP_0(i\mid j).
\]
With
\[
\boldsymbol{S}=\boldsymbol{\Pi}^{1/2}\boldsymbol{P}_0\boldsymbol{\Pi}^{-1/2},
\]
one has
\[
S_{ij}
=
\sqrt{\frac{\pi_i}{\pi_j}}P_0(j\mid i).
\]
Then
\[
S_{ji}
=
\sqrt{\frac{\pi_j}{\pi_i}}P_0(i\mid j)
=
\sqrt{\frac{\pi_i}{\pi_j}}P_0(j\mid i)
=S_{ij}.
\]
Thus \(\boldsymbol{S}\) is symmetric. This is the recurrent coordinate used in Theorem~\ref{thm:app-linear-equivalence}.

\subsection{Transition-gated symmetric accumulation}
\label{app:reversibility-transition-gated-symmetric-accumulation}

Let \(a_i(t)\ge0\) be the activity of assembly \(i\). The transition-gated accumulation rule is
\begin{equation}
W_{ij}^{\mathrm{tr}}
=
\sum_\tau k(\tau)
\mathbb E
\!\left[
a_i(t)a_j(t+\tau)\Xi_{ij}(t,\tau)
\right],
\label{eq:app-btsp-rule}
\end{equation}
where
\[
k(\tau)=k(-\tau),
\qquad
k(\tau)\ge0.
\]
The gate \(\Xi_{ij}(t,\tau)\ge0\) marks an accessible transition between assemblies \(i\) and \(j\) inside the plasticity window. Physical movement and internally generated sequences enter through the same transition gate. The nonnegativity of \(a_i\), \(k\), and \(\Xi\) gives
\[
W_{ij}^{\mathrm{tr}}\ge0.
\]
Define the lagged gated correlation
\[
C_{ij}(\tau)
=
\mathbb E
\!\left[
a_i(t)a_j(t+\tau)\Xi_{ij}(t,\tau)
\right].
\]
The ideal reversible sampling condition is
\begin{equation}
C_{ij}(\tau)=C_{ji}(-\tau)
\label{eq:app-time-reversal-condition}
\end{equation}
for every accessible unordered edge \(\{i,j\}\) and lag \(\tau\). This condition states joint invariance of the activity pair and transition gate under reversal of the sampled edge. Under Eq.~\eqref{eq:app-time-reversal-condition},
\begin{align}
W_{ji}^{\mathrm{tr}}
&=
\sum_\tau k(\tau)C_{ji}(\tau)
\nonumber\\
&=
\sum_\tau k(\tau)C_{ij}(-\tau)
\nonumber\\
&=
\sum_{\tau'}k(-\tau')C_{ij}(\tau')
\nonumber\\
&=
\sum_{\tau'}k(\tau')C_{ij}(\tau')
=
W_{ij}^{\mathrm{tr}}.
\label{eq:app-W-symmetric}
\end{align}
Thus \(\boldsymbol{W}^{\mathrm{tr}}\) is symmetric in expectation under stationary bidirectional sampling with a symmetric nonnegative temporal kernel and a jointly time-reversal-invariant transition gate.
Directed or policy-biased sequence statistics enter as deviations from this symmetrized effective transition memory and are handled by the perturbative nonreversible terms in App.~\ref{app:asymmetry-irreversible-sampling}.

If a transition is blocked, the gate is absent for that pair, and
\[
W_{ij}^{\mathrm{tr}}=0.
\]
The support of \(\boldsymbol{W}^{\mathrm{tr}}\) is therefore the accessible transition graph.

\subsection{Lazy normalization}
\label{app:reversibility-lazy-normalization}

Let
\[
d_i=\sum_{j\ne i}W_{ij}^{\mathrm{tr}}.
\]
Choose a fixed global \(\epsilon\) satisfying
\begin{equation}
0<\epsilon\le \frac{1}{\max_i d_i}.
\label{eq:app-epsilon-condition}
\end{equation}
The lazy normalization is
\begin{equation}
P_0(j\mid i)=\epsilon W_{ij}^{\mathrm{tr}}
\quad (j\ne i),
\qquad
P_0(i\mid i)=1-\epsilon d_i.
\label{eq:app-lazy-normalization}
\end{equation}
Every row sums to one. If \(\boldsymbol{W}^{\mathrm{tr}}\) is symmetric, then
\[
P_0(j\mid i)=P_0(i\mid j)
\]
for all \(i,j\), so \(\boldsymbol{P}_0\) is symmetric and reversible with uniform stationary measure
\[
\pi_i=\frac1N.
\]
In this regime,
\[
\boldsymbol{S}=\boldsymbol{P}_0,
\qquad
\sqrt{\boldsymbol{\pi}}\propto\boldsymbol{1}.
\]
For local environmental edits, \(\epsilon\) is kept fixed across the pre-edit and post-edit graphs and chosen to satisfy the laziness condition for both graphs. If \(\epsilon\) were globally reset after an edit, all off-diagonal weights could be rescaled, producing a nonlocal operator change.

\subsection{Standard random-walk normalization}
\label{app:reversibility-random-walk-normalization}

Another common normalization is
\begin{equation}
P_0(j\mid i)=\frac{W_{ij}^{\mathrm{tr}}}{d_i}
\qquad
(d_i>0).
\label{eq:app-rw-normalization}
\end{equation}
For isolated states with \(d_i=0\), one either removes the state from the navigable component or sets \(P_0(i\mid i)=1\) as a closed singleton component. On a connected component with \(d_i>0\), symmetric \(\boldsymbol{W}^{\mathrm{tr}}\) gives detailed balance with
\[
\pi_i=\frac{d_i}{\sum_\ell d_\ell}.
\]
The symmetric recurrent coordinate is
\begin{equation}
S_{ij}
=
\sqrt{\frac{\pi_i}{\pi_j}}P_0(j\mid i)
=
\frac{W_{ij}^{\mathrm{tr}}}{\sqrt{d_i d_j}}.
\label{eq:app-rw-symmetric-S}
\end{equation}

\section{Finite sampling, nonuniform \texorpdfstring{\(\boldsymbol{\pi}\)}{pi}, and nonreversible perturbations}
\label{app:asymmetry}

\subsection{Reversible projection of empirical transition counts}
\label{app:asymmetry-reversible-projection}

Finite sampling can create asymmetric empirical transition counts. Let \(\widehat{\boldsymbol{W}}\) denote the empirical nonnegative count matrix. A reversible projection is obtained by
\begin{equation}
\widehat{\boldsymbol{W}}^{\mathrm{sym}}
=
\frac12(\widehat{\boldsymbol{W}}+\widehat{\boldsymbol{W}}^\top).
\label{eq:app-sym-counts}
\end{equation}
Applying the lazy normalization of Eq.~\eqref{eq:app-lazy-normalization} to \(\widehat{\boldsymbol{W}}^{\mathrm{sym}}\) gives an exactly symmetric passive kernel and hence uniform \(\boldsymbol{\pi}\). Applying the random-walk normalization to \(\widehat{\boldsymbol{W}}^{\mathrm{sym}}\) gives an exactly reversible passive kernel with \(\pi_i\propto \widehat d_i\). The exact theorem applies to these reversible projected kernels.

\subsection{Perturbation by irreversible sampling error}
\label{app:asymmetry-irreversible-sampling}

Let \(\bar{\boldsymbol{P}}\) be an ideal reversible kernel with stationary measure \(\bar{\boldsymbol{\pi}}\), and let
\[
\bar{\boldsymbol{S}}=\bar{\boldsymbol{\Pi}}^{1/2}\bar{\boldsymbol{P}}\bar{\boldsymbol{\Pi}}^{-1/2}
\]
be its symmetric coordinate. Suppose a raw empirical kernel has the form
\[
\boldsymbol{P}^{\mathrm{emp}}=\bar{\boldsymbol{P}}+\Delta\boldsymbol{P},
\qquad
\Delta\boldsymbol{P}\,\boldsymbol{1}=0.
\]
In the \(\bar{\boldsymbol{\pi}}\)-symmetric coordinate define
\begin{equation}
\boldsymbol{E}=\bar{\boldsymbol{\Pi}}^{1/2}\Delta\boldsymbol{P}\,\bar{\boldsymbol{\Pi}}^{-1/2}.
\label{eq:app-empirical-E}
\end{equation}
Then the raw empirical recurrent operator in this coordinate is
\[
\boldsymbol{S}^{\mathrm{emp}}=\bar{\boldsymbol{S}}+\boldsymbol{E}.
\]
Decompose
\[
\boldsymbol{E}=\boldsymbol{E}_{\mathrm{sym}}+\boldsymbol{E}_{\mathrm{asym}},
\qquad
\boldsymbol{E}_{\mathrm{sym}}=\frac12(\boldsymbol{E}+\boldsymbol{E}^\top),
\qquad
\boldsymbol{E}_{\mathrm{asym}}=\frac12(\boldsymbol{E}-\boldsymbol{E}^\top).
\]
The antisymmetric part \(\boldsymbol{E}_{\mathrm{asym}}\) measures irreversible circulation in the empirical transition statistics.

For the full resolvent, if
\[
\alpha\|\boldsymbol{E}\|_2<1-\alpha,
\]
then
\begin{equation}
\|(\boldsymbol{I}-\alpha(\bar{\boldsymbol{S}}+\boldsymbol{E}))^{-1}-(\boldsymbol{I}-\alpha\bar{\boldsymbol{S}})^{-1}\|_2
\le
\frac{\alpha\|\boldsymbol{E}\|_2}
{(1-\alpha)(1-\alpha-\alpha\|\boldsymbol{E}\|_2)}.
\label{eq:app-nonrev-full-bound}
\end{equation}
The same bound holds for Dirichlet restrictions with \(\boldsymbol{E}\) replaced by \(\boldsymbol{E}_{\mathcal I_g\mathcal I_g}\). Thus finite sampling asymmetry and irreversible circulation give controlled corrections to the recurrent field when their operator norm is small relative to the stability margin \(1-\alpha\). Exact network--LMDP equivalence with the symmetric recurrent coordinate is retained for the reversible projection \(\bar{\boldsymbol{P}}\).

For a clamped goal, write
\[
\boldsymbol{A}_0=\boldsymbol{I}-\alpha\bar{\boldsymbol{S}}_{\mathcal I_g\mathcal I_g},
\qquad
\boldsymbol{A}_E=\boldsymbol{I}-\alpha(\bar{\boldsymbol{S}}+\boldsymbol{E})_{\mathcal I_g\mathcal I_g}.
\]
If \(\boldsymbol{b}_0\) is the reversible boundary source and \(\boldsymbol{b}_E=\boldsymbol{b}_0+\delta\boldsymbol{b}\) is the empirical boundary source, then
\[
\boldsymbol{u}_E=\boldsymbol{A}_E^{-1}\boldsymbol{b}_E,
\qquad
\boldsymbol{u}_0=\boldsymbol{A}_0^{-1}\boldsymbol{b}_0,
\]
and
\begin{equation}
\boldsymbol{u}_E-\boldsymbol{u}_0
=
(\boldsymbol{A}_E^{-1}-\boldsymbol{A}_0^{-1})\boldsymbol{b}_0
+
\boldsymbol{A}_E^{-1}\delta\boldsymbol{b}.
\label{eq:app-nonrev-field-difference}
\end{equation}
Equation~\eqref{eq:app-nonrev-full-bound} controls the first term, and
\[
\|\boldsymbol{A}_E^{-1}\|_2\le\frac{1}{1-\alpha-\alpha\|\boldsymbol{E}\|_2}
\]
controls the second.

\subsection{Perturbation of a reversible stationary measure}
\label{app:asymmetry-stationary-measure}

When the learned kernel is kept reversible and the stationary measure changes by sampling or normalization, the change in \(\boldsymbol{\pi}\) is controlled by the spectral gap. Let
\[
\boldsymbol{P}_\varepsilon=\bar{\boldsymbol{P}}+\Delta\boldsymbol{P}
\]
be row-stochastic and reversible up to the retained order, with stationary measure
\[
\boldsymbol{\pi}=\bar{\boldsymbol{\pi}}+\delta\boldsymbol{\pi}.
\]
Use the row-vector norm
\[
\|\boldsymbol{r}\|_{\bar{\boldsymbol{\pi}}^{-1}}
=
\left(\sum_i\frac{r_i^2}{\bar\pi_i}\right)^{1/2}
\]
and its compatible operator norm. The stationarity equations give, to first order,
\[
\delta\boldsymbol{\pi}(\boldsymbol{I}-\bar{\boldsymbol{P}})=\bar{\boldsymbol{\pi}}\Delta\boldsymbol{P},
\qquad
\delta\boldsymbol{\pi}\,\boldsymbol{1}=0.
\]
On the zero-sum subspace, the reversible chain has inverse norm bounded by the reciprocal spectral gap
\[
\operatorname{gap}=1-\lambda_2(\bar{\boldsymbol{P}})>0.
\]
Thus
\begin{equation}
\delta\boldsymbol{\pi}
=
\bar{\boldsymbol{\pi}}\Delta\boldsymbol{P}\,(\boldsymbol{I}-\bar{\boldsymbol{P}})^\#
+
\mathcal{O}(|\Delta\boldsymbol{P}|^2),
\label{eq:app-pi-perturbation}
\end{equation}
and, if \(\boldsymbol{P}_\varepsilon\) remains irreducible and
\[
\|\Delta\boldsymbol{P}\|_{\bar{\boldsymbol{\pi}}^{-1}\to\bar{\boldsymbol{\pi}}^{-1}}<\operatorname{gap},
\]
\begin{equation}
\|\boldsymbol{\pi}-\bar{\boldsymbol{\pi}}\|_{\bar{\boldsymbol{\pi}}^{-1}}
\le
\frac{\|\bar{\boldsymbol{\pi}}\Delta\boldsymbol{P}\|_{\bar{\boldsymbol{\pi}}^{-1}}}
{\operatorname{gap}-\|\Delta\boldsymbol{P}\|_{\bar{\boldsymbol{\pi}}^{-1}\to\bar{\boldsymbol{\pi}}^{-1}}}.
\label{eq:app-pi-bound}
\end{equation}

\subsection{Readout bias from nonuniform stationary measure}
\label{app:asymmetry-readout-bias}

For a reversible kernel with nonuniform \(\boldsymbol{\pi}\), Theorem~\ref{thm:app-linear-equivalence} gives
\[
\psi_i^g
=
\sqrt{\frac{\pi_i}{\pi_g}}\,z_i^g.
\]
Taking logarithms,
\begin{equation}
\log\psi_i^g
=
\log z_i^g
+
\frac12\log\pi_i
-
\frac12\log\pi_g.
\label{eq:app-log-pi}
\end{equation}
For an accessible edge \(e=(i,j)\),
\begin{equation}
\nabla_e\log\psi^g
=
\nabla_e\log z^g
+
\frac12\nabla_e\log\pi,
\label{eq:app-edge-pi-bias}
\end{equation}
where
\[
\nabla_e\log\pi=\log\pi_j-\log\pi_i.
\]
In a local embedded approximation, this edge relation corresponds to
\begin{equation}
\nabla\log\psi^g
=
\nabla\log z^g
+
\frac12\nabla\log\pi.
\label{eq:app-pi-bias}
\end{equation}

Let
\[
\boldsymbol{v}_0=\boldsymbol{\Sigma}_G\nabla\log z^g,
\qquad
\boldsymbol{b}_\pi=\frac12\boldsymbol{\Sigma}_G\nabla\log\pi.
\]
When \(\boldsymbol{v}_0\ne0\) and \(\|\boldsymbol{b}_\pi\|<\|\boldsymbol{v}_0\|\),
\begin{equation}
\sin\angle(\boldsymbol{v}_0+\boldsymbol{b}_\pi,\boldsymbol{v}_0)
\le
\frac{\|\boldsymbol{b}_\pi\|}{\|\boldsymbol{v}_0+\boldsymbol{b}_\pi\|}
\le
\frac{\|\boldsymbol{b}_\pi\|}{\|\boldsymbol{v}_0\|-\|\boldsymbol{b}_\pi\|}.
\label{eq:app-pi-angle-bound}
\end{equation}
An exact corrected readout in the network coordinate uses
\[
\frac{\psi_i^g}{\sqrt{\pi_i}}
=
\frac{z_i^g}{\sqrt{\pi_g}},
\]
whose log-gradient is exactly \(\nabla\log z^g\).

\section{Local place-kernel readout}
\label{app:popvec}

\subsection{Continuous local moment expansion}
\label{app:popvec-continuous-moment}

Let
\[
f(\boldsymbol{x})=\psi^g(\boldsymbol{x})>0
\]
or, under nonuniform \(\boldsymbol{\pi}\), the corrected field
\[
f(\boldsymbol{x})=\frac{\psi^g(\boldsymbol{x})}{\sqrt{\pi(\boldsymbol{x})}}.
\]
Fix the animal's current position \(\boldsymbol{y}\), and write
\[
\boldsymbol{\xi}=\boldsymbol{x}-\boldsymbol{y}.
\]
The local population-vector readout is
\begin{equation}
\Delta\boldsymbol{x}_{\mathrm{pop}}(\boldsymbol{y};f)
=
\frac{
\int \boldsymbol{\xi}\,\mathcal{G}_{\boldsymbol{y}}(\boldsymbol{y}+\boldsymbol{\xi})f(\boldsymbol{y}+\boldsymbol{\xi})\,d\boldsymbol{\xi}
}{
\int \mathcal{G}_{\boldsymbol{y}}(\boldsymbol{y}+\boldsymbol{\xi})f(\boldsymbol{y}+\boldsymbol{\xi})\,d\boldsymbol{\xi}
}.
\label{eq:app-popvec-start}
\end{equation}
Let
\[
M_0(\boldsymbol{y})=\int \mathcal{G}_{\boldsymbol{y}}(\boldsymbol{y}+\boldsymbol{\xi})\,d\boldsymbol{\xi},
\]
and define expectation under the normalized place kernel:
\[
\mathbb E_{\mathcal G}[h(\boldsymbol{\xi})]
=
\frac{1}{M_0(\boldsymbol{y})}
\int h(\boldsymbol{\xi})\mathcal{G}_{\boldsymbol{y}}(\boldsymbol{y}+\boldsymbol{\xi})\,d\boldsymbol{\xi}.
\]
The kernel mean and covariance are
\[
\boldsymbol{\mu}_G(\boldsymbol{y})=\mathbb E_{\mathcal G}[\boldsymbol{\xi}],
\]
\[
\boldsymbol{\Sigma}_G(\boldsymbol{y})
=
\mathbb E_{\mathcal G}[(\boldsymbol{\xi}-\boldsymbol{\mu}_G)(\boldsymbol{\xi}-\boldsymbol{\mu}_G)^\top].
\]
Expanding
\[
f(\boldsymbol{y}+\boldsymbol{\xi})
=
f(\boldsymbol{y})
+
\nabla f(\boldsymbol{y})^\top\boldsymbol{\xi}
+
\frac12\boldsymbol{\xi}^\top\nabla^2f(\boldsymbol{y})\boldsymbol{\xi}
+
\mathcal{O}(\|\boldsymbol{\xi}\|^3)
\]
gives, to first order in the local field variation,
\begin{equation}
\Delta\boldsymbol{x}_{\mathrm{pop}}(\boldsymbol{y};f)
=
\boldsymbol{\mu}_G(\boldsymbol{y})
+
\boldsymbol{\Sigma}_G(\boldsymbol{y})\nabla\log f(\boldsymbol{y})
+
\text{higher-order kernel-width terms}.
\label{eq:app-popvec-general-mean}
\end{equation}
The term \(\boldsymbol{\mu}_G(\boldsymbol{y})\) is the resting displacement produced by the local kernel alone:
\[
\Delta\boldsymbol{x}_{\mathrm{pop}}(\boldsymbol{y};1)=\boldsymbol{\mu}_G(\boldsymbol{y}).
\]
The goal-dependent centered readout is therefore
\begin{equation}
\Delta\boldsymbol{x}_{\mathrm{goal}}(\boldsymbol{y};f)
=
\Delta\boldsymbol{x}_{\mathrm{pop}}(\boldsymbol{y};f)-\Delta\boldsymbol{x}_{\mathrm{pop}}(\boldsymbol{y};1)
=
\boldsymbol{\Sigma}_G(\boldsymbol{y})\nabla\log f(\boldsymbol{y})
+
\text{higher-order terms}.
\label{eq:app-centered-popvec}
\end{equation}
When the local kernel is centered,
\[
\boldsymbol{\mu}_G(\boldsymbol{y})=0,
\]
the raw and centered readouts agree to leading order.

\subsection{Centered-kernel correction}
\label{app:popvec-centered-kernel}

For a centered kernel, define
\[
\boldsymbol{\Sigma}_G(\boldsymbol{y})
=
\frac{1}{M_0(\boldsymbol{y})}
\int \boldsymbol{\xi}\boldsymbol{\xi}^\top \mathcal{G}_{\boldsymbol{y}}(\boldsymbol{y}+\boldsymbol{\xi})\,d\boldsymbol{\xi}.
\]
Then
\begin{equation}
\Delta\boldsymbol{x}_{\mathrm{pop}}(\boldsymbol{y};f)
=
\boldsymbol{\Sigma}_G(\boldsymbol{y})\nabla\log f(\boldsymbol{y})
+
\text{higher-order terms}.
\label{eq:app-popvec-leading}
\end{equation}

Using repeated-index summation over spatial indices, define the third moment tensor
\[
T_{abc}(\boldsymbol{y})
=
\frac{1}{M_0(\boldsymbol{y})}
\int
\xi_a\xi_b\xi_c\mathcal{G}_{\boldsymbol{y}}(\boldsymbol{y}+\boldsymbol{\xi})\,d\boldsymbol{\xi}.
\]
For a centered kernel,
\begin{align}
[\Delta\boldsymbol{x}_{\mathrm{pop}}]_a
&=
[\boldsymbol{\Sigma}_G\nabla\log f]_a
+
\frac12
T_{abc}
\frac{\partial_b\partial_c f}{f}
\nonumber\\
&\quad
-
\frac12
[\boldsymbol{\Sigma}_G\nabla\log f]_a
(\boldsymbol{\Sigma}_G)_{bc}
\frac{\partial_b\partial_c f}{f}
+
\mathcal{O}(\sigma_G^4/\ell_f^3),
\label{eq:app-popvec-correction}
\end{align}
where \(\sigma_G\) is the local kernel-width scale and \(\ell_f\) is the local length scale on which \(f\) varies. If the kernel has local reflection symmetry, then \(T_{abc}=0\), and the first shape correction is fourth order in the kernel width.

\subsection{Reading the raw rate}
\label{app:popvec-raw-rate}

The readout of Eq.~\eqref{eq:popvec} is written in terms of the deviation \(\psi^g\propto u=r-r^0\). The following lemma, Proposition~\ref{prop:raw-rate} of the main text, shows that a readout cell can multiply its place kernel by the raw rate \(r\) without any subtraction of the resting state, and that the population-vector direction is exactly unchanged.

\paragraph{Lemma (centered kernel).} Let \(\mathcal{G}:=\mathcal{G}_{\boldsymbol{y}_t}\) be a nonnegative kernel with mass \(M_0:=\int\mathcal{G}(\boldsymbol{x})\,d\boldsymbol{x}\) and vanishing first moment about \(\boldsymbol{y}_t\), \(\int(\boldsymbol{x}-\boldsymbol{y}_t)\mathcal{G}(\boldsymbol{x})\,d\boldsymbol{x}=\boldsymbol{0}\). For any field \(f\ge0\) with \(\int\mathcal{G}f>0\) and any constant \(c\ge0\),
\[
\Delta\boldsymbol{x}_{\mathrm{pop}}[c+f]
:=\frac{\int(\boldsymbol{x}-\boldsymbol{y}_t)\,\mathcal{G}(\boldsymbol{x})\,\bigl(c+f(\boldsymbol{x})\bigr)\,d\boldsymbol{x}}{\int\mathcal{G}(\boldsymbol{x})\,\bigl(c+f(\boldsymbol{x})\bigr)\,d\boldsymbol{x}}
=\frac{\int\mathcal{G}f}{cM_0+\int\mathcal{G}f}\;\Delta\boldsymbol{x}_{\mathrm{pop}}[f].
\]

\paragraph{Proof.} The numerator splits as \(c\int(\boldsymbol{x}-\boldsymbol{y}_t)\mathcal{G}+\int(\boldsymbol{x}-\boldsymbol{y}_t)\mathcal{G}f\); the first term vanishes by centering. The denominator is \(cM_0+\int\mathcal{G}f\). Dividing numerator and denominator by \(\int\mathcal{G}f\) gives the stated factor times \(\Delta\boldsymbol{x}_{\mathrm{pop}}[f]\). \(\square\)

The factor lies in \((0,1]\) and depends on the goal only through the local mass \(\int\mathcal{G}f\), so the direction of the population vector computed from the raw rate \(r=c+u\) equals the direction computed from \(u\), and only the magnitude is reduced by the resting contribution. Setting \(f=\psi^g\) and \(c=r^0/\bar u_g\) gives Proposition~\ref{prop:raw-rate}. Because the result holds exactly for every centered kernel, it holds for the anisotropic kernels of App.~\ref{app:popvec-accessible-directions} and for the discrete population sums used in the experiments, where centering is enforced by subtracting the goal-free population vector (App.~\ref{app:exp-extended-readout}).

\paragraph{Position-dependent resting profile.} If the resting rate varies with position, \(r^0(\boldsymbol{x})=c+\varrho(\boldsymbol{x})\), the numerator acquires the additional term \(\int(\boldsymbol{x}-\boldsymbol{y}_t)\mathcal{G}\varrho\), the population vector of the resting profile. This term does not depend on the goal. At rest the entorhinal bump that sets \(\varrho\) is centered on \(\boldsymbol{y}_t\), so for a kernel with reflection symmetry about \(\boldsymbol{y}_t\) the term vanishes; in general it is a fixed offset measured once in the absence of recall, and the centered readout of App.~\ref{app:exp-extended-readout} subtracts it. In every case the goal-dependent part of the population vector is \(\int(\boldsymbol{x}-\boldsymbol{y}_t)\mathcal{G}u\), the same quantity that enters Eq.~\eqref{eq:popvec}.

\subsection{Accessible directions}
\label{app:popvec-accessible-directions}

The covariance \(\boldsymbol{\Sigma}_G(\boldsymbol{y})\) records the directions sampled by the local place kernel. In open space it is close to isotropic:
\[
\boldsymbol{\Sigma}_G(\boldsymbol{y})\approx\sigma_G^2\boldsymbol{I}.
\]
Near a wall, variance in the blocked direction is suppressed. Inside a narrow corridor, the covariance is elongated along the corridor axis. The goal-dependent vector
\[
\boldsymbol{\Sigma}_G(\boldsymbol{y})\nabla\log f(\boldsymbol{y})
\]
therefore projects the local log-gradient onto directions represented in the local place-kernel support. A nonzero kernel mean \(\boldsymbol{\mu}_G\) contributes a resting drift term, and the centered readout of Eq.~\eqref{eq:app-centered-popvec} isolates the goal-dependent component.

\section{Relation between the readout and the LMDP optimal policy}
\label{app:popvec-lmdp}

At state \(i\), the LMDP optimal next-step distribution is
\[
\mu^*(j\mid i)
=
\frac{P_0(j\mid i)z(j)}
{\sum_kP_0(k\mid i)z(k)}.
\]
Let
\[
\boldsymbol{\xi}_j=\boldsymbol{x}_j-\boldsymbol{x}_i.
\]
The optimal-policy mean displacement is
\begin{equation}
\boldsymbol{m}^*(i,g)
=
\sum_j\boldsymbol{\xi}_j\mu^*(j\mid i).
\label{eq:app-policy-mean-def}
\end{equation}

In the uniform-\(\boldsymbol{\pi}\) case, \(z=\psi^g\), and
\begin{equation}
\boldsymbol{m}^*(i,g)
=
\frac{
\sum_j\boldsymbol{\xi}_jP_0(j\mid i)\psi^g(j)
}{
\sum_jP_0(j\mid i)\psi^g(j)
}.
\label{eq:app-policy-mean}
\end{equation}
Define the passive local mean and covariance
\[
\bar{\boldsymbol{\xi}}_i=\sum_jP_0(j\mid i)\boldsymbol{\xi}_j,
\]
\[
\boldsymbol{\Sigma}_{\boldsymbol{P}_0}(i)
=
\sum_jP_0(j\mid i)(\boldsymbol{\xi}_j-\bar{\boldsymbol{\xi}}_i)(\boldsymbol{\xi}_j-\bar{\boldsymbol{\xi}}_i)^\top.
\]
A local expansion gives
\begin{equation}
\boldsymbol{m}^*(i,g)
=
\bar{\boldsymbol{\xi}}_i
+
\boldsymbol{\Sigma}_{\boldsymbol{P}_0}(i)\nabla\log\psi^g(\boldsymbol{x}_i)
+
\text{higher-order local terms}.
\label{eq:app-policy-gradient-general}
\end{equation}
Thus the centered optimal displacement is
\begin{equation}
\boldsymbol{m}^*(i,g)-\bar{\boldsymbol{\xi}}_i
=
\boldsymbol{\Sigma}_{\boldsymbol{P}_0}(i)\nabla\log\psi^g(\boldsymbol{x}_i)
+
\text{higher-order local terms}.
\label{eq:app-policy-gradient-centered}
\end{equation}
If the passive local motion has zero mean, \(\bar{\boldsymbol{\xi}}_i=0\), then Eq.~\eqref{eq:app-policy-gradient-centered} reduces to the uncentered expression.
Comparing Eq.~\eqref{eq:app-policy-gradient-centered} with Eq.~\eqref{eq:app-centered-popvec}, the place-kernel readout matches the local LMDP policy mean to leading order when \(\boldsymbol{\mu}_G(\boldsymbol{y})\approx\bar{\boldsymbol{\xi}}_i\) and \(\boldsymbol{\Sigma}_G(\boldsymbol{y})\approx\boldsymbol{\Sigma}_{\boldsymbol{P}_0}(i)\). Otherwise it is a geometrically preconditioned log-desirability readout, with preconditioner \(\boldsymbol{\Sigma}_G(\boldsymbol{y})\) rather than \(\boldsymbol{\Sigma}_{\boldsymbol{P}_0}(i)\).

For general reversible \(\boldsymbol{\pi}\),
\[
z(j)=\sqrt{\frac{\pi_g}{\pi_j}}\psi^g(j).
\]
The exact LMDP policy can be written as
\begin{equation}
\mu^*(j\mid i)
\propto
P_0(j\mid i)\frac{\psi^g(j)}{\sqrt{\pi_j}}.
\label{eq:app-policy-general-pi}
\end{equation}
Equivalently,
\[
\mu^*(j\mid i)\propto S_{ij}\psi^g(j),
\]
because the factor \(1/\sqrt{\pi_i}\) is constant for fixed \(i\).
Hence the exact policy readout under nonuniform \(\boldsymbol{\pi}\) uses the corrected field
\[
\boldsymbol{\psi}^g/\sqrt{\boldsymbol{\pi}}.
\]

\section{Nonlinear perturbation under strong goal recall}
\label{app:perturbation}

The exact equivalence in Theorem~\ref{thm:app-linear-equivalence} holds for the homogeneous linearized dynamics. Strong goal recall can drive the clamped goal assembly and nearby assemblies into the nonlinear range of \(\phi\). The log-gradient readout is invariant under spatially constant multiplicative factors, so the relevant nonlinear effect is the residual spatial shape change after absorbing the leading amplitude renormalization. The bounds below control this distortion in the perturbative regime measured by \(A_{\mathrm{eff}}\), including large raw clamps whose saturating output keeps \(A_{\mathrm{eff}}\) within the small-gain range.

\subsection{Nonlinear steady-state equation}
\label{app:perturbation-nonlinear-steady}

Work in the homogeneous-gain scaled variables of App.~\ref{app:linearization-homogeneous-gain}. Let
\[
A=\bar u_g
\]
be the clamped goal perturbation. The effective recurrent output emitted by the clamped goal, expressed in linearized activity units, is
\begin{equation}
A_{\mathrm{eff}}(A)
=
\frac{\phi(r_g^0+A)-\phi(r_g^0)}{\phi'(r_g^0)}.
\label{eq:app-Aeff}
\end{equation}
This definition assumes
\[
\phi'(r_g^0)>0.
\]
For small \(A\),
\[
A_{\mathrm{eff}}(A)=A+\mathcal{O}(A^2).
\]
For bounded saturating \(\phi\) and fixed \(\phi'(r_g^0)>0\), \(A_{\mathrm{eff}}(A)\) is bounded as \(A\to\infty\).

On \(\mathcal I_g\), the nonlinear steady-state equation is
\begin{equation}
(\boldsymbol{I}-\alpha \boldsymbol{S}_{\mathcal I_g\mathcal I_g})\boldsymbol{u}_{\mathcal I_g}
=
\alpha \boldsymbol{S}_{\mathcal I_g g}A_{\mathrm{eff}}(A)
+
\bm{\mathcal{R}}(\boldsymbol{u}_{\mathcal I_g};A),
\label{eq:app-nonlinear-steady}
\end{equation}
where \(\bm{\mathcal{R}}\) contains the quadratic and higher-order nonlinear terms away from the clamped goal.

Define
\[
\boldsymbol{G}_g(\alpha)
=
(\boldsymbol{I}-\alpha \boldsymbol{S}_{\mathcal I_g\mathcal I_g})^{-1},
\]
and define the unit-boundary linear field
\begin{equation}
\boldsymbol{H}_{\mathcal I_g}^g
=
\boldsymbol{G}_g(\alpha)\,\alpha \boldsymbol{S}_{\mathcal I_g g}.
\label{eq:app-H-field}
\end{equation}
By boundary convention,
\[
H_g^g=1,
\qquad
\boldsymbol{H}^g|_{\mathcal I_g}=\boldsymbol{H}_{\mathcal I_g}^g.
\]
The full nonlinear field can be written as
\begin{equation}
\boldsymbol{u}_{\mathcal I_g}
=
A_{\mathrm{eff}}\boldsymbol{H}_{\mathcal I_g}^g+\boldsymbol{\eta}_{\mathcal I_g},
\qquad
\boldsymbol{\eta}_{\mathcal I_g}=\boldsymbol{G}_g(\alpha)\bm{\mathcal{R}}(\boldsymbol{u}_{\mathcal I_g};A).
\label{eq:app-eta-definition}
\end{equation}

\subsection{Linear exponential envelope}
\label{app:perturbation-linear-envelope}

For the reversible passive chain,
\[
H_i^g
=
\sqrt{\frac{\pi_i}{\pi_g}}\,
\mathbb E_{\boldsymbol{P}_0,i}
\!\left[
\alpha^{\tau_g}\mathbb{1}_{\{\tau_g<\infty\}}
\right].
\]
Since \(\tau_g\ge d(i,g)\),
\begin{equation}
|H_i^g|
\le
C_h e^{-\kappa d(i,g)},
\qquad
C_h=\sup_i\sqrt{\frac{\pi_i}{\pi_g}}.
\label{eq:app-H-decay}
\end{equation}
In the uniform-\(\boldsymbol{\pi}\) case one may take \(C_h=1\). This is a path-length envelope. The sharp asymptotic decay of a screened Green function can depend on graph geometry and screening mass.

The Dirichlet Green kernel has the corresponding envelope
\begin{align}
0\le [\boldsymbol{G}_g(\alpha)]_{ij}
&=
\sqrt{\frac{\pi_i}{\pi_j}}
\sum_{t\ge0}\alpha^t
[\boldsymbol{P}_{0,\mathcal I_g\mathcal I_g}^{t}]_{ij}
\nonumber\\
&\le
C_\pi\frac{\alpha^{d(i,j)}}{1-\alpha}
=
C_\pi\frac{e^{-\kappa d(i,j)}}{1-\alpha},
\label{eq:app-G-decay}
\end{align}
where
\[
C_\pi=\sup_{i,j}\sqrt{\frac{\pi_i}{\pi_j}}.
\]
\subsection{Small-gain control of the nonlinear solution}
\label{app:perturbation-small-gain}

Define the weighted norm
\[
\|\boldsymbol{v}\|_{H}
=
\sup_{i\in\mathcal I_g}\frac{|v_i|}{H_i^g}.
\]
Because the component is finite and \(H_i^g>0\), this norm is finite. Let
\[
\boldsymbol{Q}=(Q_i)_{i\in\mathcal I_g},
\qquad
Q_i
=
(H_i^g)^2+\sum_j |S_{ij}|(H_j^g)^2,
\]
and define
\[
M_Q
=
\sup_{i\in\mathcal I_g}
\frac{[\boldsymbol{G}_g(\alpha)\boldsymbol{Q}]_i}{H_i^g}.
\]
On a finite component, \(M_Q<\infty\). For graph families, a uniform version of this estimate requires a uniform bound on neighboring ratios \(H_j^g/H_i^g\) on the support of \(\boldsymbol{S}\), or equivalently a uniform Harnack-type constant for the screened field. From Eq.~\eqref{eq:app-remainder-bound},
\[
\|\boldsymbol{G}_g(\alpha)\bm{\mathcal{R}}(\boldsymbol{v};A)\|_H
\le
C_Q\|\boldsymbol{v}\|_H^2,
\qquad
C_Q=\widetilde C_\phi M_Q,
\]
as long as \(\boldsymbol{v}\) remains inside the activity interval on which the Taylor remainder bound holds.

If
\[
4C_Q|A_{\mathrm{eff}}|<1
\]
and the corresponding local Lipschitz constant of \(\boldsymbol{G}_g\bm{\mathcal{R}}\) on the ball \(\|\boldsymbol{v}\|_H\le2|A_{\mathrm{eff}}|\) is smaller than one, the nonlinear fixed-point map
\[
\mathcal T(\boldsymbol{v})=A_{\mathrm{eff}}\boldsymbol{H}^g+\boldsymbol{G}_g(\alpha)\bm{\mathcal{R}}(\boldsymbol{v};A)
\]
is a contraction on that ball. The nonlinear solution exists uniquely in the ball and satisfies
\begin{equation}
\|\boldsymbol{u}\|_H\le2|A_{\mathrm{eff}}|,
\qquad
\|\boldsymbol{\eta}\|_H\le4C_Q|A_{\mathrm{eff}}|^2.
\label{eq:app-nonlinear-H-envelope}
\end{equation}
Consequently,
\begin{equation}
|\eta_i|
\le
C_\eta |A_{\mathrm{eff}}|^2 H_i^g
\le
C_\eta C_h |A_{\mathrm{eff}}|^2e^{-\kappa d(i,g)}.
\label{eq:app-eta-absolute-envelope}
\end{equation}
The absolute correction has the same universal path-length envelope as the linear field. The ratio \(\eta_i/(A_{\mathrm{eff}}H_i^g)\) is controlled by \(\mathcal{O}(|A_{\mathrm{eff}}|)\) under the small-gain condition and is generally not forced to decay by Eq.~\eqref{eq:app-eta-absolute-envelope} alone.

\subsection{Far-field amplitude projection}
\label{app:perturbation-far-field-amplitude}

The following geometric condition captures the far-field structure needed to separate a scalar amplitude shift from a residual shape correction.

\paragraph{Far-field ratio condition.}
There exist constants \(\beta>0\), \(C_{\mathrm{far}}<\infty\), and linear functionals \(\mathcal A_g\) such that, for every source \(\boldsymbol{v}\) with finite weighted source norm
\[
\|\boldsymbol{v}\|_{2,g}
=
\sup_j
\frac{|v_j|}{(H_j^g)^2}
<\infty,
\]
one has
\begin{equation}
\boldsymbol{G}_g(\alpha)\boldsymbol{v}
=
\mathcal A_g[\boldsymbol{v}]\boldsymbol{H}^g+\boldsymbol{\varepsilon}_g[\boldsymbol{v}],
\label{eq:app-far-field-decomp}
\end{equation}
with
\begin{equation}
|(\boldsymbol{\varepsilon}_g[\boldsymbol{v}])_i|
\le
C_{\mathrm{far}}\|\boldsymbol{v}\|_{2,g}H_i^g e^{-\beta d(i,g)}.
\label{eq:app-far-field-residual}
\end{equation}
This condition holds for standard single-channel far-field geometries such as homogeneous screened one-dimensional chains and regular corridors, and for graph regions where localized source responses have a unique outgoing positive harmonic profile. On arbitrary finite graphs it is an explicit geometric assumption. The residual log-gradient and shell-radius estimates below are conditional on this far-field ratio condition.

Applying Eq.~\eqref{eq:app-far-field-decomp} to
\[
\boldsymbol{v}=\bm{\mathcal{R}}(\boldsymbol{u};A)
\]
gives
\begin{equation}
\boldsymbol{\eta}
=
\mathcal A_g[\bm{\mathcal{R}}]\boldsymbol{H}^g+\boldsymbol{\zeta},
\qquad
|\zeta_i|
\le
C_\zeta |A_{\mathrm{eff}}|^2H_i^g e^{-\beta d(i,g)}.
\label{eq:app-eta-projected}
\end{equation}
Define the renormalized amplitude
\begin{equation}
A_{\mathrm{ren}}
=
A_{\mathrm{eff}}+\mathcal A_g[\bm{\mathcal{R}}].
\label{eq:app-Aren}
\end{equation}
Then
\begin{equation}
u_i=A_{\mathrm{ren}}H_i^g+\zeta_i.
\label{eq:app-nonlinear-amplitude-residual}
\end{equation}
For sufficiently small \(|A_{\mathrm{eff}}|\),
\[
A_{\mathrm{ren}}
=
A_{\mathrm{eff}}\left(1+\mathcal{O}(|A_{\mathrm{eff}}|)\right).
\]
On any region where
\[
A_{\mathrm{ren}}>0,
\qquad
\left|\frac{\zeta_i}{A_{\mathrm{ren}}H_i^g}\right|<1,
\]
the log-gradient is well defined. Write
\[
u_i=A_{\mathrm{ren}}H_i^g(1+\varrho_i),
\qquad
\varrho_i=\frac{\zeta_i}{A_{\mathrm{ren}}H_i^g}.
\]
Then
\begin{equation}
|\varrho_i|
\le
C_\varrho |A_{\mathrm{eff}}|e^{-\beta d(i,g)}.
\label{eq:app-varrho-decay}
\end{equation}

For an accessible edge \(e=(i,j)\),
\begin{align}
\nabla_e\log u
&=
\nabla_e\log H^g
+
\nabla_e\log(1+\varrho).
\end{align}
If \(\|\varrho\|_\infty\le\varepsilon<1\), then
\begin{equation}
|\nabla_e\log(1+\varrho)|
\le
\frac{|\varrho_j-\varrho_i|}{1-\varepsilon}.
\label{eq:app-loggrad-varrho}
\end{equation}
Using Eq.~\eqref{eq:app-varrho-decay}, edges whose endpoints are at distance at least \(r\) from the goal satisfy
\begin{equation}
|\nabla_e\log u-\nabla_e\log H^g|
\le
C_\nabla |A_{\mathrm{eff}}|e^{-\beta r}.
\label{eq:app-loggrad-decay}
\end{equation}

For a target residual relative shape tolerance \(\varepsilon\), the affected shell radius is
\begin{equation}
r_\varepsilon
=
\left[
\frac1\beta
\log
\left(
\frac{C_\varrho |A_{\mathrm{eff}}(A)|}{\varepsilon}
\right)
\right]_+,
\qquad
[x]_+=\max\{x,0\}.
\label{eq:app-shell-radius}
\end{equation}
In the linear recall range \(A_{\mathrm{eff}}(A)\approx A\), multiplying recall strength by \(c\) shifts the radius by
\[
\Delta r_\varepsilon=\frac{\log c}{\beta}.
\]
The residual decay length is \(\ell_\beta=1/\beta\). An \(e\)-fold amplitude change shifts the shell by one \(\ell_\beta\); a tenfold change shifts it by \((\log 10)\ell_\beta\). For bounded saturating \(\phi\), \(A_{\mathrm{eff}}(A)\) is bounded, and the shell radius in Eq.~\eqref{eq:app-shell-radius} has a finite upper bound.

\section{Weak-cost limit and Green functions}
\label{app:green}
\label{app:green-laplacian-identity}
The recurrent Laplacian \(\boldsymbol{L}=\boldsymbol{I}-\boldsymbol{S}\) (App.~\ref{app:notation}) coincides with the discrete Laplacian of~\citep{Zuo2026Laplacian} up to a positive scalar. Under the lazy normalization of Sec.~\ref{sec:btsp-rule} with rate \(\epsilon\) and symmetric transition counts \(\boldsymbol{W}^{\mathrm{tr}}\), one has \(\boldsymbol{S}=\boldsymbol{P}_0\) and uniform \(\boldsymbol{\pi}\), so
\begin{equation}
\boldsymbol{L} = \epsilon\,(\boldsymbol{D}_{\mathrm{deg}}-\boldsymbol{W}^{\mathrm{tr}}),
\qquad
\boldsymbol{D}_{\mathrm{deg}} = \operatorname{diag}(d_1,\ldots,d_N),
\qquad
d_i = \sum_{j\neq i}W^{\mathrm{tr}}_{ij}.
\label{eq:app-L-Lpos}
\end{equation}
The right-hand side is the discrete Laplacian \(\boldsymbol{L}_{\mathrm{pos}}=\boldsymbol{D}_{\mathrm{deg}}-\boldsymbol{A}_{\mathrm{adj}}\) of~\citep{Zuo2026Laplacian} with adjacency matrix \(\boldsymbol{A}_{\mathrm{adj}}=\boldsymbol{W}^{\mathrm{tr}}\); eigenvectors of \(\boldsymbol{L}\) and \(\boldsymbol{S}\) therefore coincide with the Laplacian eigenfunctions of~\citep{Zuo2026Laplacian}.
\subsection{Screened Dirichlet operator}
\label{app:green-screened-dirichlet}

The goal-clamped linear equation is
\[
(\boldsymbol{I}-\alpha \boldsymbol{S}_{\mathcal I_g\mathcal I_g})\boldsymbol{u}_{\mathcal I_g}
=
\alpha \boldsymbol{S}_{\mathcal I_g g}\bar u_g.
\]
With
\[
\boldsymbol{L}_g^{\mathrm{Dir}}
=
\boldsymbol{I}_{\mathcal I_g}-\boldsymbol{S}_{\mathcal I_g\mathcal I_g},
\qquad
m=\alpha^{-1}-1,
\]
one has
\begin{align}
\boldsymbol{I}_{\mathcal I_g}-\alpha \boldsymbol{S}_{\mathcal I_g\mathcal I_g}
&=
\alpha
\left[
\alpha^{-1}\boldsymbol{I}_{\mathcal I_g}
-
\boldsymbol{S}_{\mathcal I_g\mathcal I_g}
\right]
\nonumber\\
&=
\alpha
\left(
\boldsymbol{L}_g^{\mathrm{Dir}}+m\boldsymbol{I}_{\mathcal I_g}
\right).
\label{eq:app-helmholtz}
\end{align}
Therefore
\begin{equation}
(\boldsymbol{L}_g^{\mathrm{Dir}}+m\boldsymbol{I}_{\mathcal I_g})\boldsymbol{u}_{\mathcal I_g}
=
\boldsymbol{S}_{\mathcal I_g g}\bar u_g.
\label{eq:app-boundary-screened-equation}
\end{equation}
The screened Dirichlet Green operator is
\begin{equation}
\boldsymbol{G}_g(m)
=
(\boldsymbol{L}_g^{\mathrm{Dir}}+m\boldsymbol{I}_{\mathcal I_g})^{-1}.
\label{eq:app-screened-green}
\end{equation}
Thus
\begin{equation}
\boldsymbol{u}_{\mathcal I_g}
=
\boldsymbol{G}_g(m)\boldsymbol{S}_{\mathcal I_g g}\bar u_g.
\label{eq:app-boundary-green}
\end{equation}

Since \(\boldsymbol{L}_g^{\mathrm{Dir}}\succ0\), let
\[
\boldsymbol{L}_g^{\mathrm{Dir}}\boldsymbol{\varphi}_n=\lambda_n\boldsymbol{\varphi}_n,
\qquad
0<\lambda_1\le\lambda_2\le\cdots.
\]
Then
\begin{equation}
\boldsymbol{G}_g(m)
=
\sum_n
\frac{\boldsymbol{\varphi}_n\boldsymbol{\varphi}_n^\top}{\lambda_n+m}.
\label{eq:app-green-spectral}
\end{equation}
As \(m\to0\),
\begin{equation}
\boldsymbol{G}_g(m)\to(\boldsymbol{L}_g^{\mathrm{Dir}})^{-1}.
\label{eq:app-dirichlet-green-limit}
\end{equation}

\subsection{Boundary-driven weak-cost limit}
\label{app:green-boundary-driven}

The boundary-driven clamped-goal field satisfies
\[
\frac{\boldsymbol{u}_{\mathcal I_g}}{\bar u_g}
=
\boldsymbol{G}_g(m)\boldsymbol{S}_{\mathcal I_g g}.
\]
Taking \(m\to0\) gives
\begin{equation}
\lim_{m\to0}\frac{\boldsymbol{u}_{\mathcal I_g}}{\bar u_g}
=
(\boldsymbol{L}_g^{\mathrm{Dir}})^{-1}\boldsymbol{S}_{\mathcal I_g g}.
\label{eq:app-boundary-limit-general}
\end{equation}
This vector can be evaluated using the null vector of \(\boldsymbol{L}\). Since
\[
\boldsymbol{L}\sqrt{\boldsymbol{\pi}}=0,
\]
the non-goal rows give
\[
(\boldsymbol{I}_{\mathcal I_g}-\boldsymbol{S}_{\mathcal I_g\mathcal I_g})\sqrt{\boldsymbol{\pi}}_{\mathcal I_g}
=
\boldsymbol{S}_{\mathcal I_g g}\sqrt{\pi_g}.
\]
Hence
\begin{equation}
(\boldsymbol{L}_g^{\mathrm{Dir}})^{-1}\boldsymbol{S}_{\mathcal I_g g}
=
\frac{\sqrt{\boldsymbol{\pi}}_{\mathcal I_g}}{\sqrt{\pi_g}}.
\label{eq:app-boundary-limit-pi}
\end{equation}
Therefore
\begin{equation}
\lim_{m\to0}\psi_i^g
=
\sqrt{\frac{\pi_i}{\pi_g}}.
\label{eq:app-boundary-limit-network}
\end{equation}
In desirability coordinates,
\[
z_i^g=\sqrt{\frac{\pi_g}{\pi_i}}\psi_i^g,
\]
so
\begin{equation}
\lim_{m\to0}z_i^g=1
\label{eq:app-hitting-probability-limit}
\end{equation}
on the goal-connected component. This is the hitting-probability field.

Let
\[
\epsilon_q=\frac{q_0}{\lambda}.
\]
Since
\[
\alpha=e^{-\epsilon_q},
\]
Eq.~\eqref{eq:app-z-hitting} gives
\[
z_i^g
=
\mathbb E_i[e^{-\epsilon_q\tau_g}].
\]
On a finite component,
\begin{equation}
z_i^g
=
1-\epsilon_q\mathbb E_i[\tau_g]
+
\frac{\epsilon_q^2}{2}\mathbb E_i[\tau_g^2]
+
\mathcal{O}(\epsilon_q^3\mathbb E_i[\tau_g^3]).
\label{eq:app-hitting-time-expansion}
\end{equation}
Thus
\begin{equation}
\log z_i^g
=
-\epsilon_q\mathbb E_i[\tau_g]
+
\frac{\epsilon_q^2}{2}\operatorname{Var}_i(\tau_g)
+
\mathcal{O}(\epsilon_q^3),
\label{eq:app-log-z-hitting}
\end{equation}
and
\begin{equation}
V_i^g
=
-\lambda\log z_i^g
=
q_0\mathbb E_i[\tau_g]
-
\frac{q_0^2}{2\lambda}\operatorname{Var}_i(\tau_g)
+
\mathcal{O}(q_0^3/\lambda^2).
\label{eq:app-value-hitting-time}
\end{equation}
The log-gradient signal in desirability coordinates satisfies, for an edge \(e=(i,j)\),
\begin{equation}
\nabla_e\log z^g
=
-\frac{q_0}{\lambda}\nabla_e\mathbb E[\tau_g]
+
\mathcal{O}\!\left((q_0/\lambda)^2\right).
\label{eq:app-weak-cost-loggrad}
\end{equation}
Thus the spatial variation of the weak-cost desirability is carried by the hitting-time correction.

\subsection{Dirichlet source-response Green function}
\label{app:green-dirichlet-source-response}

For a localized source \(\boldsymbol{s}\) on \(\mathcal I_g\), consider
\begin{equation}
(\boldsymbol{L}_g^{\mathrm{Dir}}+m\boldsymbol{I})\boldsymbol{y}_m=\boldsymbol{s}.
\label{eq:app-source-response-screened}
\end{equation}
Then
\[
\boldsymbol{y}_m=\boldsymbol{G}_g(m)\boldsymbol{s},
\]
and
\begin{equation}
\boldsymbol{y}_m\to(\boldsymbol{L}_g^{\mathrm{Dir}})^{-1}\boldsymbol{s}
\qquad
(m\to0).
\label{eq:app-source-response-dirichlet}
\end{equation}
If \(\boldsymbol{s}=\boldsymbol{e}_a\) for \(a\in\mathcal I_g\), this limit is the \(a\)-th column of the Dirichlet Green function.

The boundary-driven field and an interior localized-source response have distinct weak-cost limits:
\[
\boldsymbol{G}_g(m)\boldsymbol{S}_{\mathcal I_g g}
\to
\frac{\sqrt{\boldsymbol{\pi}}_{\mathcal I_g}}{\sqrt{\pi_g}},
\]
whereas
\[
\boldsymbol{G}_g(m)\boldsymbol{e}_a
\to
(\boldsymbol{L}_g^{\mathrm{Dir}})^{-1}\boldsymbol{e}_a.
\]
\subsection{Full-graph mean-zero Green function}
\label{app:green-full-mean-zero}

On a connected component, the full Laplacian
\[
\boldsymbol{L}=\boldsymbol{I}-\boldsymbol{S}
\]
has null vector
\[
\boldsymbol{\varphi}_0=\sqrt{\boldsymbol{\pi}}.
\]
Let
\[
\boldsymbol{P}_\perp=\boldsymbol{I}-\boldsymbol{\varphi}_0\boldsymbol{\varphi}_0^\top
\]
be the Euclidean orthogonal projection onto the subspace orthogonal to \(\sqrt{\boldsymbol{\pi}}\). The screened full-graph Green operator is
\[
\boldsymbol{G}(m)=(\boldsymbol{L}+m\boldsymbol{I})^{-1}.
\]
With orthonormal eigenpairs
\[
\boldsymbol{L}\boldsymbol{\varphi}_n=\lambda_n\boldsymbol{\varphi}_n,
\qquad
n\ge1,
\qquad
\lambda_n>0,
\]
the spectral expansion is
\begin{equation}
\boldsymbol{G}(m)
=
\frac1m\boldsymbol{\varphi}_0\boldsymbol{\varphi}_0^\top
+
\sum_{n\ge1}\frac{\boldsymbol{\varphi}_n\boldsymbol{\varphi}_n^\top}{\lambda_n+m}.
\label{eq:app-full-green}
\end{equation}
For a source \(\boldsymbol{s}\) satisfying
\[
\langle \boldsymbol{s},\boldsymbol{\varphi}_0\rangle=0,
\]
the divergent mode is absent and
\begin{equation}
\boldsymbol{G}(m)\boldsymbol{s}\to \boldsymbol{L}^\dagger \boldsymbol{s}
\qquad
(m\to0),
\label{eq:app-mean-zero-green}
\end{equation}
where
\[
\boldsymbol{L}^\dagger
=
\sum_{n\ge1}\frac{\boldsymbol{\varphi}_n\boldsymbol{\varphi}_n^\top}{\lambda_n}
\]
is the Moore--Penrose inverse on the \(\sqrt{\boldsymbol{\pi}}\)-orthogonal subspace.

For a point source at \(g\) in the symmetric network coordinate, the mean-zero source is
\begin{equation}
\boldsymbol{s}^{(g)}
=
\boldsymbol{P}_\perp \boldsymbol{e}_g
=
\boldsymbol{e}_g-\sqrt{\pi_g}\sqrt{\boldsymbol{\pi}}.
\label{eq:app-mean-zero-point-source}
\end{equation}
The corresponding full-graph mean-zero Green column is
\begin{equation}
\boldsymbol{g}_0^{(g)}
=
\boldsymbol{L}^\dagger
\left(
\boldsymbol{e}_g-\sqrt{\pi_g}\sqrt{\boldsymbol{\pi}}
\right).
\label{eq:app-full-green-column}
\end{equation}

For any vector \(\boldsymbol{w}\) in the symmetric network coordinate, the correct centering operation is
\begin{equation}
\mathcal C_{\boldsymbol{\pi}}(\boldsymbol{w})
=
\boldsymbol{P}_\perp \boldsymbol{w}
=
\boldsymbol{w}-\sqrt{\boldsymbol{\pi}}\,\langle\sqrt{\boldsymbol{\pi}},\boldsymbol{w}\rangle.
\label{eq:app-pi-centering}
\end{equation}
A scale-free shape normalization is
\begin{equation}
\mathcal N_{\boldsymbol{\pi}}(\boldsymbol{w})
=
\frac{\mathcal C_{\boldsymbol{\pi}}(\boldsymbol{w})}{\|\mathcal C_{\boldsymbol{\pi}}(\boldsymbol{w})\|_2},
\qquad
\mathcal C_{\boldsymbol{\pi}}(\boldsymbol{w})\ne0.
\label{eq:app-pi-shape-normalization}
\end{equation}
Arithmetic mean-centering coincides with Eq.~\eqref{eq:app-pi-centering} exactly in the uniform-\(\boldsymbol{\pi}\) case.

For an uncentered point response,
\begin{equation}
\boldsymbol{G}(m)\boldsymbol{e}_g
=
\frac{\sqrt{\pi_g}}{m}\sqrt{\boldsymbol{\pi}}
+
\boldsymbol{L}^\dagger
\left(
\boldsymbol{e}_g-\sqrt{\pi_g}\sqrt{\boldsymbol{\pi}}
\right)
+
\mathcal{O}(m).
\label{eq:app-point-source-divergence}
\end{equation}
Subtracting the \(\sqrt{\boldsymbol{\pi}}\) component isolates the finite mean-zero Green function. If the full graph has several connected components, \(\boldsymbol{P}_\perp\) is replaced by the projection onto the orthogonal complement of the full nullspace of \(\boldsymbol{L}\).

\section{Local environmental edits and low-rank structure}
\label{app:woodbury}

\subsection{Rank of a local connectivity edit}
\label{app:woodbury-rank-local-edit}

Let
\[
\boldsymbol{S}'=\boldsymbol{S}+\Delta\boldsymbol{S}
\]
be the recurrent matrix after a local environmental change. Closing or opening \(k\) undirected edges changes \(\mathcal{O}(k)\) off-diagonal entries. Under the lazy normalization with fixed \(\epsilon\), the same edit changes the self-loop entries of the affected rows so that row sums remain one. Under standard random-walk normalization, the same rank scaling holds on bounded-degree graphs, because a changed edge rescales only the outgoing weights incident to the affected vertices. Thus
\begin{equation}
\operatorname{rank}(\Delta\boldsymbol{S})\le r_{\mathrm{edit}},
\qquad
r_{\mathrm{edit}}=\mathcal{O}(k),
\label{eq:app-deltaS-rank}
\end{equation}
where \(k\) counts changed local matrix entries up to a bounded-degree constant.

Write a rank factorization
\begin{equation}
\Delta\boldsymbol{S}=\boldsymbol{U}_\Delta \boldsymbol{C}_\Delta \boldsymbol{V}_\Delta^\top,
\label{eq:app-deltaS-factor}
\end{equation}
with \(\boldsymbol{U}_\Delta,\boldsymbol{V}_\Delta\in\mathbb R^{N\times r_{\mathrm{edit}}}\) and nonsingular \(\boldsymbol{C}_\Delta\in\mathbb R^{r_{\mathrm{edit}}\times r_{\mathrm{edit}}}\), after discarding zero singular directions. A single directed entry change \(\Delta S_{ab}=\delta\) is the rank-one matrix \(\delta \boldsymbol{e}_a \boldsymbol{e}_b^\top\). A symmetric edge edit with lazy row-stochastic renormalization is a sum of a constant number of such rank-one terms.

\subsection{Fixed-goal Dirichlet update}
\label{app:woodbury-fixed-goal-update}

For a fixed goal \(g\), define
\[
\boldsymbol{A}_g=\boldsymbol{I}-\alpha \boldsymbol{S}_{\mathcal I_g\mathcal I_g},
\qquad
\boldsymbol{b}^{(g)}=\alpha \boldsymbol{S}_{\mathcal I_g g}\bar u_g.
\]
The original field is
\[
\boldsymbol{u}^{(g)}=\boldsymbol{A}_g^{-1}\boldsymbol{b}^{(g)}.
\]
After the edit,
\[
\boldsymbol{A}'_g
=
\boldsymbol{I}-\alpha \boldsymbol{S}'_{\mathcal I_g\mathcal I_g}
=
\boldsymbol{A}_g-\alpha\Delta \boldsymbol{S}_{\mathcal I_g\mathcal I_g},
\]
and
\[
\boldsymbol{b}'^{(g)}
=
\boldsymbol{b}^{(g)}+\delta \boldsymbol{b}^{(g)},
\qquad
\delta \boldsymbol{b}^{(g)}
=
\alpha\Delta \boldsymbol{S}_{\mathcal I_g g}\bar u_g.
\]
Factor
\begin{equation}
-\alpha\Delta \boldsymbol{S}_{\mathcal I_g\mathcal I_g}
=
\boldsymbol{U}_g\boldsymbol{C}_g\boldsymbol{V}_g^\top.
\label{eq:app-fixed-goal-factor}
\end{equation}
The factorization is chosen with nonsingular \(\boldsymbol{C}_g\) on the retained rank.
Then
\[
\boldsymbol{A}'_g=\boldsymbol{A}_g+\boldsymbol{U}_g\boldsymbol{C}_g\boldsymbol{V}_g^\top.
\]
Woodbury gives
\begin{equation}
(\boldsymbol{A}'_g)^{-1}
=
\boldsymbol{A}_g^{-1} -
\boldsymbol{A}_g^{-1}\boldsymbol{U}_g
\left(
\boldsymbol{C}_g^{-1}+\boldsymbol{V}_g^\top \boldsymbol{A}_g^{-1}\boldsymbol{U}_g
\right)^{-1}
\boldsymbol{V}_g^\top \boldsymbol{A}_g^{-1}.
\label{eq:app-woodbury-fixed-goal}
\end{equation}
Therefore
\begin{align}
\boldsymbol{u}'^{(g)}-\boldsymbol{u}^{(g)}
&=
\boldsymbol{A}_g^{-1}\delta \boldsymbol{b}^{(g)}
\nonumber\\
&\quad
-
\boldsymbol{A}_g^{-1}\boldsymbol{U}_g
\left(
\boldsymbol{C}_g^{-1}+\boldsymbol{V}_g^\top \boldsymbol{A}_g^{-1}\boldsymbol{U}_g
\right)^{-1}
\boldsymbol{V}_g^\top \boldsymbol{A}_g^{-1}
\bigl(\boldsymbol{b}^{(g)}+\delta \boldsymbol{b}^{(g)}\bigr).
\label{eq:app-field-difference-fixed-goal}
\end{align}
The fixed-goal field change lies in the span of
\[
\boldsymbol{A}_g^{-1}\boldsymbol{U}_g
\]
together with
\[
\boldsymbol{A}_g^{-1}\delta \boldsymbol{b}^{(g)}.
\]
Under fixed-\(\epsilon\) lazy normalization, if no edited matrix entry lies in the goal column \(\boldsymbol{S}_{\mathcal I_g g}\), then
\[
\delta \boldsymbol{b}^{(g)}=0,
\]
and the field change lies in an \(\mathcal{O}(k)\)-dimensional subspace. Under random-walk normalization, \(\delta \boldsymbol{b}^{(g)}\) can also change when an edit changes the degree of a vertex adjacent to \(g\), or changes \(d_g\). When the boundary column changes, \(\delta \boldsymbol{b}^{(g)}\) is supported on the locally edited boundary entries under lazy normalization and on the affected goal-neighborhood entries under bounded-degree random-walk normalization, contributing another \(\mathcal{O}(k)\)-dimensional subspace.

\subsection{Shared full-resolvent update}
\label{app:woodbury-shared-full-resolvent}

Let
\[
\boldsymbol{R}=(\boldsymbol{I}-\alpha \boldsymbol{S})^{-1},
\qquad
\boldsymbol{R}'=(\boldsymbol{I}-\alpha \boldsymbol{S}')^{-1}.
\]
Factor
\begin{equation}
-\alpha\Delta\boldsymbol{S}=\boldsymbol{U}\boldsymbol{C}\boldsymbol{V}^\top.
\label{eq:app-full-factor}
\end{equation}
The factorization is chosen with nonsingular \(\boldsymbol{C}\) on the retained rank.
Then
\[
\boldsymbol{I}-\alpha \boldsymbol{S}'=\boldsymbol{I}-\alpha \boldsymbol{S}+\boldsymbol{U}\boldsymbol{C}\boldsymbol{V}^\top.
\]
Woodbury gives
\begin{equation}
\boldsymbol{R}'
=
\boldsymbol{R} -
\boldsymbol{R}\boldsymbol{U}
\left(
\boldsymbol{C}^{-1}+\boldsymbol{V}^\top \boldsymbol{R}\boldsymbol{U}
\right)^{-1}
\boldsymbol{V}^\top \boldsymbol{R}.
\label{eq:app-resolvent-update}
\end{equation}
Hence
\begin{equation}
\boldsymbol{R}'-\boldsymbol{R}
=
-
\boldsymbol{R}\boldsymbol{U}
\left(
\boldsymbol{C}^{-1}+\boldsymbol{V}^\top \boldsymbol{R}\boldsymbol{U}
\right)^{-1}
\boldsymbol{V}^\top \boldsymbol{R},
\label{eq:app-resolvent-difference}
\end{equation}
and
\begin{equation}
\operatorname{rank}(\boldsymbol{R}'-\boldsymbol{R})\le r_{\mathrm{edit}}=\mathcal{O}(k).
\label{eq:app-resolvent-rank}
\end{equation}

For a fixed source matrix \(\boldsymbol{B}\) independent of the edit,
\[
\boldsymbol{F}=\boldsymbol{R}\boldsymbol{B},
\qquad
\boldsymbol{F}'=\boldsymbol{R}'\boldsymbol{B},
\]
so
\begin{equation}
\boldsymbol{F}'-\boldsymbol{F}=(\boldsymbol{R}'-\boldsymbol{R})\boldsymbol{B}
\label{eq:app-fixed-source-update}
\end{equation}
has rank at most \(r_{\mathrm{edit}}\). If the source matrix changes by a local low-rank term \(\Delta\boldsymbol{B}\), then
\[
\boldsymbol{F}'=\boldsymbol{R}'(\boldsymbol{B}+\Delta\boldsymbol{B}),
\]
and
\begin{equation}
\boldsymbol{F}'-\boldsymbol{F}=(\boldsymbol{R}'-\boldsymbol{R})\boldsymbol{B}+\boldsymbol{R}'\Delta\boldsymbol{B},
\label{eq:app-source-change-update}
\end{equation}
with
\[
\operatorname{rank}(\boldsymbol{F}'-\boldsymbol{F})
\le
r_{\mathrm{edit}}+\operatorname{rank}(\Delta\boldsymbol{B}).
\]
\subsection{Unit-clamped all-goal fields}
\label{app:woodbury-unit-clamped-all-goal-fields}

A unit-clamped goal field can be represented through the full resolvent. For a source \(c_g \boldsymbol{e}_g\),
\[
\boldsymbol{u}^{(g)}=c_g\boldsymbol{R}\boldsymbol{e}_g.
\]
Choosing
\[
c_g=\frac{\bar u}{R_{gg}}
\]
enforces
\[
u_g^{(g)}=\bar u.
\]
For all goals with common clamp amplitude \(\bar u\), define
\begin{equation}
\boldsymbol{\Psi}=\boldsymbol{R}\boldsymbol{D}_{\mathrm{clamp}},
\qquad
\boldsymbol{D}_{\mathrm{clamp}}=\bar u\,\operatorname{diag}(R_{11}^{-1},\ldots,R_{NN}^{-1}).
\label{eq:app-clamped-field-matrix}
\end{equation}
If \(\boldsymbol{S}\ge0\), then
\[
\boldsymbol{R}=\boldsymbol{I}+\alpha \boldsymbol{S}+\alpha^2\boldsymbol{S}^2+\cdots
\]
has nonnegative entries and
\[
R_{gg}\ge1,
\]
so \(\boldsymbol{D}_{\mathrm{clamp}}\) is well defined.

After the edit,
\[
\boldsymbol{\Psi}'=\boldsymbol{R}'\boldsymbol{D}_{\mathrm{clamp}}',
\qquad
\boldsymbol{D}_{\mathrm{clamp}}'=\bar u\,\operatorname{diag}(R_{11}'^{-1},\ldots,R_{NN}'^{-1}).
\]
Thus
\begin{equation}
\boldsymbol{\Psi}'-\boldsymbol{\Psi}
=
(\boldsymbol{R}'-\boldsymbol{R})\boldsymbol{D}_{\mathrm{clamp}}
+
\boldsymbol{R}'(\boldsymbol{D}_{\mathrm{clamp}}'-\boldsymbol{D}_{\mathrm{clamp}}).
\label{eq:app-clamped-field-difference}
\end{equation}
The first term has rank at most \(r_{\mathrm{edit}}=\mathcal{O}(k)\). The second term is the source-renormalization term required to keep every goal clamped to the same amplitude. Its algebraic rank can be full because the diagonal entries \(R_{gg}\) can change for many goals.

The full-rank source-renormalization term is spatially localized in goal index under an exponential resolvent envelope. Let \(\mathcal V_{\mathrm{edit}}\) be the set of vertices incident to the edited entries. Suppose
\[
|R_{ij}|\le C_R e^{-\kappa_R d(i,j)},
\qquad
|R'_{ij}|\le C_R' e^{-\kappa_R d(i,j)}
\]
and the Woodbury middle factor
\[
\left(
\boldsymbol{C}^{-1}+\boldsymbol{V}^\top \boldsymbol{R}\boldsymbol{U}
\right)^{-1}
\]
is uniformly bounded, with \(\boldsymbol{U}\) and \(\boldsymbol{V}\) supported on \(\mathcal V_{\mathrm{edit}}\) and its edited neighborhood. Then Eq.~\eqref{eq:app-resolvent-difference} gives
\begin{equation}
|\Delta R_{gg}|
=
|R'_{gg}-R_{gg}|
\le
C_{\Delta R}e^{-2\kappa_R d(g,\mathcal V_{\mathrm{edit}})}.
\label{eq:app-diagonal-resolvent-change}
\end{equation}
Since \(R_{gg}\ge1\) and \(R'_{gg}\ge1\),
\begin{equation}
|D'_{\mathrm{clamp},gg}-D_{\mathrm{clamp},gg}|
=
\bar u
\left|
\frac1{R'_{gg}}-\frac1{R_{gg}}
\right|
\le
C_D e^{-2\kappa_R d(g,\mathcal V_{\mathrm{edit}})}.
\label{eq:app-diagonal-source-change}
\end{equation}

For a truncation radius \(R_{\mathrm{edit}}\), retain only goals satisfying
\[
d(g,\mathcal V_{\mathrm{edit}})\le R_{\mathrm{edit}}.
\]
The retained source-renormalization term has rank at most the number of goals in that edited-neighborhood ball. The discarded diagonal entries obey
\[
|D'_{\mathrm{clamp},gg}-D_{\mathrm{clamp},gg}|
\le
C_D e^{-2\kappa_R R_{\mathrm{edit}}}.
\]
On graphs with subexponential or sufficiently slow volume growth, the discarded Frobenius energy decays exponentially in \(R_{\mathrm{edit}}\). As \(\alpha\) approaches one, the decay rate \(\kappa_R\) may decrease and resolvent constants may grow, so this locality estimate weakens in the weak-cost limit.

For each individual goal, multiplying the field by a scalar leaves the log-gradient unchanged:
\[
\nabla\log(c_g f^{(g)})=\nabla\log f^{(g)}
\qquad
(c_g>0).
\]
Thus diagonal source renormalization affects raw all-goal field matrices and cancels exactly in per-goal log-gradient readouts.

\section{Numerical experiment details}
\label{app:experiments}

This section gives the numerical protocol used for all simulations in
Sec.~\ref{sec:experiments}. The purpose of the simulations is to check four
claims made in the main text and in the theoretical appendix. First, the
linear goal-clamped network must reproduce the LMDP desirability field with the
stationary-measure correction in Theorem~\ref{thm:equiv}. Second, the same
readout must remain stable when a strong goal clamp pushes local neurons into
the nonlinear range of the activation function. Third, transition-gated BTSP
must write enough legal local connectivity during goal-free exploration to plan
through a maze with obstacles, and a local environmental edit must produce the
low-rank remapping predicted by App.~\ref{app:woodbury-unit-clamped-all-goal-fields}. Fourth, the same
recurrent resolvent must connect to the screened Green function in
App.~\ref{app:green-full-mean-zero}, and the effect of asymmetric BTSP sampling must be
controlled by the perturbation scale derived in App.~\ref{app:asymmetry-irreversible-sampling}. Additional experiments that test the scope of these assumptions, measure the realized relaxation dynamics, and compare the circuit with learned goal-conditioned models are reported in App.~\ref{app:exp-extended}.

\subsection{Common graph, field, and readout conventions}
\label{app:exp-common}

Each open grid cell was represented by one assembly. Coordinates were integer
cell centers, denoted by \(\boldsymbol{x}_i\in\mathbb R^2\). Unless a specific
experiment below states a different normalization, the accessible transition
memory was a symmetric weighted grid graph with eight-neighbor connectivity. A
horizontal or vertical step and a diagonal step were both allowed when the two
end cells were open. Diagonal corner cutting was removed: a diagonal edge was
kept only when the two associated axial edges were also legal. The raw edge
weight was
\begin{equation}
W_{ij}=\exp\left(-\frac{\|\boldsymbol{x}_j-\boldsymbol{x}_i\|_2^2}{2\sigma_e^2}\right),
\qquad \sigma_e=1.
\label{eq:app-exp-edge-weight}
\end{equation}
Thus axial edges had weight \(e^{-1/2}\), diagonal edges had weight \(e^{-1}\),
and blocked transitions had weight zero. The matrix was symmetrized after
edge construction.

For the equivalence, BTSP-learning, low-rank, and Green-limit simulations, the
symmetric transition memory \(\boldsymbol{W}\) was converted into a passive
kernel by the standard random-walk normalization of App.~\ref{app:reversibility-random-walk-normalization}:
\begin{equation}
P_0(j\mid i)=\frac{W_{ij}}{d_i},
\qquad
d_i=\sum_j W_{ij},
\qquad
\pi_i=\frac{d_i}{\sum_\ell d_\ell}.
\label{eq:app-exp-rw-normalization}
\end{equation}
The recurrent matrix used in the linear network was the symmetric coordinate
\begin{equation}
\boldsymbol{S}=\boldsymbol{\Pi}^{1/2}\boldsymbol{P}_0\boldsymbol{\Pi}^{-1/2}.
\label{eq:app-exp-S-normalization}
\end{equation}
This choice makes the stationary-measure factor in
Theorem~\ref{thm:equiv} visible near doors, walls, and boundaries. The
asymmetric BTSP experiment used a fixed-\(\epsilon\) lazy normalization instead;
that case is described separately in App.~\ref{app:exp-asymmetric-btsp}.

For a recalled goal \(g\), the LMDP desirability field was computed from
\begin{equation}
(\boldsymbol{I}_{\mathcal I_g}-\alpha\boldsymbol{P}_{0,\mathcal I_g\mathcal I_g})
\boldsymbol{z}_{\mathcal I_g}^g
=
\alpha\boldsymbol{P}_{0,\mathcal I_g g},
\qquad
z_g^g=1.
\label{eq:app-exp-lmdp-solve}
\end{equation}
The linear recurrent field was computed from
\begin{equation}
(\boldsymbol{I}_{\mathcal I_g}-\alpha\boldsymbol{S}_{\mathcal I_g\mathcal I_g})
\boldsymbol{u}_{\mathcal I_g}
=
\alpha\boldsymbol{S}_{\mathcal I_g g}\bar u_g,
\qquad
u_g=\bar u_g.
\label{eq:app-exp-network-solve}
\end{equation}
The normalized network field was \(\psi_i^g=u_i/\bar u_g\). When
\(\boldsymbol{\pi}\) was nonuniform, the LMDP-coordinate field was
\begin{equation}
\widehat z_i
=
\sqrt{\frac{\pi_g}{\pi_i}}\,\psi_i^g,
\qquad
\widehat z_g=1.
\label{eq:app-exp-pi-correction}
\end{equation}
This is the numerical version of Eq.~\eqref{eq:state-map}. Sparse linear systems
were solved directly with SciPy sparse solvers.

The place-tuned readout was the discrete version of Eq.~\eqref{eq:popvec}. At a
current state \(i\), a graph-geodesic Gaussian place kernel was formed around
\(i\),
\begin{equation}
\mathcal G_i(j)=\exp\left[-\frac{d_{\rm graph}(i,j)^2}{2\sigma_G^2}\right],
\qquad
d_{\rm graph}(i,j)\le 3\sigma_G,
\label{eq:app-exp-place-kernel}
\end{equation}
where graph distance used Euclidean edge lengths on the legal graph. The
population-vector readout was
\begin{equation}
\Delta\boldsymbol{x}_{\rm pop}(i;\boldsymbol{f})
=
\frac{\sum_j \mathcal G_i(j) f_j(\boldsymbol{x}_j-\boldsymbol{x}_i)}
{\sum_j \mathcal G_i(j) f_j},
\label{eq:app-exp-discrete-popvec}
\end{equation}
with \(\boldsymbol{f}=\boldsymbol{z}^g\), \(\widehat{\boldsymbol{z}}\), or the relevant
normalized source-response field. Readout-angle errors were computed as the
angle between two such vectors at the same state, omitting states where either
vector had norm below numerical precision.

For rollout figures, the continuous readout vector was projected onto the legal
neighboring moves. If \(\boldsymbol{v}_i\) is the readout at state \(i\), each
neighbor \(j\) received the score
\begin{equation}
\mathrm{score}(j\mid i)
=
\frac{(\boldsymbol{x}_j-\boldsymbol{x}_i)^\top \boldsymbol{v}_i}
{\|\boldsymbol{x}_j-\boldsymbol{x}_i\|_2\,\|\boldsymbol{v}_i\|_2}
+w_f\frac{f_j-f_i}{|f_i|+10^{-12}}
-w_\ell N_j,
\label{eq:app-exp-rollout-score}
\end{equation}
where \(N_j\) is the number of previous visits to \(j\). The next state was the
legal neighbor with the largest score. The pure readout rollouts used
\(w_f=w_\ell=0\). Some diagnostic rollouts used a small field-gain or loop
penalty, and the exact values are listed below. A route was counted as
successful when it reached the goal state, or when the bug-trap protocol used a
specified goal-zone radius. Route length was the sum of Euclidean lengths of
successive grid moves. Route ratio was route length divided by the shortest
legal graph path length under the same success criterion.

Relative field errors were reported as
\begin{equation}
\varepsilon_{\rm field}(\boldsymbol{a},\boldsymbol{b})
=
\frac{\|\boldsymbol{a}-\boldsymbol{b}\|_2}{\|\boldsymbol{b}\|_2+10^{-15}}.
\label{eq:app-exp-field-error}
\end{equation}
Field correlations were Pearson correlations across all open cells. All plots
of fields used \(\log(\boldsymbol{f}+10^{-10})\) or \(\log(\boldsymbol{f}+10^{-12})\) only for display;
metrics were computed on the untransformed fields unless stated otherwise.

\subsection{Equivalence and strong-recall robustness}
\label{app:exp-equivalence}

The first experiment checks the exact theorem before adding learning or maze
changes. The maze was a two-room diagnostic environment with a zero-thickness
vertical wall. The default size was \(68\times48\) grid cells, with the wall
between columns \(33\) and \(34\), a near door over rows \(10\)--\(15\), and a
far door over rows \(34\)--\(40\). The start was \((9,12)\), the goal was
\((58,37)\), and the legal graph had \(3036\) states and \(11709\) undirected
edges. This maze was chosen because the bottlenecks make \(\boldsymbol{\pi}\)
mildly nonuniform, so the factor \(\sqrt{\pi_i/\pi_g}\) in
Theorem~\ref{thm:equiv} can be tested directly.

For the linear equivalence check, Eq.~\eqref{eq:app-exp-lmdp-solve} and
Eq.~\eqref{eq:app-exp-network-solve} were solved at
\begin{equation}
\alpha\in\{0.90,0.94,0.97,0.985,0.992,0.996\},
\qquad
\bar u_g=1.
\label{eq:app-exp-alpha-grid}
\end{equation}
The network field was converted to the LMDP coordinate by
Eq.~\eqref{eq:app-exp-pi-correction}. At the main setting \(\alpha=0.985\), the
corrected field matched the LMDP desirability with relative error
\(9.55\times10^{-16}\) and correlation \(1.000\). Across the whole \(\alpha\)
grid, the corrected relative error stayed between \(6.29\times10^{-16}\) and
\(9.59\times10^{-16}\). The uncorrected field \(\boldsymbol{\psi}^g\) differed
from \(\boldsymbol{z}^g\) by \(3.17\%\) at \(\alpha=0.985\), matching the size
of the stationary-measure correction. This panel tests the algebraic identity
in Theorem~\ref{thm:equiv}; the uncorrected points show why the
\(\sqrt{\boldsymbol{\pi}}\) factor is needed in a nonuniform environment.

The nonlinear recall test used the same maze and \(\alpha=0.985\). The activity
nonlinearity was
\begin{equation}
\phi(u)=\tanh u.
\label{eq:app-exp-tanh}
\end{equation}
For each clamp amplitude
\begin{equation}
A=\bar u_g\in\{0.05,0.10,0.20,0.40,0.80,1.60,3.20,5.00\},
\label{eq:app-exp-A-grid}
\end{equation}
the non-goal fixed point was iterated from the linear solution using
\begin{equation}
\boldsymbol{u}_{\mathcal I_g}
=
\alpha\boldsymbol{S}_{\mathcal I_g\mathcal I_g}\tanh(\boldsymbol{u}_{\mathcal I_g})
+
\alpha\boldsymbol{S}_{\mathcal I_g g}\tanh A,
\qquad
u_g=A.
\label{eq:app-exp-nonlinear-fixed-point}
\end{equation}
Damped fixed-point iteration used damping factor \(0.72\) and relative tolerance
\(2\times10^{-11}\) for the one-dimensional amplitude sweep. The nonlinear
field was first mapped to the LMDP coordinate by Eq.~\eqref{eq:app-exp-pi-correction}
and then aligned to the linear field by a single scalar least-squares factor.
This scalar alignment removes the uniform amplitude change that cannot affect
Eq.~\eqref{eq:popvec}.

The visible nonlinear shell was measured through local gain loss. For the tanh
unit,
\begin{equation}
1-\frac{\phi'(u)}{\phi'(0)}=\tanh^2 u.
\label{eq:app-exp-gain-loss}
\end{equation}
Cells were binned by graph distance from the goal. In each distance bin, the
\(90\)th percentile of Eq.~\eqref{eq:app-exp-gain-loss} was computed. The shell
radius was the largest distance bin whose \(90\)th percentile exceeded
\(3\times10^{-3}\). This radius measures where the local gain has visibly left
the small-signal regime. At \(A=5\), the radius was eight graph-distance cells.
Outside that shell, the population-vector readout remained close to the linear
prediction: the mean angle error was \(0.080^\circ\), and the \(99\)th percentile
was \(0.60^\circ\). At a probe location near the door and far from the goal,
the angle error was \(0.015^\circ\). The shell radius grew approximately as
\(\log A\), with \(R^2=0.94\). This directly tests the shell-radius estimate in
App.~\ref{app:perturbation-far-field-amplitude}: strong recall changes the local shape near the
clamp, while the action-relevant log-gradient remains stable outside that local
region.

A two-parameter stress sweep repeated the nonlinear solve for every
\(\alpha\) in Eq.~\eqref{eq:app-exp-alpha-grid} and every \(A\) in
Eq.~\eqref{eq:app-exp-A-grid}, using tolerance \(5\times10^{-10}\). The displayed
main row \(\alpha=0.985\) stayed below a \(1^\circ\) outside-shell \(99\)th
percentile across the whole recall range; its maximum was \(0.60^\circ\). This
stress map separates two quantities that can be conflated visually: the field
amplitude can be strongly compressed near the goal, yet the far-field readout
that drives behavior changes by less than one degree in the main regime.

\subsection{Goal-free BTSP learning and local low-rank remapping}
\label{app:exp-btsp}

The second experiment asks whether the recurrent operator needed for planning
can be written by local goal-free transition events. It also asks whether a
local environmental edit produces the low-rank field remapping predicted by
App.~\ref{app:woodbury-unit-clamped-all-goal-fields}. The experiment has two parts: a challenging bug-trap
navigation task and a smaller two-door remapping diagnostic.

Goal-free exploration was simulated as short novelty-biased bouts on the legal
graph. At the beginning of each bout, the animal was placed in a low-visit
region, defined by the lowest quartile of current visit counts. During a bout,
the next state was sampled among legal neighbors using the score
\begin{equation}
q(j\mid i)
=
-1.20\,n_j
-0.80\,n_{\{i,j\}}
+0.55\,h_{ij}
+0.05\,\xi_j.
\label{eq:app-exp-exploration-score}
\end{equation}
Here \(n_j\) is the visit count of state \(j\), \(n_{\{i,j\}}\) is the visit
count of the undirected edge \(\{i,j\}\), \(h_{ij}\) is the cosine alignment
between the proposed step and the current heading, and \(\xi_j\) is a standard
normal random variable. A softmax with temperature \(0.55\) converted these
scores into choice probabilities. This exploration rule has no goal input and
no shortest-path information. It samples under-visited local transitions while
keeping a mild heading persistence.

Each sampled transition \(i\to j\) added one count to both \(W_{ij}\) and
\(W_{ji}\). Thus the learned memory is the symmetric BTSP transition memory in
Sec.~\ref{sec:btsp-rule}. A self-only numerical floor of \(10^{-7}\) was added
before normalization so that isolated rows during early sparse learning did not
make the sparse linear system singular. This floor did not add untraversed
transition geometry. Because every event was sampled from the legal graph, the
learned edge precision was exactly \(1.000\) in all BTSP-learning runs; the
main variable was edge recall.

\paragraph{Bug-trap navigation.}
The bug-trap maze was \(112\times76\) cells. The start was \((102,45)\), the
goal was \((72,43)\), and the legal graph had \(7285\) states and \(27669\)
undirected edges. The start was only \(30.1\) Euclidean cells from the goal, but
the shortest legal route to the goal zone had length \(110.43\). The goal-zone
radius was \(1.5\) cells. This geometry makes the task useful scientifically:
a local cue pointing straight toward the goal hits a wall, so successful
navigation requires a field that has already integrated the global obstacle
layout.

The bug-trap field used \(\alpha=0.993\). The readout width was
\(\sigma_G=3.2\), and pure readout rollouts used
\(w_f=w_\ell=0\), maximum length \(3600\) steps, and goal-zone radius \(1.5\).
Five exploration seeds were used: \(4,15,27,41,57\). Performance was measured at
\(1.0,1.5,2.0,2.5\), and \(3.0\) goal-free transitions per assembly. At the
displayed milestone, \(3.0\) transitions per assembly corresponded to
\(21855\) sampled transition events.

At this milestone, edge recall was \(68.9\%\) on average across the five seeds,
precision was \(1.000\), and the learned field correlated with the ideal-maze
field at \(0.988\pm0.003\) (mean \(\pm\) SEM). All five seeds reached the goal
zone using the pure population-vector readout. The median route ratio was
\(1.16\). For the displayed seed, the pure readout route took \(110\) legal
steps with route ratio \(1.16\); the first readout vector was
\(93.8^\circ\) away from the direct Euclidean vector to the goal. A Euclidean
greedy baseline, which chose the legal neighbor most aligned with the direct
goal vector and used only a small revisit penalty, stalled at the wall after
\(16\) legal steps. This result tests the core claim that local goal-free BTSP
can write a map whose relaxation solves a global obstacle-avoidance query.

\paragraph{Two-door remapping diagnostic.}
The low-rank remapping diagnostic used a smaller \(60\times42\) two-door maze
with the wall between columns \(29\) and \(30\), a near door over rows
\(8\)--\(13\), and a far door over rows \(29\)--\(35\). The start was
\((8,10)\), the goal was \((51,32)\), and the graph had \(2320\) states and
\(8905\) undirected edges. The discount was \(\alpha=0.988\), the readout width
was \(\sigma_G=2.9\), and rollouts used maximum length \(2200\), zero goal
radius, and loop penalty \(w_\ell=0.02\). The exploration checkpoints were
\(0.15,0.30,0.50,0.80,1.10,1.50,2.00\), and \(2.80\) transitions per assembly,
again across seeds \(4,15,27,41,57\).

The displayed remapping point was \(2.0\) transitions per assembly. At this
sparse milestone, the learned graph had mean edge recall \(49.4\%\), precision
\(1.000\), and learned-versus-ideal field correlation \(0.973\pm0.006\). All
five seeds reached the goal, with median route ratio \(1.27\). For the displayed
seed, the open-door route ratio was \(1.20\).

The local edit closed the near door by removing the \(16\) coordinate-level
crossing edges that passed through that door. No other learned weights were
changed, and no new exploration was run. After the edit, the next relaxation
rerouted the displayed goal through the far door; the displayed closed-door
route ratio was \(1.09\). For the Woodbury check, the effective support of
\(\Delta\boldsymbol{S}\) above numerical threshold \(10^{-10}\) involved
\(26\) rows. A direct solve on the closed-door operator and the Woodbury update
of App.~\ref{app:woodbury-fixed-goal-update} agreed with relative error \(1.82\times10^{-15}\).

To test whether the same local edit has low rank across many goals, \(34\) goal
states were selected by taking spatial quantiles of the score
\(\boldsymbol{q}^\top\boldsymbol{x}_i\), with \(\boldsymbol{q}=(1,0.7)^\top\). For each sampled goal, the open-door and closed-door linear network
fields were solved, corrected to the LMDP coordinate by
Eq.~\eqref{eq:app-exp-pi-correction}, and differenced. These \(34\) field
differences were stacked as columns of a matrix. The first two singular modes
explained \(99.46\%\) of its total squared energy. This is the numerical version
of the low-rank statement in Sec.~\ref{sec:lowrank}: a local edit produces a
small set of propagation patterns that explains the field change for many
recalled goals.

\subsection{Weak-cost Green-function limit}
\label{app:exp-green-limit}

The Green-limit experiment used the same \(60\times42\) two-door diagnostic maze
as the remapping diagnostic. It checks the operator identity in
App.~\ref{app:green-full-mean-zero} and clarifies which Green object is being visualized. The
clamped boundary field \(\boldsymbol{\psi}^g\) has a flat hitting-probability
limit in desirability coordinates as \(\alpha\to1\). To isolate the Laplacian
Green object used for spectral navigation, the numerical experiment used the
full-graph point-source response
\begin{equation}
\boldsymbol{r}_\alpha^{(g)}
=(\boldsymbol{I}-\alpha\boldsymbol{S})^{-1}\boldsymbol{e}_g.
\label{eq:app-exp-source-resolvent}
\end{equation}
With \(\boldsymbol{L}=\boldsymbol{I}-\boldsymbol{S}\) and
\begin{equation}
m_\alpha=\frac{1-\alpha}{\alpha},
\label{eq:app-exp-screening-mass}
\end{equation}
the exact algebraic relation is
\begin{equation}
\alpha\boldsymbol{r}_\alpha^{(g)}
=
(\boldsymbol{L}+m_\alpha\boldsymbol{I})^{-1}\boldsymbol{e}_g.
\label{eq:app-exp-resolvent-green-identity}
\end{equation}
This identity was tested at
\begin{equation}
\alpha\in\{0.90,0.94,0.97,0.985,0.992,0.996,0.998,0.999,0.9995\}.
\label{eq:app-exp-green-alpha-grid}
\end{equation}
The relative error in Eq.~\eqref{eq:app-exp-resolvent-green-identity} stayed at
floating-point precision, with maximum value \(2.67\times10^{-13}\) at
\(\alpha=0.9995\).

The weak-cost comparison used the mean-zero Laplacian Green column in the
symmetric network coordinate. Let
\begin{equation}
\boldsymbol{\varphi}_0=\sqrt{\boldsymbol{\pi}},
\qquad
\boldsymbol{P}_\perp
=
\boldsymbol{I}-\boldsymbol{\varphi}_0\boldsymbol{\varphi}_0^\top.
\label{eq:app-exp-null-projection}
\end{equation}
The reference was
\begin{equation}
\boldsymbol{G}_0^{(g)}
=
\boldsymbol{L}^\dagger
\left(\boldsymbol{e}_g-\sqrt{\pi_g}\sqrt{\boldsymbol{\pi}}\right),
\label{eq:app-exp-pseudoinverse-reference}
\end{equation}
computed by a sparse least-squares solve with absolute and relative tolerances
\(10^{-10}\). For each \(m_\alpha\), the screened Green field was projected by
\(\boldsymbol{P}_\perp\) and normalized to unit Euclidean norm before comparison
with Eq.~\eqref{eq:app-exp-pseudoinverse-reference}. The resulting shape error
decreased from \(1.168\) at \(m_\alpha=0.111\) to \(0.162\) at
\(m_\alpha=5.0\times10^{-4}\), showing convergence to the Laplacian Green shape
as the cost decreases.

The spectral reconstruction in Fig.~\ref{fig:app-green-limit} used the
\(72\) smallest eigenpairs of \(\boldsymbol{L}\), computed by sparse symmetric
eigensolver. The screened Green expansion was
\begin{equation}
\boldsymbol{G}_{m_\alpha}^{(g)}
=
\sum_n
\frac{(\boldsymbol{\varphi}_n)_g\,\boldsymbol{\varphi}_n}{\lambda_n+m_\alpha}.
\label{eq:app-exp-spectral-green}
\end{equation}
Panel A of Fig.~\ref{fig:app-green-limit} shows the non-null gains normalized by the first non-null gain and the
shape error after adding successive low-frequency modes. This panel checks that
the Green field is dominated by low-frequency Laplacian modes in the weak-cost
regime, matching the spectral interpretation in Sec.~\ref{sec:exp-green}.

For the displayed field, \(\alpha=0.996\), so \(m_\alpha=0.00402\). The plotted
field was \(\boldsymbol{G}_{m_\alpha}^{(g)}\) normalized by its goal value; white
contours show the independently computed recurrent source response
\(\boldsymbol{r}_\alpha^{(g)}\), also normalized by its goal value. The two agree
with relative error \(2.2\times10^{-15}\). The rollout in the displayed panel
used \(\sigma_G=2.9\), maximum length \(1800\), zero goal radius, and loop
penalty \(w_\ell=0.02\).

\begin{figure}[!t]
\centering
\includegraphics[width=\linewidth]{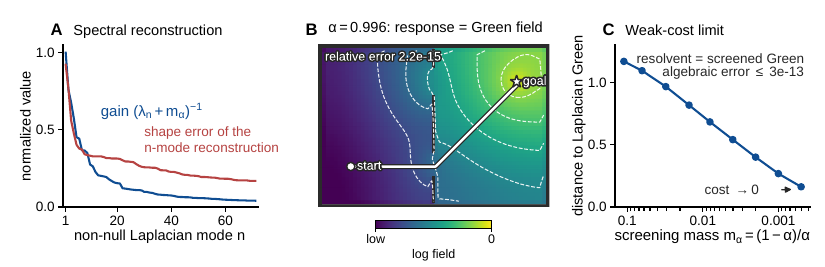}
\caption{\textbf{The recurrent source response is a screened Green function and
approaches the Laplacian Green shape in the weak-cost limit.}
(\textbf{A}) Spectral reconstruction of the screened Green response from the
low-frequency eigenmodes of \(\boldsymbol{L}=\boldsymbol{I}-\boldsymbol{S}\):
normalized gains \((\lambda_n+m_\alpha)^{-1}\) (blue) and the shape error of the
reconstruction after adding successive modes (red).
(\textbf{B}) At \(\alpha=0.996\), the source response of the recurrent network
(white contours) matches the screened Green field (color) to relative error
\(2.2\times10^{-15}\); the white route is the population-vector rollout.
(\textbf{C}) As \(m_\alpha=(1-\alpha)/\alpha\) decreases, the mean-zero screened
Green shape approaches the Laplacian pseudoinverse Green column. The algebraic
identity between the recurrent resolvent and the screened Green operator holds
to \(3\times10^{-13}\) or better throughout the sweep.}
\label{fig:app-green-limit}
\end{figure}

\subsection{Asymmetric BTSP perturbation and bilateral replay}
\label{app:exp-asymmetric-btsp}

The exact theorem assumes a reversible passive kernel. Empirical BTSP sampling
can contain directional residue from biased movement, finite sampling, or
forward-dominant sequence statistics. The asymmetric BTSP experiment asks a
more precise question than generic robustness: whether all deviations from the
reversible reference are ordered by the single dimensionless perturbation scale
\begin{equation}
\chi=\frac{\alpha\|\boldsymbol{E}_{\rm asym}\|_2}{1-\alpha}
\label{eq:app-exp-chi-def}
\end{equation}
from App.~\ref{app:asymmetry-irreversible-sampling}. This question is quantitative and falsifiable.
If \(\eta\), \(\alpha\), maze position, or sampling seed controlled the error
separately, field error and readout error would scatter when plotted against
\(\chi\). A collapse onto one curve supports the perturbation analysis.

The experiment used the same \(112\times76\) bug-trap maze as the main BTSP
navigation task. The start, goal, graph size, and goal-zone radius were the same
as in App.~\ref{app:exp-btsp}. A directed effective count matrix
\(\boldsymbol{C}\) was built from the legal directed edges of the maze. The
symmetric legal support came from the same edge weights in
Eq.~\eqref{eq:app-exp-edge-weight}. A weak circulation bias was then added to
make \(C_{ij}\) and \(C_{ji}\) slightly different while keeping the support
legal.

For each directed edge \(i\to j\), define the unit step
\begin{equation}
\boldsymbol{s}_{ij}
=
\frac{\boldsymbol{x}_j-\boldsymbol{x}_i}
{\|\boldsymbol{x}_j-\boldsymbol{x}_i\|_2},
\label{eq:app-exp-step-vector}
\end{equation}
the edge midpoint \(\boldsymbol{m}_{ij}=(\boldsymbol{x}_i+\boldsymbol{x}_j)/2\),
and a circulation center \(\boldsymbol{c}\). The tangent direction around that
center was
\begin{equation}
\boldsymbol{t}_{ij}
=
\frac{[-(m_{ij,y}-c_y),\; m_{ij,x}-c_x]^\top}
{\|[-(m_{ij,y}-c_y),\; m_{ij,x}-c_x]^\top\|_2}.
\label{eq:app-exp-tangent}
\end{equation}
The directional score was
\begin{equation}
s_{ij}^{\rm flow}=\operatorname{clip}(\boldsymbol{s}_{ij}^\top\boldsymbol{t}_{ij},-1,1),
\label{eq:app-exp-flow-score}
\end{equation}
and the directed count rate was
\begin{equation}
C_{ij}=W_{ij}\exp(\rho s_{ij}^{\rm flow}).
\label{eq:app-exp-directed-count}
\end{equation}
Across sampling seeds, the center was drawn as
\(\boldsymbol{c}=(70,38)+\boldsymbol{\xi}\), with standard deviations
\((1.4,1.0)\) in the two coordinates, and the strength was
\(\rho=0.050\exp(\zeta)\), with \(\zeta\sim\mathcal N(0,0.07^2)\). The sampling
seeds were \(4,15,27,41,57\). This construction gives a controlled directional
residue without adding illegal edges or changing the maze.

The asymmetry knob was
\begin{equation}
\boldsymbol{W}^{\eta}
=(1+\eta)\boldsymbol{C}+(1-\eta)\boldsymbol{C}^\top
=\boldsymbol{C}+\boldsymbol{C}^\top+
\eta(\boldsymbol{C}-\boldsymbol{C}^\top),
\label{eq:app-exp-W-eta}
\end{equation}
with
\begin{equation}
\eta\in\{0,0.05,0.10,0.20,0.40,0.60,0.80\}.
\label{eq:app-exp-eta-grid}
\end{equation}
The off-diagonal symmetric support \(\boldsymbol{C}+\boldsymbol{C}^\top\) was
held fixed, and \(\eta\) scaled the antisymmetric residue
\(\boldsymbol{C}-\boldsymbol{C}^\top\). For each sampling seed, a single global
lazy-normalization constant was chosen as
\begin{equation}
\epsilon
=
\frac{0.88}{\max_{\eta}\max_i\sum_{j\ne i}W_{ij}^{\eta}}.
\label{eq:app-exp-fixed-epsilon}
\end{equation}
The raw asymmetric passive kernel was then
\begin{equation}
P_\eta(j\mid i)=\epsilon W_{ij}^{\eta}\quad (j\ne i),
\qquad
P_\eta(i\mid i)=1-\epsilon\sum_{j\ne i}W_{ij}^{\eta}.
\label{eq:app-exp-P-eta}
\end{equation}
This fixed-\(\epsilon\) rule keeps the maximum off-diagonal row mass at
\(0.88\), so every self-loop is nonnegative. The diagonal lazy term changes only
to keep each row stochastic. It contributes no antisymmetric part.

The reversible reference for the same seed was \(\boldsymbol{P}_0\), obtained by setting
\(\eta=0\). For each \(\eta\), the antisymmetric perturbation in the uniform
matrix coordinate was
\begin{equation}
\boldsymbol{E}_{\rm asym}^{\eta}
=\frac{1}{2}\left[(\boldsymbol{P}_\eta-\boldsymbol{P}_0)
-(\boldsymbol{P}_\eta-\boldsymbol{P}_0)^\top\right]
=
\eta\epsilon(\boldsymbol{C}-\boldsymbol{C}^\top).
\label{eq:app-exp-E-asym-eta}
\end{equation}
The spectral norm \(\|\boldsymbol{E}_{\rm asym}^{\eta}\|_2\) was computed with a
largest-singular-value sparse solver. The sweep used
\begin{equation}
\alpha\in\{0.970,0.985,0.993,0.997\}.
\label{eq:app-exp-asym-alpha-grid}
\end{equation}
Together with the five seeds and seven \(\eta\) values, this gave \(140\) raw
asymmetric conditions.

For each raw condition, the Dirichlet field \(\boldsymbol{z}_\eta^g\) was solved
from Eq.~\eqref{eq:app-exp-lmdp-solve} with \(\boldsymbol{P}_\eta\) in place of
\(\boldsymbol{P}_0\). The reference field \(\boldsymbol{z}_0^g\) used the same seed and
\(\alpha\), with \(\eta=0\). The field error was
\(\|\boldsymbol{z}_\eta^g-\boldsymbol{z}_0^g\|_2/\|\boldsymbol{z}_0^g\|_2\). The
readout-angle error was the median angle across \(148\) probe states. Probe
states were selected on a stride-seven grid, excluding states within five graph
steps of the goal, and the set was augmented by the start and three corridor or
chamber states \((52,25)\), \((84,43)\), and \((96,24)\). This probe set tests
readout changes away from the clamped goal, where behavior is chosen.

Rollouts in this experiment used the pure population-vector rule with
\(\sigma_G=3.2\), \(w_f=w_\ell=0\), maximum length \(1200\), and goal-zone radius
\(1.5\). Route ratio was normalized by the same shortest path to the goal zone,
\(110.43\), used in the bug-trap BTSP task.

The bilateral replay control projected each raw asymmetric transition memory
back to its reversible part:
\begin{equation}
\frac{1}{2}\left(\boldsymbol{W}^{\eta}+(\boldsymbol{W}^{\eta})^\top\right)
=\boldsymbol{C}+\boldsymbol{C}^\top.
\label{eq:app-exp-bilateral-projection}
\end{equation}
After the same fixed-\(\epsilon\) lazy normalization, this control is exactly
\(\boldsymbol{P}_0\) for the corresponding seed. It represents a circuit in which forward
and reverse replay, or other bilateral sampling, cancels the directional
residue before the recurrent map is used, consistent with bidirectional
hippocampal replay motifs \citep{PfeifferFoster2013,WidloskiFoster2022}. The
control was plotted at the same raw \(\chi\) values, so it asks whether the
same asymmetric experience can be rescued by a reversible projection.

Figure~\ref{fig:app-asymmetric-btsp} is organized as a mechanistic chain. The
perturbation knob \(\eta\) scales only the directional residue in
Eq.~\eqref{eq:app-exp-W-eta}, while the quantitative axis used later is \(\chi\),
the perturbation size after multiplication by the recurrent resolvent gain
\(\alpha/(1-\alpha)\). This matters because the same raw asymmetry is more
dangerous when \(\alpha\) is closer to one.

Panels A--C show one representative seed at \(\alpha=0.985\). The displayed
conditions are \(\eta=0\), \(\eta=0.10\), and \(\eta=0.60\), with
\(\chi=0.00\), \(0.19\), and \(1.11\). Panel A is the reference log field
\(\log\boldsymbol{z}_0^g\) with the reference rollout. Panels B and C show the
signed difference
\(\Delta\log\boldsymbol{z}_\eta^g=\log\boldsymbol{z}_\eta^g-\log\boldsymbol{z}_0^g\),
with red indicating that the asymmetric field is higher than the reference and
blue indicating that it is lower, together with the rollout on the asymmetric
field (white) and the reference rollout (grey). The perturbation visibly twists
the field, and the twist grows continuously with \(\chi\). At \(\eta=0\) and
\(\eta=0.10\), the rollout reaches the goal with route ratio \(1.02\). At
\(\eta=0.60\), the rollout still reaches the goal, with route ratio \(1.13\),
after taking a visibly shifted route. Complete raw failures occur in the larger
sweep shown in panel F.

Panels D and E test the scaling law. Color encodes \(\alpha\), each point is one
seed--\(\eta\) condition, and both panels use a logarithmic \(\chi\) axis. The
field error and the median readout-angle error both collapse onto single
increasing trends when plotted against \(\chi\). In the green band
\(\chi\le0.3\), the raw asymmetric branch has median field error
\(0.0040\), \(75\)th-percentile field error \(0.0070\), and median readout-angle
error \(0.13^\circ\). All raw rollouts in this band succeed, with median route
ratio \(1.03\). For \(1\le\chi<2\), median field error rises to \(0.053\), median
readout-angle error to \(1.88^\circ\), and median route ratio to \(1.14\); two
raw rollouts fail at \(\alpha=0.997\), \(\eta=0.20\), with
\(\chi\approx1.7\). For \(2\le\chi<8\), median field error is \(0.088\), median
readout-angle error is \(3.88^\circ\), and median route ratio is \(1.24\). These
numbers show that the theory gives more than a qualitative robustness claim: it
identifies the scale on which deviations should be compared across different
\(\alpha\) and \(\eta\) values.

Panel F gives the biological interpretation of the same scale. The green band
marks the low-\(\chi\) region used as the biological safe zone, motivated by
small measured directional imbalance in BTSP induction and BTSP-like recurrent
plasticity \citep{Bittner2017,Li2024}. In this region, raw asymmetric BTSP
produces fields and rollouts close to the reversible reference. Outside the
band, raw replay degrades and can fail. The bilateral
projection in Eq.~\eqref{eq:app-exp-bilateral-projection} keeps every rollout
successful across the full sweep, with route ratios between \(1.02\) and
\(1.08\). Thus the same plot carries two conclusions. Moderate asymmetry leaves
the theorem's prediction effectively intact, and bidirectional replay provides
a concrete circuit-level recovery mechanism when raw directional residue grows
large.

\begin{figure}[!t]
\centering
\includegraphics[width=\linewidth]{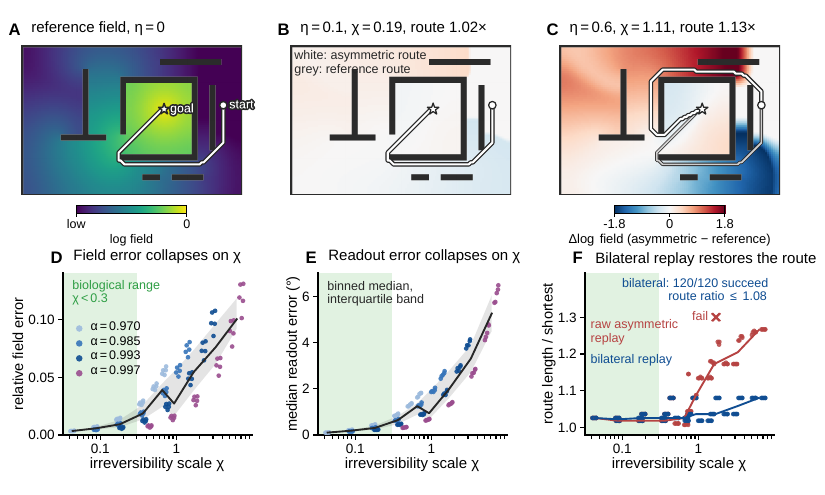}
\caption{\textbf{Asymmetric BTSP deviations are controlled by a single
irreversibility scale.}
The perturbation knob \(\eta\) scales the directional residue of the effective
BTSP memory; the relevant error scale is
\(\chi=\alpha\|\boldsymbol{E}_{\rm asym}\|_2/(1-\alpha)\), which combines the raw
asymmetry with the recurrent amplification.
(\textbf{A}) Reference field \(\log\boldsymbol{z}_0^g\) in the bug-trap maze at
\(\alpha=0.985\) with the reference route.
(\textbf{B}, \textbf{C}) Signed deviations
\(\Delta\log\boldsymbol{z}_\eta^g=\log\boldsymbol{z}_\eta^g-\log\boldsymbol{z}_0^g\)
for \(\eta=0.1\) (\(\chi=0.19\)) and \(\eta=0.6\) (\(\chi=1.11\)); red means the
asymmetric field is higher than the reference. The white route is the
population-vector rollout on the asymmetric field, the grey route the reference;
low \(\chi\) leaves the route unchanged (ratio \(1.02\)), larger \(\chi\)
produces a visible detour (ratio \(1.13\)).
(\textbf{D}) Relative field error collapses when plotted against \(\chi\)
across seeds, discounts (color), and \(\eta\) values; the black curve is the
binned median and the grey band the interquartile range. The green band marks
\(\chi<0.3\).
(\textbf{E}) The median readout-angle error follows the same \(\chi\)-ordered
trend.
(\textbf{F}) Raw asymmetric replay lengthens routes at large \(\chi\) and fails
twice (crosses); bilateral replay projects the same experience back to the
reversible support and keeps every route within \(1.08\times\) the shortest.}
\label{fig:app-asymmetric-btsp}
\end{figure}

\section{Additional experiments: scope, robustness, and comparisons}
\label{app:exp-extended}

This section reports a suite of additional experiments (Exp.~01--35) that test the assumptions of Table~\ref{tab:scope}, measure the realized relaxation dynamics, and compare the circuit with learned goal-conditioned models at matched experience; they supply the numbers quoted in Sec.~\ref{sec:exp-scope} and Fig.~\ref{fig:scope}. All experiments run from the released code with fixed seeds and single-threaded numerical libraries (App.~\ref{app:exp-reproduction}); every condition, including failures, is retained in the result files. Unless stated otherwise, ``field error'' is the relative \(L_2\) distance between the recurrent field and the direct LMDP solution in the same coordinate, ``policy TV'' is the total-variation distance between the branch-current policy of Theorem~\ref{thm:equiv} and the optimal policy \(\mu^*\), and ``stochastic arrival'' counts rollouts that sample the policy and reach the exact goal within a \(4{,}096\)-step horizon. Table~\ref{tab:app-extended} lists the headline result of each experiment; the subsections give protocols and details.

\begin{table}[!t]
\centering
\footnotesize
\setlength{\tabcolsep}{3.5pt}
\renewcommand{\arraystretch}{1.0}
\begin{tabular}{@{}l>{\raggedright\arraybackslash}p{0.30\linewidth}>{\raggedright\arraybackslash}p{0.62\linewidth}@{}}
\toprule
Exp. & Question & Headline result \\
\midrule
01 & realized relaxation dynamics & slow-rate relative error \(\le1.8\times10^{-10}\); monotone energy in every ODE condition \\
02 & exact policy and readout & field error \(\le3.8\times10^{-16}\); branch-current policy TV \(\le2.8\times10^{-16}\) \\
03 & directed support & field error \(\le6.8\times10^{-16}\); policy TV \(\le2.8\times10^{-16}\) (18 conditions) \\
04 & overlapping population code & full-rank normalized-field error \(\le1.4\times10^{-14}\) \\
05 & local opening and closure events & opening learned to zero field error after 40 balanced events; blocked-transition gate reroutes \(8/8\) \\
06 & state-dependent costs & field/policy errors \(\le3.8\times10^{-16}/5.2\times10^{-16}\) \\
07 & edit algebra and storage & \(\operatorname{rank}(\Delta\boldsymbol{R})=\operatorname{rank}(\Delta\boldsymbol{S})\) in every edit; Woodbury error \(\le2.6\times10^{-14}\) \\
08 & asymmetry coordinate audit & closed-form and measured antisymmetric coordinates agree to \(4.4\times10^{-15}\) \\
09 & transition-support errors & one spurious wall edge gives mean safe arrival \(0.86\); a count-matched missing edge keeps \(1.00\) \\
10 & structured perturbations & first-order prediction Spearman \(\rho=0.9998\) over 360 conditions \\
11 & strong-recall shells & greedy arrival \(103{,}190/103{,}190\); shell \(\le5.2\%\) of states \\
12 & horizon--latency law & horizon versus 5\% latency Spearman \(\rho=0.990\) \\
13 & amortized caches versus exact query & exact recurrent query \(768/768\) held-out rollouts before and after an edit; frozen spectral caches \(864/3{,}072\) \\
14 & irregular place-center graphs & \(702/702\) collision-free routes; field/policy errors \(\le1.3\times10^{-15}/4.0\times10^{-16}\) \\
15 & matched-experience representations & sparse operator and SR/SF fields differ by \(\le2.7\times10^{-15}\); all models \(4{,}096/4{,}096\) at the highest budget \\
16 & place-cell dropout & \(1{,}080/1{,}080\) routes through 80\% dropout, uniform and clustered \\
17 & gain heterogeneity and activity noise & gain success \(2{,}185{,}200/2{,}185{,}200\); noise success \(2{,}884{,}382/2{,}913{,}600\) state starts \\
18 & anisotropic population kernels & moment matching improves \(24/24\) condition medians; routes \(26/26\) \\
19 & directed acquisition & direction-selective rule: zero illegal entries, minimum safe arrival \(1.000\) \\
20 & continuum refinement & monotone field convergence \(4/4\); routes \(12/12\) \\
21 & anytime action and warm starts & median current-action speedup \(82\times\) (\(\alpha\ge0.985\)); pre-edit initial error \(\le1.3\times10^{-4}\) \\
22 & learned architecture family & tabular KL-Q policy error \(\le1.4\times10^{-15}\); width-1,024 goal-coordinate networks \(38{,}760/49{,}152\) held-out arrivals \\
23 & experience-driven learners & recurrent operator \(40{,}960/40{,}960\) arrivals across all budgets; sampled learners \(39\)--\(40\%\) on the bug trap at \(10^6\) transitions \\
24 & fixed-support dropout & field error \(\le1.4\times10^{-14}\); no new survivor edges \\
25 & gain boundary factorial & \(13{,}930{,}648/13{,}930{,}650\) state starts arrive through 30\% heterogeneity \\
26 & natural directed learning & highest-exposure safe routes \(160/160\) \\
27 & passive acquisition of an opening & \(30/30\) routes, zero field and policy discrepancy at \(320{,}000\) transitions \\
28 & amplitude--noise range laws & range-coordinate median Spearman \(\rho=0.933\) \\
29 & kernel validity surface & moment correction improves \(414/512\) condition--goal cells \\
30 & \(\epsilon\) scaling and large maps & identity through \(38{,}556\) states (\(1.2\times10^{-14}\)); routes \(3/3\) \\
31 & irregular continuum convergence & field routes \(2{,}880/2{,}880\); recovered screening error \(\le1.8\times10^{-11}\) \\
32 & goal-conditioned actor--critic & \(10^6\)-transition held-out arrival \(25{,}450/30{,}720\); zero illegal execution mass \\
33 & state-matched kernel policy & policy TV and KL exactly 0; vector routes \(8/8\) \\
34 & sparse acquisition after dropout & native routes \(475/480\) through 30\% and \(689/720\) through 50\% dropout \\
35 & time-varying stochastic dynamics & fixed-CV arrival \(600/600\) at the noise ceiling; fixed-Fano range law \(\rho=0.966\)--\(0.979\) \\
\bottomrule
\end{tabular}
\caption{Headline results of the additional experiments of App.~\ref{app:exp-extended}; experiment numbers follow the released code.}
\label{tab:app-extended}
\end{table}

\subsection{Realized relaxation dynamics, latency, and anytime readiness}
\label{app:exp-extended-latency}

Exp.~01 integrates the continuous-time dynamics of Eq.~\eqref{eq:linear} instead of solving its fixed point. On the \(1{,}472\)-state diagnostic map the measured slowest relaxation rate agrees with the spectral prediction \(\delta_g/\tau\), where \(\delta_g\) is the spectral gap of \(\boldsymbol{I}-\alpha\boldsymbol{S}_{\mathcal I_g\mathcal I_g}\), to relative error \(1.8\times10^{-10}\), and the energy \(\tfrac12\boldsymbol{u}^\top(\boldsymbol{I}-\alpha\boldsymbol{S})\boldsymbol{u}-\boldsymbol{b}_g^\top\boldsymbol{u}\) decreases monotonically in every recorded condition, as required by the gradient-flow form of the symmetric dynamics. Exp.~12 sweeps \(\alpha\in\{0.9,0.94,0.97,0.98,0.985,0.99,0.993\}\) over two maps and three goal strata (28 conditions). The cold-start latency to \(5\%\) field convergence spans \(17\)--\(304\,\tau\) and follows the discounted horizon \(\alpha/(1-\alpha)\) with Spearman \(\rho=0.990\); the identity \(T_{\rm slow}\delta_g/\tau=1\) holds to \(1.1\times10^{-16}\). Warm starts shorten the transient: reusing the field of a nearby goal gives a median speedup of \(1.17\times\) (\(1.36\times\) for \(\alpha\ge0.97\)), and after a local edit the pre-edit field begins inside the \(1\%\) criterion in all seven edit conditions (maximum initial error \(0.008\)).

Exp.~21 extends the sweep to seven \(\alpha\) values on both maps (42 conditions) and separates whole-field convergence from the readiness of the action at the current state. The current-state argmax is defined in all 42 conditions; in the 24 conditions with \(\alpha\ge0.985\) the current action is correct before the whole field has converged to \(5\%\) in \(24/24\) cells, with median speedup \(82.4\times\) and minimum \(1.99\times\) (Fig.~\ref{fig:scope}D shows the bug trap at \(\alpha=0.993\): \(11\,\tau\) versus \(331\,\tau\)). Measured slow rates agree with the spectral prediction to \(1.9\times10^{-12}\). Under structured operator perturbations every system remains stable, the actual field response follows the inverse spectral gap (\(\rho=0.974\)) and the first-order prediction (\(\rho=0.997\)), and branch-current policy TV stays at most \(0.002\) in condition mean and \(0.015\) statewise. The pre-edit field starts inside the \(5\%\) and \(1\%\) criteria in all 42 conditions (maximum initial error \(1.3\times10^{-4}\)).

\subsection{Strong recall: distance-resolved shells and behavior}
\label{app:exp-extended-recall}

Exp.~11 extends the recall-amplitude sweep of Sec.~\ref{sec:exp-linearization} to \(A\in[0.1,100]\) on both maps with \(\phi=\tanh\), evaluating every non-goal state. Greedy rollouts of the nonlinear policy reach the exact goal in \(103{,}190/103{,}190\) state starts, the same count as the linear reference, and the complete route-ratio distribution is identical between the two policies in every condition. Tie-aware first-action agreement is \(103{,}142/103{,}190\) (minimum condition-level agreement \(99.90\%\)). The shell of distorted field occupies at most \(5.21\%\) of non-goal states (at \(A=100\): \(158\) states on the diagnostic map and \(225\) on the bug trap, maximum geodesic radii \(8.7\) and \(9.7\) cells). Outside the shell the maximum mean policy TV is \(1.6\times10^{-4}\) and the largest statewise TV is \(0.003\); the largest single-state TV, \(0.76\), occurs inside the shell adjacent to the goal. Stochastic rollouts started inside the shell are captured by the goal in \(33{,}792/33{,}792\) cases, and the distance-stratified stochastic assay reaches \(327{,}680/327{,}680\). Fig.~\ref{fig:equivalence}F plots the shell fraction and arrival against \(A\).

\subsection{Gain heterogeneity, activity noise, and operator perturbations}
\label{app:exp-extended-robustness}

Exp.~10 perturbs the operator in 18 structured directions (360 conditions, 288 with nonzero amplitude). All perturbed systems remain stable; the first-order resolvent predictor of the field response has Spearman \(\rho=0.9998\) with the realized error, so the relevant diagnostic is the perturbation structure relative to the goal-clamped stability margin. Across 40 \(\tanh\) conditions the residual-resolvent identity holds to \(3.9\times10^{-13}\); under the strongest saturation the raw field error reaches \(0.69\) while the branchwise policy TV has condition mean \(0.006\) and worst-state value \(0.22\), and we retain this regime as the boundary of the linear description.

Exp.~17 and Exp.~25 draw heterogeneous unit gains \(\rho_i\) from three spatial families at six depths and three \(\alpha\) values on both maps, with 30 paired realizations per cell (2,700 conditions). Through \(30\%\) gain modulation, \(13{,}930{,}648\) of \(13{,}930{,}650\) state starts reach the goal (minimum condition success \(99.973\%\)), with median field error \(0.020\) and policy TV \(0.006\); the first-order prediction of the field response has pooled \(\rho=0.962\). In Exp.~17 the gain assay reaches \(2{,}185{,}200/2{,}185{,}200\) starts and the activity-noise assay \(2{,}884{,}382/2{,}913{,}600\). Exp.~28 and Exp.~35 vary the goal amplitude over three decades together with readout noise, dynamic noise, synaptic jitter, and fresh-innovation noise in the recurrent dynamics (50 fixed streams). Under \(1\%\) multiplicative synaptic jitter the maximum condition-mean policy TV is \(0.004\) and the minimum all-state success \(91.9\%\); the logarithmic range law that links amplitude, noise, and the distance over which the readout is reliable has median Spearman \(\rho=0.933\) (Exp.~28) and \(\rho=0.966\)--\(0.979\) across map--\(\alpha\) cells (Exp.~35). With fixed coefficient of variation the policy is amplitude invariant (\(600/600\) arrivals at the noise ceiling); with fixed Fano factor, raising the amplitude from \(0.3\) to \(30\) reduces the median policy TV from \(0.126\) to \(0.049\) and raises median all-state arrival from \(45\%\) to \(68\%\), the expected signature of a signal that grows faster than its noise.

\subsection{Exact policy readout and anisotropic population kernels}
\label{app:exp-extended-readout}

Exp.~02 evaluates the branch-current policy of Theorem~\ref{thm:equiv} on the \(3{,}036\)-state diagnostic map and the \(7{,}285\)-state bug trap: the recurrent field matches the direct LMDP solution to \(3.8\times10^{-16}\), the local synaptic readout \(S_{ij}\psi_j\) matches \(\mu^*\) to \(2.8\times10^{-16}\) in TV, and the Bellman residual is at most \(1.7\times10^{-16}\). The source-at-current-state query and the goal clamp of Fig.~\ref{fig:concept}D are verified on the same operator. For the population-vector readout of Eq.~\eqref{eq:popvec}, the raw readout on the bug trap has median angular error \(9.3^\circ\) relative to \(\mu^*\); near walls the flat-field kernel drift has median norm \(1.5\) grid units over \(1{,}932\) states, and subtracting this goal-independent baseline (the centering of App.~\ref{app:popvec-raw-rate}) reduces the error to median \(3.9^\circ\) and p90 \(26.8^\circ\) near walls. Raw and centered population readouts reach the goal with route ratios \(1.033\) and \(1.034\); Dijkstra gives \(1.000\), and the Euclidean-direction baseline terminates in a repeated cycle.

Exp.~18 applies determinant-preserving anisotropic Gaussian kernels with condition numbers \(\{1,2,5,10,50\}\) and orientations \(\{0^\circ,45^\circ,90^\circ\}\) at every non-goal state of both maps (\(134{,}147\) state--condition rows). The goal-independent moment correction \(\boldsymbol{\Sigma}_{\rm passive}\boldsymbol{\Sigma}_{\rm kernel}^{+}\Delta\boldsymbol{x}_{\rm centered}\), which uses only the kernel and passive-transition covariances, improves the condition-level median angle in all \(24/24\) anisotropic cells, from a raw range of \(0.4^\circ\)--\(55.4^\circ\) to \(0.7^\circ\)--\(15.3^\circ\); raw readout reaches the goal in \(23/26\) conditions, the moment-matched readout in \(26/26\) with median route ratio \(1.034\). Exp.~29 maps the validity surface over 64 width--condition cells and four goals per map and finds the correction helpful in \(414/512\) regular cells. Exp.~33 uses the population kernel whose second moment matches the passive successor moments at every state: its policy agrees with \(\mu^*\) with TV and KL exactly zero and reaches the goal in \(8/8\) vector routes. Locality, centering, accessible support, and moment matching are therefore the conditions under which Eq.~\eqref{eq:popvec} reproduces the optimal policy.

\subsection{Directed transitions}
\label{app:exp-extended-directed}

Exp.~03 builds a \(25\times17\) map (\(425\) states) with a wall carrying two one-way gates, one passable only left to right and one only right to left. For three positive diagonals \(\boldsymbol{D}\) (identity, an exponential axis gradient, and a spatial wave), two goals, and \(\alpha\in\{0.85,0.93,0.97\}\) (18 conditions), the recurrent realization \(\boldsymbol{S}=\boldsymbol{D}\boldsymbol{P}_0\boldsymbol{D}^{-1}\) of Corollary~\ref{cor:directed} reproduces the directed LMDP field to relative error \(6.8\times10^{-16}\) and the branch-current policy to \(2.8\times10^{-16}\) in TV, with Bellman residual \(2.2\times10^{-16}\); integrating the ODE from zero activity reaches the same field (final error \(8.4\times10^{-5}\) at the recorded time) with no monotonicity violation. The symmetric projection of the same map assigns \(0.52\)--\(0.71\) of its probability mass to an illegal reverse crossing of a gate (Fig.~\ref{fig:scope}A).

Exp.~19 and Exp.~26 ask whether the directed operator is acquired from experience. A direction-selective rule that writes the synapse \(j\to i\) for each executed transition \(i\to j\) (App.~\ref{app:lmdp-directed-corollary}) produces zero illegal reverse entries in every condition; from one to 32 events per entry the median field error falls from \(0.031\) to \(0.006\) and the median policy TV from \(0.026\) to \(0.005\), with minimum safe-arrival probability \(1.000\). Natural random walks on two directed environments reach full legal-edge coverage at the highest exposure, with median field error \(0.025\) and policy TV \(0.041\); routes are safe in \(160/160\) conditions and arrive in \(158/160\). A control that mixes in reverse-oriented writing creates up to 62 illegal entries and 167 unsafe routes, so the orientation of the plasticity rule, not the exposure, decides whether one-way structure is respected.

\subsection{Overlapping population codes}
\label{app:exp-extended-codes}

Exp.~04 replaces the one-unit-per-location coordinate by two nonnegative full-rank overlapping encoders (minimum rank \(140\)) and evaluates three goals each. In the recoded population the normalized desirability field is reproduced to \(1.4\times10^{-14}\), the source response to \(1.2\times10^{-14}\), and the policy to \(7.3\times10^{-14}\) in TV, and the edit-rank inequality of App.~\ref{app:woodbury} holds in every condition: the computation is invariant under full-rank recoding, as stated in Sec.~\ref{sec:scope}. Undercomplete codes are a different matter. Over a factorial of seven feature counts, four widths, three \(\alpha\) values, three goal strata, and ten paired layouts (2,520 conditions), only \(850\) decoded fields remain positive everywhere, and among those the oracle-assisted route reaches the goal in \(444/850\); exact transformed operators are dense with at least \(47\%\) negative weights. The invariance therefore holds for full-rank codes; compression below full rank requires additional structure.

\subsection{Local opening and closure events and transition-support errors}
\label{app:exp-extended-edits}

Exp.~05 streams local events through a \(1{,}472\)-state map with a ten-edge doorway interface. A positive transition-event rule that adds \(25\%\) of the native weight per observed crossing learns a newly opened doorway to median field error \(4.1\times10^{-4}\) after 16 events and to exactly zero after 40 balanced events (four per doorway edge). Exp.~27 confirms the same acquisition from goal-free random walks in the true open map: at \(320{,}000\) transitions every one of 30 trajectories covers the door edges, the recurrent field agrees with the learned-belief LMDP field within \(5.2\times10^{-15}\), and field error and policy TV relative to the true open map are exactly zero (Fig.~\ref{fig:btsp}F). Closure requires evidence of blocking: without observed transitions the field is unchanged (error \(0.0017\) before and after 16 elapsed intervals), and the positive-only rule succeeds in \(0/8\) streaming episodes. A blocked-transition gate that masks the annotated interface reroutes \(8/8\) episodes, and a signed depotentiation extension, motivated by the bidirectionality of BTSP \citep{Milstein2021}, also reaches \(8/8\) but needs a median of \(40\) blocked attempts and \(540\) decision steps with median route ratio \(11.6\); the gate is the fast mechanism, and signed LTD a slow proof of principle.

Exp.~07 audits the low-rank algebra. The shared propagation correction obeys \(\operatorname{rank}(\Delta\boldsymbol{R})=\operatorname{rank}(\Delta\boldsymbol{S})\) in every edit and the Woodbury identity of App.~\ref{app:woodbury-shared-full-resolvent} to \(2.6\times10^{-14}\); rescaling each clamp column changes the policy by at most \(7.0\times10^{-17}\). The raw unit-clamped all-goal field change can be numerically full rank because every goal column is renormalized, while its \(99\%\)-energy rank is \(4\)--\(7\); ``low rank'' therefore refers to the exact propagation correction, and the \(99.5\%\) of Fig.~\ref{fig:btsp}E is the empirical effective compression. Exp.~09 compares false-positive and false-negative edges around obstacles on one fixed map: a single spurious wall edge carries mean maximum illegal action mass \(0.33\) and reduces the mean safe-arrival probability to \(0.86\) (range \(0.79\)--\(0.91\) over 20 anchors), whereas a count-matched missing legal edge keeps it at \(1.00\). Precision of the transition gate matters more than recall in this setting, which is the regime in which BTSP operates (precision \(1.00\) in Fig.~\ref{fig:btsp}A).

\subsection{State-dependent costs}
\label{app:exp-extended-costs}

Exp.~06 adds a state-dependent cost \(q(i)\) to the LMDP objective, realized as a state-dependent gain in the recurrent operator, for one \(2{,}320\)-state map, three goals, and five hazard strengths. The recurrent field matches the normative solution to \(3.8\times10^{-16}\), the policy to \(5.2\times10^{-16}\), and the Bellman equation to \(3.3\times10^{-16}\); within each cost strength the operator is identical across goals, and only the clamped unit changes. At the strongest hazard the expected exposure falls to \(1.6\%\), \(0.9\%\), and \(1.6\%\) of its constant-cost value for the near, mid, and far goals while the expected step counts change by factors \(1.31\), \(1.13\), and \(0.99\). Terrain, effort, and risk costs therefore enter the same linear family, as anticipated in Sec.~\ref{sec:discussion}.

\subsection{Irregular place-center graphs and the continuum limit}
\label{app:exp-extended-continuum}

Exp.~14 replaces the grid by irregular graphs of place centers in continuous open and barrier arenas, with edges between centers whose connecting segment does not cross a barrier. Across the arena specifications, \(702/702\) population-vector routes reach the goal without touching a barrier, and field and policy errors relative to the LMDP on the same graph are at most \(1.3\times10^{-15}\) and \(4.0\times10^{-16}\) (Fig.~\ref{fig:scope}B shows a \(1{,}024\)-center barrier arena). Exp.~20 refines regular grids toward a fixed screened-Poisson continuum problem: the field error decreases monotonically in \(4/4\) curves with empirical convergence order \(1.64\)--\(1.69\), the coarsest-to-finest error falls by at least \(9.7\times\), and all \(12/12\) routes arrive with maximum route ratio \(1.08\). Exp.~31 repeats the refinement on irregular layouts (20 seeds, six goals, three sizes, two layouts, two arenas, two kernel widths): field routes arrive in \(2{,}880/2{,}880\) conditions and moment-readout routes in \(2{,}860/2{,}880\), the recovered screening constant differs from \(\kappa\) by at most \(1.8\times10^{-11}\), and field and readout converge with median orders \(1.45\) and \(1.63\).

\subsection{Partial coverage by place cells}
\label{app:exp-extended-coverage}

Exp.~16 removes place cells from both maps, uniformly or in spatial clusters, at nested fractions up to \(80\%\), and reconstructs the survivor graph from the geometry of the remaining centers. All \(1{,}080/1{,}080\) routes reach the exact goal (Fig.~\ref{fig:scope}C). Exp.~24 keeps the original support and only deletes units, redirecting the deleted mass to self-loops or renormalizing survivor rows: the field identity holds to \(1.4\times10^{-14}\) with no new survivor edges, and connected routes are retained through median prefixes of \(30\)--\(40\%\) loss, after which the fixed support disconnects the protected assemblies; geometry-aware reconstruction is what carries Exp.~16 through \(80\%\). Exp.~34 writes the sparse operator from scratch by goal-free BTSP after dropout: native learned-edge coverage has median \(1.000\) and minimum \(0.992\), the operator error has median \(0.019\), and native routes arrive in \(360/360\) conditions through \(20\%\) dropout, \(475/480\) through \(30\%\), and \(689/720\) through \(50\%\) (884/1,080 over the full \(0\)--\(80\%\) grid). A degree-stratified relabeling control yields \(0/1{,}080\) arrivals, so the result reflects the learned geometry and not the degree distribution.

\subsection{Comparisons with cached and learned goal-conditioned models}
\label{app:exp-extended-baselines}

Four families were compared with the recurrent operator at matched experience. (i) \emph{Oracles.} The direct LMDP solution is the ground truth of Theorem~\ref{thm:equiv} and is matched to machine precision (Exp.~02). A Dijkstra oracle gives route ratio \(1.000\) on the bug trap; the population readout gives \(1.033\); the Euclidean-direction baseline terminates in a repeated cycle. Dijkstra is a geometric shortest-path oracle, not the finite-temperature KL objective. (ii) \emph{Successor representations.} Exp.~15 builds successor features from the same transition streams: the sparse operator and the full-state SR/SF fields differ by at most \(2.7\times10^{-15}\), both reach \(4{,}096/4{,}096\) stochastic arrivals at the highest budget with median policy TV \(0.0019\), and one shared closure event raises both from \(831/2{,}048\) to \(2{,}048/2{,}048\) on the edited map. The difference is storage: the operator keeps \(\mathcal{O}(N)\) synapses and answers any goal by one relaxation, while a dense all-goal cache grows as \(N^2\) (\(252\times\) more float64 storage at \(3{,}036\) states, \(3{,}116\times\) at the largest map of Exp.~30). (iii) \emph{Frozen caches and supervised readouts.} Exp.~13 fixes a goal split (32 training, 8 validation, 12 held-out goals) on the \(3{,}036\)-state map and registers an edit closing a near gate (32 weight entries). The exact recurrent query reaches the goal in \(768/768\) paired held-out rollouts before and after the edit; validation-selected RBF readouts trained on 8/16/32 goal fields reach \(1{,}357/2{,}304\) overall and improve smoothly with goal supervision; the four frozen spectral caches of the original map reach \(864/3{,}072\) (\(28\%\)). (iv) \emph{Learned goal-conditioned models at absolute budgets.} Exp.~22 and Exp.~23 train goal-specific sampled soft-Q, held-out goal-coordinate TD value learners, tabular KL-Q, and shared goal-coordinate value networks (widths 256 and 1,024) on twenty streams per map with budgets from \(10^3\) to \(10^6\) transitions, evaluated by \(4{,}096\)-step stochastic rollouts against a same-seed exact-reference ceiling that reaches \(100\%\). The recurrent operator reaches \(40{,}960/40{,}960\) original-map arrivals across all budgets (median policy TV \(0.109\to0.001\) from \(10^3\) to \(10^6\)). At \(10^6\) transitions on the bug trap the tabular soft-Q learner reaches \(2{,}019/5{,}120\) (\(39.4\%\)) and goal-coordinate TD learners with 4/16/64 training goals \(40.3/40.2/40.0\%\), while on the easier two-door map all learners exceed \(98\%\). Tabular KL-Q, which fits the same linear family, matches the direct policy to \(1.4\times10^{-15}\) and reaches \(16{,}384/16{,}384\); width-1,024 goal-coordinate networks reach \(38{,}760/49{,}152\) held-out arrivals. Exp.~32 trains a sampled-count goal-conditioned soft actor--critic with a pooled reversible local operator: at \(10^6\) transitions on the bug trap it reaches \(48.0\%\), \(53.9\%\), and \(96.2\%\) of held-out-goal arrivals with 4, 16, and 64 training goals (\(99\)--\(100\%\) on the two-door map), with zero illegal execution mass; after 10 observed closure events on the two-door map its arrival falls to \(41\)--\(52\%\). The learned models approach the operator only with tens of training goals and lose the plan after a local edit, which is the per-goal supervision and the relearning that the zero-shot readout does not need.

\subsection{Scale and the lazy rate}
\label{app:exp-extended-scale}

Exp.~30 verifies the fixed-\(\epsilon\) construction of Eq.~\eqref{eq:lazy}: across 40 \(\epsilon\) cells on regular graphs the field error is at most \(6.6\times10^{-15}\) and the moving-policy TV at most \(1.2\times10^{-15}\), so any admissible \(\epsilon\) gives the same plan once \(\alpha\) is matched. On three maps beyond the previous \(7{,}285\)-state endpoint, up to \(38{,}556\) states, the recurrent and LMDP fields agree within \(1.2\times10^{-14}\) with policy TV \(9.7\times10^{-16}\), all \(3/3\) routes arrive, and the dense-to-sparse storage ratio reaches \(3{,}116\times\). Exp.~07 fits the sparse storage over four regular-grid sizes as \(N^{1.03}\) (\(96\) bytes per state at \(3{,}036\) states); bounded degree supplies the structural linear law (Fig.~\ref{fig:scope}F). The remaining finite cost is latency: on the \(7{,}285\)-state cold bug trap the field error reaches \(8\times10^{-4}\) by the last sampled time while the policy is available on \(94\%\) of states with mean TV \(0.027\), and the current-action readiness of App.~\ref{app:exp-extended-latency} is what makes the plan usable before the field has converged everywhere.

\section{Reproduction}
\label{app:exp-reproduction}

The code is released under the MIT license and has three parts: the core experiments of Sec.~\ref{sec:experiments} and App.~\ref{app:experiments} (\texttt{experiments/}, with result tables in \texttt{results/}), the additional experiments of App.~\ref{app:exp-extended} (\texttt{extended/}), and the scripts that draw the figures of this paper (\texttt{paper\_figures/}). Everything runs with Python 3.10 or later and
\begin{verbatim}
numpy>=1.24
scipy>=1.10
pandas>=2.0
matplotlib>=3.7
\end{verbatim}
plus \texttt{pytest} for the unit tests. NumPy, SciPy, and pandas are distributed under the BSD 3-Clause license, Matplotlib under the PSF-based Matplotlib license, and pytest under the MIT license; no other third-party code or data is used. The scripts set \texttt{OPENBLAS\_NUM\_THREADS}, \texttt{OMP\_NUM\_THREADS}, and \texttt{MKL\_NUM\_THREADS} to one unless these variables are already set. Randomness enters through NumPy's \texttt{default\_rng}; the seeds listed in this appendix fully determine the simulations.

From the repository root, the core experiments and their tables are reproduced by
{\footnotesize\begin{verbatim}
python -m experiments.exp1_equivalence --outdir results --figdir figures --seed 2
python -m experiments.exp2_btsp_navigation_lowrank --outdir results --figdir figures --seed 4
python -m experiments.exp3_green_limit --outdir results --figdir figures --seed 5
python -m experiments.exp4_asymmetric_btsp --outdir results --figdir figures --seed 4
\end{verbatim}}
which write
{\footnotesize\begin{verbatim}
results/exp1_equivalence_metrics.csv          results/exp2_lowrank_remap_metrics.csv
results/exp1_nonlinear_shell_metrics.csv      results/exp3_green_limit_metrics.csv
results/exp1_alpha_recall_sweep_metrics.csv   results/exp4_asymmetric_btsp_metrics.csv
results/exp2_btsp_learning_metrics.csv        results/exp4_asymmetric_btsp_summary.csv
results/exp2_challenge_navigation_metrics.csv
\end{verbatim}}
The additional experiments are run by
\begin{verbatim}
python -m extended.run_all
\end{verbatim}
which executes the unit tests, runs Exp.~01--35 with fixed seeds and single-threaded numerical libraries, reruns the four core experiments in an isolated directory and compares their result tables with the stored ones, and records the environment, the command lines, and a SHA-256 hash of every result file. A single experiment is run by its module name, for example
\begin{verbatim}
python -m extended.experiments.exp21_anytime_horizon_robustness
\end{verbatim}
Each experiment writes raw CSV files and a JSON summary to \texttt{extended/results/}; the numbers quoted in App.~\ref{app:exp-extended} are read from these files, and failed conditions are kept in the same files as successful ones.

Each figure is drawn by its own script in \texttt{paper\_figures/} (Fig.~\ref{fig:btsp}, for example, by \texttt{fig3\_btsp.py}), which recomputes the displayed fields, rollouts, and spectra from the experiment code and reads the sweep statistics from the result files above.

The stored tables were produced on macOS (arm64). Values reported at machine precision (\(10^{-16}\) to \(10^{-13}\)) vary at that level between platforms, and the pseudoinverse shape errors of App.~\ref{app:exp-green-limit} vary at the \(10^{-8}\) level. The two-mode energy fraction of Fig.~\ref{fig:btsp}E depends on how ties are broken when the \(34\) goals are selected by spatial quantiles, because NumPy's default sort is not stable and its vectorized x86 implementation orders ties differently; across tie-breaking orders the first two modes carry \(98.3\)--\(99.5\%\) of the energy (\(99.46\%\) with the reference order used in Fig.~\ref{fig:btsp}E). The README of the code release gives the setting that restores the reference order on x86 processors.


\clearpage
\DisableAppendixContents
\section*{NeurIPS Paper Checklist}

\begin{enumerate}

\item {\bf Claims}
    \item[] Question: Do the main claims made in the abstract and introduction accurately reflect the paper's contributions and scope?
    \item[] Answer: \answerYes{}
    \item[] Justification: The abstract and Sec.~\ref{sec:intro} state the specific contributions, and each is supported by the equivalence theorem and its directed corollary (Sec.~\ref{sec:equivalence}), the readout analysis (Sec.~\ref{sec:readout}), the BTSP learning model and low-rank update analysis (Sec.~\ref{sec:learning}), and the experiments of Sec.~\ref{sec:experiments}. The claims are scoped to the assumptions listed in Table~\ref{tab:scope} and to the simulated environments.
    \item[] Guidelines:
    \begin{itemize}
        \item The answer \answerNA{} means that the abstract and introduction do not include the claims made in the paper.
        \item The abstract and/or introduction should clearly state the claims made, including the contributions made in the paper and important assumptions and limitations. A \answerNo{} or \answerNA{} answer to this question will not be perceived well by the reviewers. 
        \item The claims made should match theoretical and experimental results, and reflect how much the results can be expected to generalize to other settings. 
        \item It is fine to include aspirational goals as motivation as long as it is clear that these goals are not attained by the paper. 
    \end{itemize}

\item {\bf Limitations}
    \item[] Question: Does the paper discuss the limitations of the work performed by the authors?
    \item[] Answer: \answerYes{}
    \item[] Justification: Sec.~\ref{sec:discussion} contains a dedicated ``Limitations'' paragraph (unspecified hippocampal subfield, idealized transition gate, the population-vector readout as an approximation of the exact branch-current policy), Table~\ref{tab:scope} lists every assumption with where it is relaxed or tested, and Apps.~\ref{app:linearization}, \ref{app:asymmetry}, \ref{app:perturbation}, and~\ref{app:exp-extended} state the perturbative regimes and the failure modes retained in the result files.
    \item[] Guidelines:
    \begin{itemize}
        \item The answer \answerNA{} means that the paper has no limitation while the answer \answerNo{} means that the paper has limitations, but those are not discussed in the paper. 
        \item The authors are encouraged to create a separate ``Limitations'' section in their paper.
        \item The paper should point out any strong assumptions and how robust the results are to violations of these assumptions (e.g., independence assumptions, noiseless settings, model well-specification, asymptotic approximations only holding locally). The authors should reflect on how these assumptions might be violated in practice and what the implications would be.
        \item The authors should reflect on the scope of the claims made, e.g., if the approach was only tested on a few datasets or with a few runs. In general, empirical results often depend on implicit assumptions, which should be articulated.
        \item The authors should reflect on the factors that influence the performance of the approach. For example, a facial recognition algorithm may perform poorly when image resolution is low or images are taken in low lighting. Or a speech-to-text system might not be used reliably to provide closed captions for online lectures because it fails to handle technical jargon.
        \item The authors should discuss the computational efficiency of the proposed algorithms and how they scale with dataset size.
        \item If applicable, the authors should discuss possible limitations of their approach to address problems of privacy and fairness.
        \item While the authors might fear that complete honesty about limitations might be used by reviewers as grounds for rejection, a worse outcome might be that reviewers discover limitations that aren't acknowledged in the paper. The authors should use their best judgment and recognize that individual actions in favor of transparency play an important role in developing norms that preserve the integrity of the community. Reviewers will be specifically instructed to not penalize honesty concerning limitations.
    \end{itemize}

\item {\bf Theory assumptions and proofs}
    \item[] Question: For each theoretical result, does the paper provide the full set of assumptions and a complete (and correct) proof?
    \item[] Answer: \answerYes{}
    \item[] Justification: Theorem~\ref{thm:equiv} states its assumptions, including the linear-regime assumption, in the theorem; the proof is in App.~\ref{app:lmdp-derivation}, the proof of Corollary~\ref{cor:directed} in App.~\ref{app:lmdp-directed-corollary}, and the proof of Proposition~\ref{prop:raw-rate} in App.~\ref{app:popvec-raw-rate}. The supporting derivations for linearization, reversibility, readout, nonlinear perturbation, Green-function limits, and low-rank updates are in Apps.~\ref{app:linearization}--\ref{app:woodbury}.
    \item[] Guidelines:
    \begin{itemize}
        \item The answer \answerNA{} means that the paper does not include theoretical results. 
        \item All the theorems, formulas, and proofs in the paper should be numbered and cross-referenced.
        \item All assumptions should be clearly stated or referenced in the statement of any theorems.
        \item The proofs can either appear in the main paper or the supplemental material, but if they appear in the supplemental material, the authors are encouraged to provide a short proof sketch to provide intuition. 
        \item Inversely, any informal proof provided in the core of the paper should be complemented by formal proofs provided in appendix or supplemental material.
        \item Theorems and Lemmas that the proof relies upon should be properly referenced. 
    \end{itemize}

    \item {\bf Experimental result reproducibility}
    \item[] Question: Does the paper fully disclose all the information needed to reproduce the main experimental results of the paper to the extent that it affects the main claims and/or conclusions of the paper (regardless of whether the code and data are provided or not)?
    \item[] Answer: \answerYes{}
    \item[] Justification: App.~\ref{app:experiments} specifies graph construction, maze geometries, solvers, tolerances, seeds, sweep values, metrics, and rollout rules; App.~\ref{app:exp-extended} gives the protocols of the additional experiments, and App.~\ref{app:exp-reproduction} the exact reproduction commands. The supplemental material contains the code, result files, and figure scripts.
    \item[] Guidelines:
    \begin{itemize}
        \item The answer \answerNA{} means that the paper does not include experiments.
        \item If the paper includes experiments, a \answerNo{} answer to this question will not be perceived well by the reviewers: Making the paper reproducible is important, regardless of whether the code and data are provided or not.
        \item If the contribution is a dataset and\slash or model, the authors should describe the steps taken to make their results reproducible or verifiable. 
        \item Depending on the contribution, reproducibility can be accomplished in various ways. For example, if the contribution is a novel architecture, describing the architecture fully might suffice, or if the contribution is a specific model and empirical evaluation, it may be necessary to either make it possible for others to replicate the model with the same dataset, or provide access to the model. In general. releasing code and data is often one good way to accomplish this, but reproducibility can also be provided via detailed instructions for how to replicate the results, access to a hosted model (e.g., in the case of a large language model), releasing of a model checkpoint, or other means that are appropriate to the research performed.
        \item While NeurIPS does not require releasing code, the conference does require all submissions to provide some reasonable avenue for reproducibility, which may depend on the nature of the contribution. For example
        \begin{enumerate}
            \item If the contribution is primarily a new algorithm, the paper should make it clear how to reproduce that algorithm.
            \item If the contribution is primarily a new model architecture, the paper should describe the architecture clearly and fully.
            \item If the contribution is a new model (e.g., a large language model), then there should either be a way to access this model for reproducing the results or a way to reproduce the model (e.g., with an open-source dataset or instructions for how to construct the dataset).
            \item We recognize that reproducibility may be tricky in some cases, in which case authors are welcome to describe the particular way they provide for reproducibility. In the case of closed-source models, it may be that access to the model is limited in some way (e.g., to registered users), but it should be possible for other researchers to have some path to reproducing or verifying the results.
        \end{enumerate}
    \end{itemize}

\item {\bf Open access to data and code}
    \item[] Question: Does the paper provide open access to the data and code, with sufficient instructions to faithfully reproduce the main experimental results, as described in supplemental material?
    \item[] Answer: \answerYes{}
    \item[] Justification: The supplemental material includes the full experiment code with README files and requirements, a one-command runner for the additional experiments, result CSV and JSON files, and the figure scripts; App.~\ref{app:exp-reproduction} gives the exact commands.
    \item[] Guidelines:
    \begin{itemize}
        \item The answer \answerNA{} means that paper does not include experiments requiring code.
        \item Please see the NeurIPS code and data submission guidelines (\url{https://neurips.cc/public/guides/CodeSubmissionPolicy}) for more details.
        \item While we encourage the release of code and data, we understand that this might not be possible, so \answerNo{} is an acceptable answer. Papers cannot be rejected simply for not including code, unless this is central to the contribution (e.g., for a new open-source benchmark).
        \item The instructions should contain the exact command and environment needed to run to reproduce the results. See the NeurIPS code and data submission guidelines (\url{https://neurips.cc/public/guides/CodeSubmissionPolicy}) for more details.
        \item The authors should provide instructions on data access and preparation, including how to access the raw data, preprocessed data, intermediate data, and generated data, etc.
        \item The authors should provide scripts to reproduce all experimental results for the new proposed method and baselines. If only a subset of experiments are reproducible, they should state which ones are omitted from the script and why.
        \item At submission time, to preserve anonymity, the authors should release anonymized versions (if applicable).
        \item Providing as much information as possible in supplemental material (appended to the paper) is recommended, but including URLs to data and code is permitted.
    \end{itemize}

\item {\bf Experimental setting/details}
    \item[] Question: Does the paper specify all the training and test details (e.g., data splits, hyperparameters, how they were chosen, type of optimizer) necessary to understand the results?
    \item[] Answer: \answerYes{}
    \item[] Justification: App.~\ref{app:experiments} gives maze construction, connectivity rules, kernel widths, discounts, recall amplitudes, seeds, solver tolerances, the exploration sampler, rollout scoring, and metrics; App.~\ref{app:exp-extended} gives the goal splits, budgets, horizons, and model configurations of the learned comparators.
    \item[] Guidelines:
    \begin{itemize}
        \item The answer \answerNA{} means that the paper does not include experiments.
        \item The experimental setting should be presented in the core of the paper to a level of detail that is necessary to appreciate the results and make sense of them.
        \item The full details can be provided either with the code, in appendix, or as supplemental material.
    \end{itemize}

\item {\bf Experiment statistical significance}
    \item[] Question: Does the paper report error bars suitably and correctly defined or other appropriate information about the statistical significance of the experiments?
    \item[] Answer: \answerYes{}
    \item[] Justification: Stochastic BTSP experiments report means and SEM across five exploration seeds (Sec.~\ref{sec:exp-btsp-navigation}, App.~\ref{app:exp-btsp}); the additional experiments report complete counts over all conditions, medians with interquartile ranges, and Wilson bounds where applicable (App.~\ref{app:exp-extended}). Deterministic algebraic checks report relative errors and tolerances instead of error bars.
    \item[] Guidelines:
    \begin{itemize}
        \item The answer \answerNA{} means that the paper does not include experiments.
        \item The authors should answer \answerYes{} if the results are accompanied by error bars, confidence intervals, or statistical significance tests, at least for the experiments that support the main claims of the paper.
        \item The factors of variability that the error bars are capturing should be clearly stated (for example, train/test split, initialization, random drawing of some parameter, or overall run with given experimental conditions).
        \item The method for calculating the error bars should be explained (closed form formula, call to a library function, bootstrap, etc.)
        \item The assumptions made should be given (e.g., Normally distributed errors).
        \item It should be clear whether the error bar is the standard deviation or the standard error of the mean.
        \item It is OK to report 1-sigma error bars, but one should state it. The authors should preferably report a 2-sigma error bar than state that they have a 96\% CI, if the hypothesis of Normality of errors is not verified.
        \item For asymmetric distributions, the authors should be careful not to show in tables or figures symmetric error bars that would yield results that are out of range (e.g., negative error rates).
        \item If error bars are reported in tables or plots, the authors should explain in the text how they were calculated and reference the corresponding figures or tables in the text.
    \end{itemize}

\item {\bf Experiments compute resources}
    \item[] Question: For each experiment, does the paper provide sufficient information on the computer resources (type of compute workers, memory, time of execution) needed to reproduce the experiments?
    \item[] Answer: \answerYes{}
    \item[] Justification: All experiments run on CPU on a single Apple-silicon laptop, without GPUs; the supplemental compute report lists the wall-clock time of each core experiment, and App.~\ref{app:exp-extended} lists the size of every sweep in the additional experiments.
    \item[] Guidelines:
    \begin{itemize}
        \item The answer \answerNA{} means that the paper does not include experiments.
        \item The paper should indicate the type of compute workers CPU or GPU, internal cluster, or cloud provider, including relevant memory and storage.
        \item The paper should provide the amount of compute required for each of the individual experimental runs as well as estimate the total compute. 
        \item The paper should disclose whether the full research project required more compute than the experiments reported in the paper (e.g., preliminary or failed experiments that didn't make it into the paper). 
    \end{itemize}
    
\item {\bf Code of ethics}
    \item[] Question: Does the research conducted in the paper conform, in every respect, with the NeurIPS Code of Ethics \url{https://neurips.cc/public/EthicsGuidelines}?
    \item[] Answer: \answerYes{}
    \item[] Justification: The work is theoretical and simulation-based on synthetic mazes, involves no human subjects, private data, or deployment, and conforms with the NeurIPS Code of Ethics.
    \item[] Guidelines:
    \begin{itemize}
        \item The answer \answerNA{} means that the authors have not reviewed the NeurIPS Code of Ethics.
        \item If the authors answer \answerNo, they should explain the special circumstances that require a deviation from the Code of Ethics.
        \item The authors should make sure to preserve anonymity (e.g., if there is a special consideration due to laws or regulations in their jurisdiction).
    \end{itemize}

\item {\bf Broader impacts}
    \item[] Question: Does the paper discuss both potential positive societal impacts and negative societal impacts of the work performed?
    \item[] Answer: \answerNo{}
    \item[] Justification: Sec.~\ref{sec:discussion} discusses scientific implications and testable predictions. The work is foundational neuroscience theory with synthetic simulations and no direct deployment path, so a separate discussion of positive and negative societal impacts is not included.
    \item[] Guidelines:
    \begin{itemize}
        \item The answer \answerNA{} means that there is no societal impact of the work performed.
        \item If the authors answer \answerNA{} or \answerNo, they should explain why their work has no societal impact or why the paper does not address societal impact.
        \item Examples of negative societal impacts include potential malicious or unintended uses (e.g., disinformation, generating fake profiles, surveillance), fairness considerations (e.g., deployment of technologies that could make decisions that unfairly impact specific groups), privacy considerations, and security considerations.
        \item The conference expects that many papers will be foundational research and not tied to particular applications, let alone deployments. However, if there is a direct path to any negative applications, the authors should point it out. For example, it is legitimate to point out that an improvement in the quality of generative models could be used to generate Deepfakes for disinformation. On the other hand, it is not needed to point out that a generic algorithm for optimizing neural networks could enable people to train models that generate Deepfakes faster.
        \item The authors should consider possible harms that could arise when the technology is being used as intended and functioning correctly, harms that could arise when the technology is being used as intended but gives incorrect results, and harms following from (intentional or unintentional) misuse of the technology.
        \item If there are negative societal impacts, the authors could also discuss possible mitigation strategies (e.g., gated release of models, providing defenses in addition to attacks, mechanisms for monitoring misuse, mechanisms to monitor how a system learns from feedback over time, improving the efficiency and accessibility of ML).
    \end{itemize}
    
\item {\bf Safeguards}
    \item[] Question: Does the paper describe safeguards that have been put in place for responsible release of data or models that have a high risk for misuse (e.g., pre-trained language models, image generators, or scraped datasets)?
    \item[] Answer: \answerNA{}
    \item[] Justification: The paper releases simulation code, synthetic maze procedures, result files, and figures; no pretrained models or scraped data with a risk of misuse are released.
    \item[] Guidelines:
    \begin{itemize}
        \item The answer \answerNA{} means that the paper poses no such risks.
        \item Released models that have a high risk for misuse or dual-use should be released with necessary safeguards to allow for controlled use of the model, for example by requiring that users adhere to usage guidelines or restrictions to access the model or implementing safety filters. 
        \item Datasets that have been scraped from the Internet could pose safety risks. The authors should describe how they avoided releasing unsafe images.
        \item We recognize that providing effective safeguards is challenging, and many papers do not require this, but we encourage authors to take this into account and make a best faith effort.
    \end{itemize}

\item {\bf Licenses for existing assets}
    \item[] Question: Are the creators or original owners of assets (e.g., code, data, models), used in the paper, properly credited and are the license and terms of use explicitly mentioned and properly respected?
    \item[] Answer: \answerYes{}
    \item[] Justification: No external datasets, pretrained models, or third-party research code are used. The open-source packages used by the code (NumPy, SciPy, and pandas under the BSD 3-Clause license, Matplotlib under the PSF-based Matplotlib license, and pytest under the MIT license) are named with their licenses in App.~\ref{app:exp-reproduction}.
    \item[] Guidelines:
    \begin{itemize}
        \item The answer \answerNA{} means that the paper does not use existing assets.
        \item The authors should cite the original paper that produced the code package or dataset.
        \item The authors should state which version of the asset is used and, if possible, include a URL.
        \item The name of the license (e.g., CC-BY 4.0) should be included for each asset.
        \item For scraped data from a particular source (e.g., website), the copyright and terms of service of that source should be provided.
        \item If assets are released, the license, copyright information, and terms of use in the package should be provided. For popular datasets, \url{paperswithcode.com/datasets} has curated licenses for some datasets. Their licensing guide can help determine the license of a dataset.
        \item For existing datasets that are re-packaged, both the original license and the license of the derived asset (if it has changed) should be provided.
        \item If this information is not available online, the authors are encouraged to reach out to the asset's creators.
    \end{itemize}

\item {\bf New assets}
    \item[] Question: Are new assets introduced in the paper well documented and is the documentation provided alongside the assets?
    \item[] Answer: \answerYes{}
    \item[] Justification: The released code (core experiments, additional experiments, and figure scripts) is documented by README files in the code release and by Apps.~\ref{app:experiments}--\ref{app:exp-reproduction}, and is distributed under the MIT license.
    \item[] Guidelines:
    \begin{itemize}
        \item The answer \answerNA{} means that the paper does not release new assets.
        \item Researchers should communicate the details of the dataset\slash code\slash model as part of their submissions via structured templates. This includes details about training, license, limitations, etc. 
        \item The paper should discuss whether and how consent was obtained from people whose asset is used.
        \item At submission time, remember to anonymize your assets (if applicable). You can either create an anonymized URL or include an anonymized zip file.
    \end{itemize}

\item {\bf Crowdsourcing and research with human subjects}
    \item[] Question: For crowdsourcing experiments and research with human subjects, does the paper include the full text of instructions given to participants and screenshots, if applicable, as well as details about compensation (if any)? 
    \item[] Answer: \answerNA{}
    \item[] Justification: The paper does not involve crowdsourcing or research with human subjects; all results come from simulations.
    \item[] Guidelines:
    \begin{itemize}
        \item The answer \answerNA{} means that the paper does not involve crowdsourcing nor research with human subjects.
        \item Including this information in the supplemental material is fine, but if the main contribution of the paper involves human subjects, then as much detail as possible should be included in the main paper. 
        \item According to the NeurIPS Code of Ethics, workers involved in data collection, curation, or other labor should be paid at least the minimum wage in the country of the data collector. 
    \end{itemize}

\item {\bf Institutional review board (IRB) approvals or equivalent for research with human subjects}
    \item[] Question: Does the paper describe potential risks incurred by study participants, whether such risks were disclosed to the subjects, and whether Institutional Review Board (IRB) approvals (or an equivalent approval/review based on the requirements of your country or institution) were obtained?
    \item[] Answer: \answerNA{}
    \item[] Justification: The paper does not involve human participants or identifiable private data, so IRB review is not applicable.
    \item[] Guidelines:
    \begin{itemize}
        \item The answer \answerNA{} means that the paper does not involve crowdsourcing nor research with human subjects.
        \item Depending on the country in which research is conducted, IRB approval (or equivalent) may be required for any human subjects research. If you obtained IRB approval, you should clearly state this in the paper. 
        \item We recognize that the procedures for this may vary significantly between institutions and locations, and we expect authors to adhere to the NeurIPS Code of Ethics and the guidelines for their institution. 
        \item For initial submissions, do not include any information that would break anonymity (if applicable), such as the institution conducting the review.
    \end{itemize}

\item {\bf Declaration of LLM usage}
    \item[] Question: Does the paper describe the usage of LLMs if it is an important, original, or non-standard component of the core methods in this research? Note that if the LLM is used only for writing, editing, or formatting purposes and does \emph{not} impact the core methodology, scientific rigor, or originality of the research, declaration is not required.
    \item[] Answer: \answerNA{}
    \item[] Justification: LLMs are not an important, original, or non-standard component of the core scientific method, theory, simulations, or evaluation pipeline. Any use limited to writing, editing, or formatting would not affect the methodology, scientific rigor, or originality and therefore does not require declaration under this checklist item.
    \item[] Guidelines:
    \begin{itemize}
        \item The answer \answerNA{} means that the core method development in this research does not involve LLMs as any important, original, or non-standard components.
        \item Please refer to our LLM policy in the NeurIPS handbook for what should or should not be described.
    \end{itemize}

\end{enumerate}

\end{document}